\documentclass[11pt]{article}
\usepackage{amsmath,amsfonts,amsthm,amssymb,color}
\usepackage{fancybox}
\usepackage{tikz}
\usepackage[procnumbered,ruled,vlined,linesnumbered]{algorithm2e}

\newtheorem{theorem}{Theorem}[section]
\newtheorem{corollary}[theorem]{Corollary}
\newtheorem{lemma}[theorem]{Lemma}

\newtheorem{claim}[theorem]{Claim}
\newtheorem{fact}[theorem]{Fact}

\theoremstyle{definition}
\newtheorem{definition}[theorem]{Definition}

\newenvironment{fminipage}%
{\begin{Sbox}\begin{minipage}}%
		{\end{minipage}\end{Sbox}\fbox{\TheSbox}}

\def\prob#1#2{\mbox{Pr}_{#1}\left[ #2 \right]}

\def\expec#1#2{{\mathbb{E}}_{#1}\left[ #2 \right]}
\def\var#1#2{\mbox{\bf Var}_{#1}\left[ #2 \right]}

\def\defeq{\stackrel{\mathrm{def}}{=}}

\def\pr#1{\left( #1 \right ) }

\def\abs#1{\left|#1  \right|}

\def\norm#1{\left\| #1 \right\|}

\renewcommand\Pr{\boldsymbol{Pr}}
\newcommand\schur{\textsc{Sc}}

\newcommand\cchi{\boldsymbol{\chi}}

\renewcommand\deg{\boldsymbol{\mathit{d}}}

\newcommand\pp{\boldsymbol{\mathit{p}}}

\newcommand\ww{\boldsymbol{\mathit{w}}}
\newcommand\yy{\boldsymbol{\mathit{y}}}
\newcommand\zz{\boldsymbol{\mathit{z}}}
\newcommand\xx{\boldsymbol{\mathit{x}}}

\renewcommand\AA{\boldsymbol{\mathit{A}}}

\newcommand\DD{\boldsymbol{\mathit{D}}}

\newcommand\II{\boldsymbol{\mathit{I}}}

\newcommand\MM{\boldsymbol{\mathit{M}}}
\newcommand\LL{\boldsymbol{\mathit{L}}}

\newcommand\Otil{\widetilde{O}}

\newcommand\xhat{\widehat{\mathit{x}}}

\newcommand\supp{\text{supp}}

\newcommand\ttau{\boldsymbol{\mathit{\tau}}}
\newcommand\ttaubar{\boldsymbol{\overline{\mathit{\tau}}}}
\newcommand\ttautil{\boldsymbol{\widetilde{\mathit{\tau}}}}
\newcommand\er{\mathcal{R}_{eff}}
\newcommand\ertil{\widetilde{\mathcal{R}}_{eff}}

\newcommand\amount{s}

\newcommand{\levscore}[1]{\ttaubar_{#1}}
\newcommand{\levscoreapprox}[1]{\ttautil_{#1}}
\newcommand{\levscoreexact}[1]{\ttaubar_{#1}}
\newcommand{\totaltrees}[1]{\mathcal{T}_{#1}}

\newcommand{\hatT}{\widehat{T}}
\newcommand{\hatF}{\widehat{F}}

\newcommand{\detp}{{\det}_{+}}

\renewcommand{\sc}[2]{\textsc{Sc}\left(#1, #2\right)}

\usepackage{hyperref}

\renewcommand\ss{\boldsymbol{\mathit{s}}}

\usetikzlibrary{arrows, decorations.markings}

\tikzstyle{smallvertex}=[circle,fill=black!25,minimum size=6,inner sep=0pt,color=blue]
\tikzstyle{vertex}=[circle,fill=black!25,minimum size=10,inner sep=0pt,color=blue]
\tikzstyle{treeedge} = [draw,line width = 1.5mm, ,-,color=red]
\tikzstyle{extraedge} = [draw,line width = 0.9mm, ,-,color=green]
\tikzstyle{edge} = [draw,line width = 0.25mm,-,color=cyan]
\tikzstyle{heavyedge} = [draw,line width = 1.00mm,-,color=cyan]
\tikzstyle{vecTinyArrow} = [thick, decoration={markings,mark=at position
	1 with {\arrow[scale=1]{open triangle 90}}},
double distance=2, shorten >= 5,
preaction = {decorate},
postaction = {draw,line width=2, white,shorten >= 8}]
\tikzstyle{vecNarrowArrow} = [thick, decoration={markings,mark=at position
	1 with {\arrow[scale=2]{open triangle 90}}},
double distance=5, shorten >= 10,
preaction = {decorate},
postaction = {draw,line width=5, white,shorten >= 8}]
\tikzstyle{vecArrow} = [thick, decoration={markings,mark=at position
	1 with {\arrow[scale=3]{open triangle 90}}},
double distance=8, shorten >= 14,
preaction = {decorate},
postaction = {draw,line width=8, white,shorten >= 12}]

\makeatletter
\newcommand{\nosemic}{\renewcommand{\@endalgocfline}{\relax}}
\newcommand{\dosemic}{\renewcommand{\@endalgocfline}{\algocf@endline}}
\let\oldnl\nl
\newcommand{\nonl}{\renewcommand{\nl}{\let\nl\oldnl}}
\makeatother

\begin{document}
	
\title{Determinant-Preserving Sparsification of SDDM
Matrices with Applications to Counting and Sampling Spanning Trees}
	
\author{
David Durfee\thanks{Georgia Institute of Technology. \texttt{email:ddurfee@gatech.edu}}
\and
John Peebles\thanks{Massachusetts Institute of Technology. \texttt{email:jpeebles@mit.edu}}
\and
Richard Peng\thanks{Georgia Institute of Technology.
\texttt{email:rpeng@cc.gatech.edu}}
\and
Anup B. Rao\thanks{Georgia Institute of Technology. \texttt{email:anup.rao@gatech.edu}}
}

\newcommand{\one}{\mathbf{1}}

	\maketitle
	\begin{abstract}
		We show variants of spectral sparsification routines can preserve the total
spanning tree counts of graphs, which by Kirchhoff's matrix-tree theorem, is
equivalent to determinant of a graph Laplacian minor, or equivalently, of any SDDM matrix.
Our analyses utilizes this combinatorial connection to bridge between statistical
leverage scores / effective resistances and the analysis of random graphs
by [Janson, Combinatorics,~Probability~and~Computing~`94].
This leads to a routine that in quadratic time, sparsifies a graph down to about
$n^{1.5}$ edges in ways that preserve both the determinant and the distribution
of spanning trees (provided the sparsified graph is viewed as a random object).
Extending this algorithm to work with Schur complements and approximate
Choleksy factorizations leads to algorithms for counting and
sampling spanning trees which are nearly optimal for dense graphs.

We give an algorithm that computes a $(1 \pm \delta)$ approximation to the determinant
of any SDDM matrix with constant probability in about $n^2 \delta^{-2}$ time.
This is the first routine for graphs that outperforms general-purpose routines for computing
determinants of arbitrary matrices. We also give an algorithm that generates in about $n^2 \delta^{-2}$ time a spanning tree of
a weighted undirected graph from a distribution with total variation
distance of $\delta$ from the $\ww$-uniform distribution .

	\end{abstract}
	
	\section{Introduction}
\label{sec:introduction}

The determinant of a matrix is a fundamental quantity in numerical
algorithms due to its connection to the rank of the matrix and
its interpretation as the volume of the ellipsoid corresponding
of the matrix.
For graph Laplacians, which are at the core of spectral graph
theory and spectral algorithms, the matrix-tree theorem gives that
the determinant of the minor obtained by removing one row and the corresponding column equals
to the total weight of all the spanning trees in the graph \cite{Kirchhoff47} .
Formally on a weighted graph $G$ with $n$ vertices we have:
\[
\det \pr{ \LL^{G}_{1:n-1, 1:n-1}} = \totaltrees{G},
\]
where $\LL^{G}$ is the graph Laplacian of $G$ and
and $\totaltrees{G}$ is the total weight of all the spanning trees of $G$.
As the all-ones vector is in the null space of $\LL^{G}$,
we need to drop its last row and column and work with
$\LL^{G}_{1:n-1, 1:n-1}$, which is precisely the definition
of SDDM matrices in numerical analysis~\cite{SpielmanTengSolver:journal}.
The study of random spanning trees builds directly upon this
connection between tree counts and determinants, and also plays
an important role in graph theory~\cite{GoyalRV09,AsadpourGMGS10,FungHHP11}.

While there has been much progress in the development of faster
spectral algorithms, the estimation of determinants encapsulates
many shortcomings of existing techniques.
Many of the nearly linear time algorithms rely on sparsification
procedures that remove edges from a graph while provably
preserving the Laplacian matrix as an operator,
and in turn, crucial algorithmic quantities such as
cut sizes, Rayleigh quotients, and eigenvalues.
The determinant of a matrix on the other hand is the product of
all of its eigenvalues.
As a result, a worst case guarantee of $1 \pm (\epsilon / n)$ per
eigenvalue is needed to obtain a good overall approximation, and this
in turn leads to additional factors of $n$ in the number of edges
needed in the sparse approximate.

Due to this amplification of error by a factor of $n$,
previous works on numerically approximating determinants
without dense-matrix multiplications~\cite{BoutsidisDKZ15:arxiv,
HunterAB14:arxiv,HanMS15} usually focus on the log-determinant,
and (under a nearly-linear running time) give
errors of additive $\epsilon n$ in the log determinant estimate,
or a multiplicative error of $\exp(\epsilon n)$ for the determinant.
The lack of a sparsification procedure also led to the running
time of random spanning tree sampling algorithms to be limited
by the sizes of the dense graphs generated in intermediate
steps~\cite{KelnerM09,MadryST15,DurfeeKPRS16}.

In this paper, we show that a slight variant of spectral sparsification
preserves determinant approximations to a much higher accuracy
than applying the guarantees to individual edges.
Specifically, we show that sampling $\omega(n^{1.5})$ edges 
from a distribution given by leverage scores,
or weight times effective resistances, produces a sparser graph whose
determinant approximates that of the original graph.
Furthermore, by treating the sparsifier itself as a random object,
we can show that the spanning tree distribution produced by
sampling a random tree from a random sparsifier is close to the
spanning tree distribution in the original graph in total variation
distance.
Combining extensions of these algorithms with sparsification based
algorithms for graph Laplacians then leads to quadratic time algorithms
for counting and sampling random spanning trees, which are
nearly optimal for dense graphs with $m = \Theta(n^2)$.

This determinant-preserving sparsification phenomenon is surprising
in several aspects: because we can also show---both experimentally and mathematically---that on the complete graph, about $n^{1.5}$ edges are necessary
to preserve the determinant, this is one of the first graph sparsification
phenomenons that requires the number of edges to be between $>>n$.
The proof of correctness of this procedure also hinges upon combinatorial
arguments based on the matrix-tree theorem in ways motivated by
a result for Janson for complete graphs~\cite{Janson94},
instead of the more common matrix-concentration bound based
proofs~\cite{SpielmanS11,Tropp12,CohenP15,Cohen16}.
Furthermore, this algorithm appears far more delicate than spectral
sparsification: it requires global control on the number of samples,
high quality estimates of resistances (which is the running time
bottleneck in Theorem~\ref{thm:SparsifyMain} below), and only holds
with constant probability.
Nonetheless, the use of this procedure into our
determinant estimation and spanning tree generation algorithms still demonstrates
that it can serve as a useful algorithmic tool.

\subsection{Our Results}

We will use $G = (V, E, \ww)$ to denote weighted
multigraphs, and $\deg_u \defeq \sum_{e:e \ni u} \ww_{e}$
to denote the weighted degree of vertex $u$.
The weight of a spanning tree in a weighed undirected
multigraph is:
\[
\ww\left(T \right) \defeq \prod_{e \in T} \ww_e.
\]

We will use $\totaltrees{G}$ to denote the total
weight of trees, $\totaltrees{G} \defeq
\sum_{T \in \mathcal{T}} \ww( T )$. Our key sparsification result can be described by the following
theorem:
\begin{theorem}
	\label{thm:SparsifyDeterminant}
	Given any graph $G$ and any parameter $\delta$, we can compute
	in $O(n^2 \delta^{-2})$ time a graph $H$ with $O(n^{1.5} \delta^{-2})$
	edges such that with constant probability we have
	\[
	\left( 1 - \delta \right) \totaltrees{G}
	\leq \totaltrees{H} \leq
	\left( 1 + \delta \right) \totaltrees{G}.
	\]
\end{theorem}
This implies that graphs can be sparsified in a manner that preserves the determinant, albeit to a density that is not nearly-linear in $n.$

We show how to make our sparsification 
routine to errors in estimating leverage scores, and how our scheme can
be adapted to implicitly sparsify dense objects that we do not have explicit
access to.
In particular, we utilize tools such as rejection sampling and
high quality effective resistance estimation via projections to
extend this routine to give determinant-preserving sparsification algorithms for Schur complements, which are intermediate states
of Gaussian elimination on graphs,
using ideas from the sparsification of random walk polynomials.

We use these extensions of our routine to obtain a variety
of algorithms built around our graph sparsifiers. Our two main algorithmic applications are as follows.
We achieve the first algorithm for estimating the determinant of an SDDM matrix that is faster than general purpose algorithms for the matrix determinant problem. Since the determinant of an SDDM m corresponds to the determinant of a graph Laplacian with one row/column
removed.

\begin{theorem}\label{thm:detApproxIntro}
Given an SDDM matrix $\MM$, there is a routine $\textsc{DetApprox}$ which in $\widetilde{O}\pr{ n^2\delta^{-2} }$ time  outputs $D$ such that $D = \pr{ 1 \pm {{\delta}}}\det(\MM)$ with high probability
\end{theorem} 

A crucial thing to note which distinguishes the above guarantee from most other similar results is that we give a multiplicative approximation of the $\det(M)$. This is much stronger than giving a multiplicative approximation of $\log \det (M),$ which is what other work typically tries to achieve.

The sparsifiers we construct will also approximately preserve the spanning tree distribution, which we leverage to yield a faster algorithm for sampling random spanning trees. Our new algorithm improves upon the current fastest algorithm for general weighted graphs when one wishes to achieve constant---or slightly sub-constant---total variation distance.


\begin{theorem}\label{thm:spanningTreeAlgoIntro}
	Given an undirected, weighted graph $G=(V,E,\ww)$, there is a routine $\textsc{ApproxTree}$ which in expected time $\widetilde{O}\pr{n^2\delta^{-2}}$ outputs a random spanning tree from a distribution that has total variation distance $\leq \delta$ from the $\ww$-uniform distribution on $G$.
\end{theorem}

 
\subsection{Prior Work}



\subsubsection{Graph Sparsification}

In the most general sense, a graph sparsification procedure is a method for taking a potentially dense graph and returning a sparse graph called a \emph{sparsifier} that approximately still has many of the same properties of the original graph. It was introduced in \cite{EGIN97} for preserving properties related to minimum spanning trees, edge connectivity, and related problems. \cite{BK96} defined the notion of \emph{cut sparsification} in which one produces a graph whose cut sizes approximate those in the original graph. \cite{SpielmanTengSparsification:journal} defined the more general notion of \emph{spectral sparsification} which requires that the two graphs' Laplacian matrices approximate each other as quadratic forms.\footnote{If two graphs Laplacian matrices approximate each other as quadratic forms then their cut sizes also approximate each other.} In particular, this spectral sparsification samples $\widetilde{O}(n/\epsilon^2)$ edges from the original graph, yielding a graph with $\widetilde{O}(n/\epsilon^2)$ whose quadratic forms---and hence, eigenvalues---approximate each other within a factor of $(1 \pm \epsilon)$. This implies that their determinants approximate each other within $(1 \pm \epsilon)^{n}$. This is not useful from the perspective of preserving the determinant: since one would need to samples $\Omega(n^3)$ edges to get a constant factor approximation, one could instead exactly compute the determinant or sample spanning trees using exact algorithms with this runtime.

All of the above results on sparsification are for undirected graphs. Recently, \cite{CKPPRSV17} has defined a useful notion of sparsification for directed graphs along with a nearly linear time algorithm for constructing sparsifiers under this notion of sparsification.

\subsubsection{Determinant Estimation}

Exactly calculating the the determinant of an arbitrary matrix is known to be equivalent to matrix multiplication \cite{BS83}. For approximately computing the log of the determinant, \cite{ID11} uses the identity $\log(\det(A)) = \text{tr}(\log(B)) + \text{tr}(\log(B^{-1} A))$ to do this whenever one can find a matrix $B$ such that the $\text{tr}(\log(B)) = \log(\det(B))$ and $\text{tr}(\log(B^{-1} A))=\log(\det(B^{-1} A)$ can both be quickly approximated.\footnote{Specifically, they take $B$ as the diagonal of $A$ and prove sufficient conditions for when the log determinant of $B^{-1}A$ can be quickly approximated with this choice of $B$.}

For the special case of approximating the log determinant of an SDD matrix, \cite{HAB14} applies this same identity recursively where the $B$ matrices are a sequence of ultrasparsifiers that are inspired by the recursive preconditioning framework of \cite{SpielmanTengSolver:journal}. They obtain a running time of $O(m (n^{-1}\epsilon^{-2} + \epsilon^{-1}) \text{polylog}(n \kappa / \epsilon) )$ for estimating the log determinant to additive error $\epsilon$.

\cite{BDKZ15} estimates the log determinant of arbitrary positive definite matrices, but has runtime that depends linearly on the condition number of the matrix.

In contrast, our work is the first we know of that gives a multiplicative approximation of the determinant itself, rather than its log. Despite achieving a much stronger approximation guarantee, our algorithm has essentially the same runtime as that of \cite{HAB14} when the graph is dense. Note also that if one wishes to conduct an ``apples to apples'' comparison by setting their value of $\epsilon$ small enough in order to match our approximation guarantee, their algorithm would only achieve a runtime bound of $O(m n \delta^{-2} \text{polylog}(n \kappa / \epsilon))$, which is never better than our runtime and can be as bad as a factor of $n$ worse.\footnote{This simplification of their runtime is using the substitution $\epsilon = \delta/n$ which gives roughly $(1 \pm \delta)$ multiplicative error in estimating the determinant for their algorithm. This simplification is also assuming $\delta \leq 1$, which is the only regime we analyze our algorithm in and thus the only regime in which we can compare the two.}

\subsubsection{Sampling Spanning Trees}

Previous works on sampling random spanning trees are a combination
of two ideas: that they could be generated using random walks,
and that they could be mapped from a random integer via Kirchoff's
matrix tree theorem.
The former leads to running times of the form $O(nm)$
\cite{Broder89, Aldous90}, while the latter
approach\cite{Guenoche83, Kulkarni90,ColbournMN96,HarveyX16}
led to routines that run in $O(n^{\omega})$ time,
where  $\omega \approx 2.373$ is
the matrix multiplication exponent~\cite{Williams12}.

These approaches have been combined in algorithms by
Kelner and Madry~\cite{KelnerM09}
and Madry, Straszak and Tarnawski~\cite{MadryST15}.
These algorithms are based on simulating the walk more efficiently
on parts of the graphs, and combining this with graph decompositions
to handle the more expensive portions of the walks globally.
Due to the connection with random-walk based spanning tree
sampling algorithms, these routines often have inherent
dependencies on the edge weights.
Furthermore, on dense graphs their running times are still worse
than the matrix-multiplication time routines.

The previous best running time for generating a random spanning tree
from a weighted graph was $\Otil \pr{n^{5/3} m^{1/3} \log^2\pr{1/\delta}}$ achieved by \cite{DurfeeKPRS16}.
It works by combining a recursive procedure similar to those used in
the more recent $O(n^\omega)$ time algorithms~\cite{HarveyX16} with spectral sparsification ideas, achieving a runtime of $\widetilde{O}(n^{5/3} m^{13})$.
When $m = \Theta \pr{n^2},$ the algorithm in \cite{DurfeeKPRS16} takes $\Otil \pr{n^{7/3}}$ time to produce a tree from a distribution that is $o(1)$ away from the $\ww$-uniform distribution, which is slower
by nearly a $n^{1/3}$ factor than the algorithm given in this paper.

Our algorithm can be viewed as a natural extension of the
sparsification0-based approach from~\cite{DurfeeKPRS16}: instead of
preserving the probability of a single edge being chosen in a
random spanning tree, we instead aim to preserve the entire distribution
over spanning trees, with the sparsifier itself also considered as
a random variable.
This allow us to significantly reduce the sizes of intermediate
graphs, but at the cost of a higher total variation distance in
the spanning tree distributions.
This characterization of a random spanning tree is not present
in any of the previous works, and we believe it is an interesting
direction to combine our sparsification procedure with the other
algorithms.

\subsection{Organization}

Section~\ref{sec:background} will introduce the necessary notation and some of the previously known fundamental results regarding the mathematical objects that we work with throughout the paper. Section~\ref{sec:overview} will give a high-level sketch of our primary results and concentration bounds for total tree weight under specific sampling schemes. Section~\ref{sec:sparsification} leverages these concentration bounds to give a quadratic time sparsification procedure (down to $\Omega(n^{1.5})$ edges) for general graphs. Section~\ref{sec:ImplicitSchur} uses random walk connections to extend our sparsification procedure to the Schur complement of a graph. Section~\ref{sec:determinant_algo} utilizes the previous routines to achieve a quadratic time algorithm for computing the determinant of SDDM matrices. Section~\ref{sec:spanning_tree} combines our results and modifies previously known routines to give a quadratic time algorithm for sampling random spanning trees with low total variation distance. Section~\ref{sec:cond_conc} extends our concentration bounds to random samplings where an arbitrary tree is fixed, and is necessary for the error accounting of our random spanning tree sampling algorithm. Section~\ref{sec:TVBound} proves the total variation distance bounds given for our random sampling tree algorithm.

\section{Background}\label{sec:background}

\subsection{Graphs, Matrices, and Random Spanning Trees}
\label{subsec:GraphMatricesTrees}
The goal of generating a random spanning tree
is to pick tree $T$ with probability proportional to its weight,
which we formalize in the following definition.

\begin{definition}[$\ww$-uniform distribution on trees]
Let $\Pr_{T}^{G}(\cdot)$ be a probability distribution on
$\mathcal{T}_G$ such that
\[
\Pr_{T}^{G} \left( T = T_0 \right)
= \frac{\Pi_{e \in T_0} \ww_e}{\mathcal{T}_G}.
\]
We refer to $\Pr_{T}^{G}(\cdot)$ as the $\ww$-uniform distribution
on the trees of $G$.

When the graph $G$ is unweighted, this corresponds to the uniform
distribution on $\mathcal{T}_G.$
\end{definition}
We refer to $\Pr_{T}^{G}(\cdot)$ as the $\ww$-uniform distribution
on $\mathcal{T}_G$.
When the graph $G$ is unweighted, this corresponds to the uniform
distribution on $\mathcal{T}_G$.
Furthermore, as we will manipulate the probability of a particular
tree being chosen extensively, we will denote such probabilities
with $\Pr^{G}(\hatT)$, aka:
\[
\Pr^{G}\left(\hatT\right)
\defeq \Pr^{G}_{T} \left( T = \hatT \right).
\]

The Laplacian of a graph $G = (V, E, \ww)$ is an
$n \times n$ matrix specified by:
\[
\LL_{uv} \defeq
\begin{cases}
\deg_u & \text{if $u = v$}\\
-\ww_{uv} & \text{if $u \neq v$}
\end{cases}
\]
We will write $\LL^G$ when we wish to indicate which graph $G$ that the Laplacian corresponds to and $\LL$ when the context is clear.
When the graph has multi-edges, we define $\ww_{uv}$ as the sum of weights of all the edges $e$ that 
go between vertices $u,v.$ Laplacians are natural objects to consider when
dealing with random spanning trees due to the matrix
tree theorem, which states that
the determinant of $\LL$ with any row/column corresponding
to some vertex removed is the total weight of spanning trees.
We denote this removal of a vertex $u$ as $\LL_{-u}$. As the index of vertex removed does not affect the
result, we will usually work with $\LL_{-n}$.
Furthermore, we will use $\det{(\MM)}$ to denote the
determinant of a matrix.
As we will work mostly with graph Laplacians,
it is also useful for us to define the `positive determinant'
$\detp$, where we remove the last row and column.
Using this notation, the matrix tree theorem can be
stated as:\[
\totaltrees{G} = \det(\LL^{G}_{-n}) = \detp\left( \LL^{G} \right).
\]
We measure the distance between two probability distributions by total variation
distance.
\begin{definition}
\label{def:TV}
Given two probability distributions $p$ and $q$
on the same index set $\Omega$, the \textit{total variation distance}
between $p$ and $q$ is given by
\[
d_{TV}\left( p, q \right)
\defeq \frac{1}{2} \sum_{x \in \Omega} \abs{p(x) - q(x)}.
\]
\end{definition}
Let $G=(V,E, \ww)$ be a graph and  $e \in E$ an edge. We write $G/e$ to denote the graph obtained by contracting the edge $e$, i.e., identifying the two endpoints of $e$ and deleting any self loops formed in the resulting graph. We write $G\backslash e$ to denote the graph obtained by deleting the edge $e$ from $G$. We extend these definitions to $G/F$ and $G\backslash F$ for $F \subseteq E$ to refer to the graph obtained by contracting all the edges in $F$ and deleting all the edges in $F$, respectively. 

Also, for a subset of vertices $V_1$, we use $G[V_1]$ to denote
the graph induced on the vertex of $V_1$.
letting $G(V_1)$ be the edges associated  with $\LL_{[V_1,V_1]}$ in the Schur complement.

\subsection{Effective Resistances and Leverage Scores}

The matrix tree theorem also gives connections to another
important algebraic quantity: the \textit{effective resistance}
between two vertices.
This quantity is formally given as $\er(u, v)
\defeq \cchi_{uv}^\intercal \LL^{-1} \cchi_{uv}$ where $\cchi_{uv}$
is the indicator vector with $1$ at $u$, $-1$ at $v$,
and $0$ everywhere else.
Via the adjugate matrix, it can be shown
that the effective resistance of an edge is precisely
the ratio of the number of spanning trees in $G / e$
over the number in $G$:
\[
\er(u, v)
= \frac{\totaltrees{G / e}}{\totaltrees{G}}. 
\]
As $\ww_e \cdot \totaltrees{G / e}$ is the total weight
of all trees in $G$ that contain edge $e$, the
fraction\footnote{provided one thinks of an edge with weight $w$ as representing $w$ parallel edges, or equivalently, counts spanning trees with multiplicity according to their weight} of spanning trees that contain $e = uv$
is given by $\ww_e \er(u, v)$. This quantity
is called the \textit{statistical leverage score} of an edge,
and we denote it by $\ttaubar_{e}$.
It is fundamental component of many randomized algorithms for sampling / sparsifying
graphs and matrices~\cite{SpielmanS11, Vishnoi12, Tropp12}.

The fact that $\ttaubar_e$ is the fraction of trees
containing $e$ also gives one way of deriving the sum of these quantities:
\begin{fact}
\label{fact:foster}
(Foster's Theorem)
On any graph $G$ we have
\[
\sum_{e} \levscoreexact{e} = n - 1.
\]
\end{fact}

The resistance $\er(u, v)$, and in turn the statistical
leverage scores $\ttaubar_e$ can be estimated
using linear system solves and random
projections~\cite{SpielmanS11}.
For simplicity, we follow the abstraction utilized by
Madry, Straszak, and Tarnawski~\cite{MadryST15}, except we also allow the intermediate linear
system solves to utilize a sparsifier instead of the original graph.

\begin{lemma}
\label{lem:ERDS}
(Theorem 2.1. of~\cite{MadryST15})

Let $G = (V, E)$ be a graph with m edges.
For every $\epsilon > 0$ we can find in
$\tilde{O}(\min\{m \epsilon^{-2}, m + n\epsilon^{-4} \})$ time
an embedding of the effective resistance metric into
$\Re^{O(\epsilon^{-2 \log{m}})}$
such that with high probability	allows one to compute an estimate $\ertil(u,v)$ of any effective resistance satisfying
\[
\forall u, v \in V
\qquad
\left(1 - \epsilon\right) \ertil\left(u, v\right)
\leq \er \left(u, v\right) \leq
\left(1 + \epsilon\right) \ertil \left(u, v\right).
\]
Specifically, each vertex $u$ in this embedding is associated with
an (explicitly stored) $\zz_u \in \Re^{O(\epsilon^{-2 \log{m}})}$,
and for any pair of vertices, the estimate $\ertil(u, v)$ is given by:
\[
\ertil\left(u, v\right)
= \norm{\zz_u - \zz_v}_2^2,
\]
which takes $O(\epsilon^{-2} \log{m})$ time to compute once we have the embedding.
\end{lemma}
 
\subsection{Schur Complements}

For our applications, we will utilize our determinant-preserving sparsification
algorithms in recursions based on Schur complements.
A partition of the vertices, which we will denote using
\[
V = V_1 \sqcup V_2,
\]
partitions the corresponding graph Laplacian into
blocks which we will denote using indices in the subscripts:
\[
\LL = \left[
 \begin{array}{cc}
\LL_{[V_1, V_1]} & \LL_{[V_1, V_2]}\\
\LL_{[V_2, V_1]} & \LL_{[V_2, V_2]}
\end{array}
\right].
\]
The Schur complement of $G$, or $\LL$, onto $V_1$ is then:
\[
\sc{G}{V_1}
= \sc{\LL^{G}}{V_1}
\defeq  \LL^{G}_{[V_1, V_1]}
-  \LL^{G}_{[V_1, V_2]} \left( \LL^{G}_{[V_2, V_2]} \right)^{-1}
	\LL^{G}_{[V_2, V_1]},
\]
and we will use $\sc{G}{V_1}$ and $\sc{\LL^{G}}{V_1}$
interchangeably. We further note that we will always consider $V_1$ to be the vertex set we Schur complement onto, and $V_2$ to be the vertex set we eliminate, except for instances in which we need to consider both $\sc{G}{V_1}$ and $\sc{G}{V_2}$.

Schur complements behave nicely with respect to determinants determinants, which suggests the general structure of the recursion we will use for estimating the determinant.
\begin{fact}\label{fact:detMinor}
For any matrix $\MM$ where $\MM_{[V_2, V_2]}$ is invertible,
\[
\det{(\MM_{-n})}
= \det{\left(\MM_{\left[V_2, V_2\right]}\right)}
\cdot \detp{\left(\sc{\MM}{V_1}\right)}.
\]
\end{fact}
This relationship also suggests that
there should exist a bijection between spanning tree distribution in $G$ and the product distribution given by sampling
spanning trees independently from
$\sc{\LL}{V_1}$ and the graph Laplacian formed by adding one
row/column to $\LL_{[V_2, V_2]}$.

Finally, our algorithms for approximating Schur
complements rely on the fact that they preserve certain marginal probabilities. The algorithms of ~\cite{ColbournDN89,ColbournMN96,HarveyX16,DurfeeKPRS16} also use variants of some of these facts, which are closely related to the preservation of the spanning tree distribution on $\sc{\LL}{V_1}$. (See 
Section~\ref{sec:spanning_tree} for details.)
\begin{fact}
	\label{fact:SchurResistance}
	Let $V_1$ be a subset of vertices of a graph $G$, then for any vertices
	$u, v \in V_1$, we have:
	\[
	\er^{G}\left(u, v\right)
	= \er^{\sc{G}{V_1}}\left(u, v\right).
	\]
\end{fact}

\begin{theorem}[Burton and Premantle \cite{burton93}]
For any set of edges $F \subseteq E$ in a graph $G = (V,E, \ww),$ the probability $F$ is contained in a $\ww$-uniform random spanning tree is
\[\Pr_{T}^{G}(F \subseteq T)  = \det(\MM_{(\LL,F)}) ,\]
where $\MM_{(\LL,F)}$ is a $|F| \times |F|$ matrix whose $(e,f)$'th entry, for $e, f \in F,$ is given by $\sqrt{\ww(e) \ww(f)} \chi_e^T \LL^{\dag} \chi_f.$ 
\end{theorem}
By a standard property of Schur complements (see \cite{Horn12}), we have
\[
\left(\LL^{-1}\right)[V_1,V_1] = \sc{G}{V_1}^{\dag}.
\]
Here $(\LL^{\dag})[V_1,V_1]$ is the minor of $\LL^{\dag}$ with row and column indices in $V_1.$ This immediately implies that when $F $ is incident only on vertices in $V_1,$  we have  $\MM_{(\LL,F)} = \MM_{(\sc{G}{V_1},F)}.$ Putting these together, we have
\begin{fact}
	\label{fact:schurPreservesProb}
	Given a partition of the vertices $V = V_1 \sqcup V_2$. For any set of edges $F$  contained in $G[V_1]$,  we have
	\[\Pr_{T}^{G}(F \subseteq T) = \Pr_{T }^{\sc{G}{V_1}}(F \subseteq T).\]
\end{fact}


\section{Sketch of the Results}
\label{sec:overview}

The starting point for us is the paper by Janson \cite{Janson94} which
gives (among other things) the limiting distribution of the number of
spanning trees in the $\mathcal{G}_{n,m}$ model of random graphs.
Our concentration result for the number of spanning trees in the sparsified
graph is inspired by this paper, and our algorithmic use of this sparsification
routine is motivated by sparsification based algorithms for matrices
related to graphs~\cite{PengS14,ChengCLPT15,KyngLPSS16}.
The key result we will prove is a concentration bound
on the number of spanning trees when the graph is sparsified by sampling edges with probability approximately proportional to effective resistance.

\subsection{Concentration Bound}

Let $G$ be a weighted graph with $n$ vertices and $m$ edges,
and $H$ be a random subgraph obtained by choosing a subset of edges
of size $\amount$ uniformly randomly.
The probability of a subset of edges, which could either
be a single tree, or the union of several trees, being
kept in $H$ can be bounded precisely.
Since we will eventually choose $\amount > n^{1.5},$ we will treat the quantity $n^3 / \amount^2$ as negligible.
The probability of $H$ containing a fixed tree was shown
by Janson to be:
\begin{lemma}
\label{lem:subsetSampled}
If $m \geq \frac{\amount^2}{n}$, then
for any tree $T$, the probability of it
being included in $H$ is
\[
	\prob{H}{T\in H}
	=
	\frac{(\amount)_{n-1}}{(m)_{n-1}}
	=  p^{n-1}
	\cdot \exp\left(-\frac{n^2}{2\amount} - O\left(\frac{n^3}{\amount^2}\right) \right).
\]
where $(a)_b$ denotes the product
$a \cdot (a - 1) \cdots (a - (b-1))$.
\end{lemma}
By linearity of expectation, the expected total weight of spanning trees in $H$ is:
\begin{equation}
\expec{H}{\totaltrees{H}}
= \totaltrees{G}
\cdot p^{n-1}
\cdot \exp\left(-\frac{n^2}{2 \amount}
- O\left(\frac{n^3}{\amount^2}\right) \right).
\label{eqn:Expectation}
\end{equation}
As in \cite{Janson94}, the second moment,  $\expec{H}{\totaltrees{H}^2} = \expec{H}{ \sum_{(T_1,T_2)} \ww(T_1) \ww(T_2) \Pr \pr{T_1,T_2 \in H}}$, can be written as a sum over all pairs
of trees $\pr{T_1, T_2}.$ 
Due to symmetry, the probability of
a particular pair of trees $T_1, T_2$ both being subgraphs of $H$
depends only on the size of their intersection.
The following bound is shown in Appendix~\ref{sec:deferred}.
\begin{lemma}
\label{lem:PairwiseProb}
Let $G$ be a graph with $n$ vertices and $m$ edges, and
$H$ be a uniformly random subset of $s > 10 n$ edges chosen from $G$, where $m \geq \frac{\amount^2}{n}$.
Then for any two spanning trees $T_1$ and $T_2$ of $G$ with
$ \left| T_1 \cap T_2 \right |= k,$
we have:
\[
\prob{H}{T_1,T_2\in H}
\leq 
p^{2n-2} \exp\left(-\frac{2n^2}{\amount}\right)
\left( \frac{1}{p}\left(1 + \frac{2n}{\amount}\right)\right)^{k},
\]
where $p = s / m.$ 
\end{lemma}
The crux of the bound on the second moment in Janson's
proof is getting a handle on the number of tree pairs $\pr{T_1,T_2}$ with  $\left | T_1 \cap T_2 \right |= k$ in the complete graph where all edges are symmetric. 
An alternate way to obtain a bound on the number
of spanning trees can also be obtained using leverage
scores, which describe the fraction of spanning
trees that utilize a single edge. A well known fact about random spanning tree distributions \cite{burton93} is that the edges are 
negatively correlated:
\begin{fact}[Negative Correlation]
\label{fact:negCorrelation}
Suppose $F$ is subset of edges in a graph $G$, then
$$\Pr_T^{G} \pr{F \subseteq T} \leq \Pi_{e \in F} \Pr_T^{G} \pr{e \in T}.$$
\end{fact}
An easy consequence of Fact~\ref{fact:negCorrelation} is
\begin{lemma}
\label{lem:SubsetTree}
For any subset of edges $F$ we have that the total weight of all spanning trees containing $F$ is given by
\[
\sum_{\substack{\text{$T$ is a spanning tree of $G$} \\
F \subseteq T}} \ww\left(T \right)
\leq \totaltrees{G} \prod_{e \in F} \ttaubar_e.
\]
\end{lemma}

The combinatorial
view of all edges being interchangable in the complete graph
can therefore be replaced with an algebraic view in terms of the leverage scores.
Specifically, invoking Lemma~\ref{lem:SubsetTree}
in the case where all edges have leverage score at most
$\frac{n}{m}$ gives the following lemma  which is proven in Appendix~\ref{sec:deferred}.

\begin{lemma}
\label{lem:IntersectionPairs}
In a graph $G$ where all edges have leverage scores at most
$\frac{n}{m}$,
we have
\[
\sum_{\substack{T_1, T_2 \\ \abs{T_1 \cap T_2} = k}}
\ww\left( T_1 \right) \cdot \ww\left( T_2 \right)
\leq
\totaltrees{G}^2 \cdot \frac{1}{k!} \left( \frac{n^2}{m} \right)^{k}
\]
\end{lemma}

With Lemma \ref{lem:IntersectionPairs}, we can finally prove the following bound on the second moment which gives our concentration result.

\begin{lemma}
\label{lem:SecondMoment}
Let $G$ be a graph on $n$ vertices and $m$ edges such that all edges
have statistical leverage scores $\leq \frac{n}{m}$.
For a random subset of $s > 10 n$ edges, $H$, where $m \geq \frac{\amount^2}{n}$ we have:
\[
\expec{H}{\totaltrees{H}^2}
\leq \totaltrees{G}^2 p^{2n-2}
\exp\left(-\frac{n^2}{\amount} + O\left(\frac{n^3}{\amount^2} \right) \right)
= \expec{H}{\totaltrees{H}}^2 \exp\left( O\left(\frac{n^3}{\amount^2}\right) \right).
\]
\end{lemma}

\begin{proof}
By definition of the second moment, we have:
\[
\expec{H}{\totaltrees{H}^2}
= \sum_{T_1, T_2}
\ww\left( T_1 \right) \cdot \ww\left(T_2 \right)
\cdot \prob{H}{T_1 \cup T_2 \subseteq H}.
\]
Re-writing the above sum in terms of the size of the intersection $k$, and invoking Lemma~\ref{lem:PairwiseProb}
gives:
\[
\expec{H}{\totaltrees{H}^2}
\leq \sum_{k = 0}^{n - 1}
\sum_{\substack{T_1, T_2 \\ \abs{T_1 \cap T_2} = k}}
\ww\left( T_1\right) \cdot \ww\left( T_2 \right) \cdot 
p^{2n-2} \exp\left(-\frac{2n^2}{\amount}\right)
\left( \frac{1}{p} \left(1 + \frac{2n}{\amount}\right) \right)^{k}.
\]
Note that the trailing term only depends on $k$ and can be pulled outside the summation of $T_1,T_2$, so we
then use Lemma~\ref{lem:IntersectionPairs} to bound this by:
\[
\expec{H}{\totaltrees{H}^2}
\leq
\sum_{k = 0}^{n - 1}
\totaltrees{G}^2 \cdot \frac{1}{k!} \left( \frac{n^2}{m} \right)^{k} \cdot 
p^{2n-2} \exp\left(-\frac{2n^2}{\amount}\right)
\left( \frac{1}{p} \left(1 + \frac{2n}{\amount}\right) \right)^{k}.
\]
Which upon pulling out the terms that are independent of $k$,
and substituting in $p = s / m$ gives:
\[
\expec{H}{\totaltrees{H}^2}
\leq
\totaltrees{G}^2 \cdot p^{2n-2} \cdot \exp\left(-\frac{2n^2}{\amount}\right) \cdot
\sum_{k = 0}^{n - 1}
\frac{1}{k!} \cdot \left( \frac{n^2}{s} \left(1 + \frac{2n}{\amount}\right)  \right)^{k}.
\]
From the Taylor expansion of $\exp(\cdot),$ we have:
\begin{align*}
\expec{H}{\totaltrees{H}^2}
&\leq 
\totaltrees{G}^2 \cdot p^{2n-2} \cdot \exp\left(-\frac{2n^2}{\amount}\right)
\cdot \exp \left( \frac{n^2}{s} \left(1 + \frac{2n}{\amount}\right)  \right)\\
&= \totaltrees{G}^2 \cdot p^{2n-2} \cdot \exp\left(-\frac{n^2}{\amount}\right)
\cdot \exp \left( O\left( \frac{n^3}{\amount^2} \right)\right).
\end{align*}

\end{proof}


This bound implies that once we set $\amount^2 > n^3$,
the variance becomes less than the square of the expectation.
It forms the basis of our key concentration results,
which we show in Section~\ref{sec:sparsification},
and also leads to Theorem~\ref{thm:SparsifyDeterminant}.
In particular, we demonstrate that this sampling scheme extends to importance
sampling, where edges are picked with probabilities proportional to
(approximations of) of their leverage scores.

A somewhat surprising aspect of this concentration result
is that there is a difference between models $\mathcal{G}_{n,m}$ and the Erdos-Renyi model $\mathcal{G}_{n,p}$ when the quantity of interest is the number of spanning trees.
In particular, the number of spanning trees of a graph
$G \sim \mathcal{G}_{n,m}$ is approximately normally distributed when
$m = \omega \pr{n^{1.5}},$ whereas it has approximate log-normal 
distribution when  $G \sim \mathcal{G}_{n,p}$ and $p <1.$

An immediate consequence of this is that we can now approximate
$\detp{(\LL^{G})}$ by computing $\detp{(\LL^{H})}$.
It also becomes natural to consider speedups of random spanning tree
sampling algorithms that generate a spanning tree from a sparsifier.
Note however that we cannot hope to preserve the distribution
over all spanning trees via a single sparsifier, as some of the
edges are no longer present.

To account for this change in support, we instead consider the
randomness used in generating the sparsifier as also part of the
randomness needed to produce spanning trees.
In Section~\ref{subsec:EasierTVBound}, we show that just bounds
on the variance of $\totaltrees{H}$ suffices for a bound on the
TV distances of the trees.

\begin{lemma}
\label{lem:VarianceTV}
Suppose $\mathcal{H}$ is a distribution over rescaled subgraphs
of $G$ such that for some parameter some $0 < \delta < 1$ we have
\[
\frac{\expec{H \sim \mathcal{H}}{\totaltrees{H}^2}}{
	\expec{H \sim \mathcal{H}}{\totaltrees{H}}^2}
\leq 1 + \delta,
\]
and for any tree $\hat{T}$ and any graph from the distribution that contain
it, $H$ we have:
\[
\ww^{H}\left( \hatT \right)
= \ww^{G} \left( \hatT \right) \cdot
	\prob{H' \sim \mathcal{H}}{\hatT \subseteq H'}^{-1}
	\cdot \frac{\expec{H' \sim \mathcal{H}}{\totaltrees{H'}}}{\totaltrees{G}},
\]
then the distribution given by $\Pr^{G}(T)$, $p$,
and the distribution induced by $\expec{H \sim \mathcal{H}}{\Pr^{H}(T)}$,
$\tilde{p}$ satisfies
\[
d_{TV}\left(p,\tilde{p}\right) \leq \sqrt{\delta}.
\] 
\end{lemma}
Note that uniform sampling meets the property about $\ww^{H}(T)$
because of linearity of expectation.
We can also check that the importance sampling based routine
that we will discuss in Section~\ref{subsec:GeneralLeverage}
also meets this criteria.
Combining this with the running time bounds from
Theorem~\ref{thm:SparsifyDeterminant}, as well as the
$\tilde{O}(m^{1/3} n^{5/3})$ time random spanning tree sampling
algorithm from~\cite{DurfeeKPRS16} then leads to a faster algorithm.

\begin{corollary}
\label{cor:AlgoOneShot}
For any graph $G$ on $n$ vertices and any $\delta > 0$,
there is an algorithm that generates a tree from a distribution
whose total variation is at most $\delta$ from the random tree
distribution of $G$ in time $\tilde{O}(n^{\frac{13}{6} = 2.1666 \ldots } \delta^{-2/3} + n^2\delta^{-2})$.
\end{corollary}


\subsection{Integration Into Recursive Algorithms}

As a one-step invocation of our concentration bound leads to speedups
over previous routines, we  investigate
tighter integrations of the sparsification routine into
algorithms.
In particular, the sparsified Schur complement
algorithms~\cite{KyngLPSS16} provide a natural place
to substitute spectral sparsifiers with
determinant-preserving ones.
In particular, the identity of
\[
\detp{(\LL)}
= \det{(\LL_{[V_2, V_2]})} \cdot \detp{(\sc{\LL}{V_1})}.
\] where $\detp$ is the determinant of the matrix minor,
suggests that we can approximate $\det{(\LL_{-n})}$
by approximating $\det{(\LL_{[V_2, V_2]})}$ and 
$\detp{(\sc{\LL}{V_1})}$ instead.
Both of these subproblems are smaller by a constant factor,
and we also have $\abs{V_1} + \abs{V_2} = n$.
So this leads to a recursive scheme where the total
number of vertices involved at all layers is $O(n \log{n})$.
This type of recursion underlies both our determinant
estimation and spanning tree sampling algorithms.

The main difficulty remaining for the determinant
estimation algorithm is then sparsifying $\sc{G}{V_1}$
while preserving its determinant.
For this, we note that some $V_1$ are significantly easier
than others: in particular, when $V_2 = V \setminus V_1$
is an independent set, the Schur complement of each of
the vertices in $V_2$ can be computed independently.
Furthermore, it is well understood how to sample these
complements, which are weighted cliques, by a distribution
that exceeds their true leverage scores.

\begin{lemma}
There is a procedure that takes a graph $G$ with $n$ vertices,
a parameter $\delta$, 
and produces in $\tilde{O}(n^2 \delta^{-1})$ time
a subset of vertices $V_1$ with $\abs{V_1} = \Theta(n)$,
along with a graph $H^{V_1}$ such that
\[
\totaltrees{\sc{G}{V_1}} \exp\left(-\delta\right)
\leq \expec{H^{V_1}}{\totaltrees{H^{V_1}}} \leq
\totaltrees{\sc{G}{V_1}} \exp\left(\delta\right),
\]
and
\[
\frac{\expec{H^{V_1}}{\totaltrees{H^{V_1}}^2}}{\expec{H^{V_1}}{\totaltrees{H^{V_1}}}^2}
\leq \exp\left( \delta \right).
\]	
\end{lemma}

 Lemma~\ref{lem:ERDS} holds  w.h.p., and we condition on this event. In our algorithmic applications we will be able to add  the polynomially small failure probability of Lemma~\ref{lem:ERDS} to the error bounds. 

The bound on variance implies that the number of spanning trees is concentrated
close to its expectation, $\totaltrees{\sc{G}{V_1}}$, and that
a random spanning tree drawn from the generated graph $H^{V_1}$ is
---over the randomness of the sparsification procedure---close in
total variation distance to a random spanning tree of the true Schur complement.

As a result, we can design schemes that:
\begin{enumerate}
\item Finds an $O(1)$-DD subset $V_2$, and set $V_1 \leftarrow V \setminus V_2$.
\item Produce a determinant-preserving sparsifier
$H^{V_1}$ for $\sc{G}{V_1}$.
\item Recurse on both $\LL_{[V_2, V_2]}$ and $H^{V_1}$.
\end{enumerate}

However, in this case, the accumulation of error is too
rapid for yielding a good approximation of determinants.
Instead, it becomes necesary to track the accumulation of
variance during all recursive calls.
Formally, the cost of sparsifying so that the variance is at most $\delta$
is about $n^2 \delta^{-1},$ where $\delta$ is the size of the problem.
This means that for a problem on $G_i$
of size $\beta_i n$ for $0 \leq \beta_i \leq 1$, we can afford an error
of $\beta_i \delta$ when working with it, since:
\begin{enumerate}
\item The sum of $\beta_i$ on any layer is at most $2$,
\footnote{each recursive call may introduce one new vertex}
so the sum of variance per layer is $O(\delta)$.
\item The cost of each sparsification step is now
$\beta_i n^2 {\delta}^{-1}$, which sums to about
$n^2 \delta^{-1}$ per layer.
\end{enumerate}


Our random spanning tree sampling algorithm in 
Section~\ref{sec:spanning_tree} is similarly based
on this careful accounting of variance.
We first modify the recursive Schur complement algorithm introduced
by Coulburn et al.~\cite{ColbournDN89} to give a simpler algorithm
that only braches two ways at each step in Section~\ref{subsec:ExactSpanningAlgo}, leading to a high level
scheme fairly similar to the recursive determinant algorithm.
Despite these similarities, the accumulation of errors becomes far
more involved here due to the choice of trees in earlier recursive
calls affecting the graph in later steps.
More specifically, the recursive structure of our determinant algorithm
can be considered analogous to a breadth-first-search, which allows us
to consider all subgraphs at each layer to be independent.
In contrast, the recursive structure of our random spanning tree
algorithm, which we show in Section~\ref{subsec:RandSpanningTreeAlgo}
is more analogous to a depth-first traversal
of the tree, where the output solution of one subproblem will
affect the input of all subsequent subproblems.

These dependency issues will be the key difficulty in considering
variance across levels.
The total variation distance tracks the discrepancy over all trees
of $G$ between their probability of being returned by the overall
recursive algorithm, and their probability in the $\ww$-uniform distribution.
Accounting for this over all trees leads us to bounding
variances in the probabilities of individual trees being picked.
As this is, in turn, is equivalent to the weight of the tree divided
by the determinant of the graph, the inverse of the probability
of a tree being picked can play a simliar role to the determinant
in the determinant sparsification algorithm described above.
However, tracking this value requires analyzing extending our
concentration bounds to the case where an arbitrary tree is fixed
in the graph and we sample from the remaining edges.
We study this Section~\ref{sec:cond_conc}, prove bounds analogous
to the concentration bounds from Section~\ref{sec:sparsification},
and incorporate the guarantees back into the recursive algorithm
in Section~\ref{subsec:RandSpanningTreeAlgo}.


\section{Determinant Preserving Sparsification}
\label{sec:sparsification}

In this section we will ultimately prove Theorem~\ref{thm:SparsifyDeterminant}, our primary result regarding
determinant-preserving sparsification. However, most of this section will be devoted to proving the following general determinant-preserving sparsification routine that also forms the core of subsequent algorithms:


\begin{theorem}
\label{thm:SparsifyMain}
Given an undirected, weighted graph $G = (V, E, \ww)$,
an error threshold $\epsilon > 0$,
parameter $\rho$ along with routines:
\begin{enumerate}
\item $\textsc{SampleEdge}_G()$ that samples an edge
$e$ from a probability distribution $\pp$
($\sum_{e} \pp_e = 1$), as well
as returning the corresponding value of $\pp_e$.
Here $\pp_e$ must satisfy:
\[
\frac{\levscoreexact{e}}{n - 1} \leq \rho \cdot \pp_e
\]
where $\levscoreexact{e}$ is the true leverage score of $e$ in $G$.
\item $\textsc{ApproxLeverage}_G(u, v, \epsilon)$ that returns
the leverage score of an edge $u, v$ in $G$ to an error of $\epsilon$.
Specifically, given an edge $e$, it returns a value
$\levscoreapprox{e}$ such that:
\[
\left( 1 - \epsilon \right) \levscoreexact{e}
\leq \levscoreapprox{e} \leq
\left( 1 + \epsilon \right) \levscoreexact{e}.
\]
\end{enumerate}
There is a routine $\textsc{DetSparsify}(G, \amount, \epsilon)$
that computes a graph $H$ with $s$ edges such that its tree count,
$\totaltrees{H}$, satisfies:
\[
\expec{H}{\totaltrees{H}}
= \totaltrees{G} \left(1 \pm O\left(\frac{n^3}{\amount^2} \right)\right),
\]
and:
\[
\frac{\expec{H}{\totaltrees{H}^2}}{\expec{H}{\totaltrees{H}}^2} \leq
\exp{\left(\frac{\epsilon^2 n^2}{\amount}
+ O\left(\frac{n^3}{\amount^2} \right)\right)}
\]	
Furthermore, the expected running time is bounded by:
\begin{enumerate}
\item $O(s \cdot \rho)$
calls to $\textsc{SampleEdge}_G(e)$
and $\textsc{ApproxLeverage}(e)$ with constant error,
\item $O(s)$ calls to $\textsc{ApproxLeverage}(e)$ with $\epsilon$ error.
\end{enumerate}
\end{theorem}
We establish guarantees for this algorithm using the following steps:
\begin{enumerate}
	\item Showing that the concentration bounds as sketched in
	Section~\ref{sec:overview} holds for approximate leverage scores
	in Section~\ref{subsec:ErrorSparsification}.
	\item Show via taking the limit of probabilistic processes that
	the analog of this process works for sampling a general graph
	where edges can have varying leverage scores.
	This proof is in Section~\ref{subsec:GeneralLeverage}.
	\item Show via rejection sampling that (high error)
	one sided bounds on statistical	leverage scores,
	such as those that suffice for spectral sparsification,
	can also be to do the initial round of sampling instead of two-sided
	approximations of leverage scores.
	This, as well as pseudocode and guarantees of the overall algorithm are given in Section~\ref{subsec:RejectionSample}.
\end{enumerate}

\subsection{Concentration Bound with Approximately Uniform Leverage Scores}
\label{subsec:ErrorSparsification}

Similar to the simplified proof as outlined in Section~\ref{sec:overview},
our proofs relied on uniformly sampling $\amount$ edges from a multi-graph with $m \geq \frac{\amount^2}{n}$ edges, such that all edges have leverage score within multiplicative $1 \pm \epsilon$
of $\frac{n - 1}{s}$, aka. approximately uniform.
The bound that we prove is an analog of Lemma~\ref{lem:SecondMoment}
\begin{lemma}
\label{lem:SecondMomentApprox}
Given a weighted multi-graph $G$ such that $m \geq \frac{\amount^2}{n}$, $\amount \geq n$, and all edges $e \in E$ have $\frac{(1 -\epsilon)(n-1)}{m}\leq \levscore{e} \leq \frac{(1 + \epsilon)(n-1)}{m}$, with $0 \leq \epsilon < 1$, then 
\[
\frac{\expec{H}{\totaltrees{H}^2}}{\expec{H}{\totaltrees{H}}^2}
\leq \exp{\left(\frac{n^2\epsilon^2}{\amount} + O\left(\frac{n^3}{\amount^2} \right)\right)}
\]
\end{lemma}

Similar to the proof of Lemma~\ref{lem:SecondMoment}
in Section~\ref{sec:overview}, we can utilize the bounds on the
probability of $k$ edges being chosen using Lemma~\ref{lem:PairwiseProb}.
The only assumption that changed was the bounds on $\levscore{e}$, which does not affect $\expec{H}{\totaltrees{H}}^2$.
The only term that changes is our upper bound the total weight of trees that contain
some subset of $k$ edges that was the produce of $k$ leverage scores.
At a glance, this product can change by a factor of up to
$(1 + \epsilon)^{k}$, which when substituted naively into the
proof of Lemma~\ref{lem:PairwiseProb} directly would yield
an additional term of
\[
\exp\left( \frac{n^2 \epsilon}{s} \right),
\]
and in turn necessitating $\epsilon < n^{-1/2}$ for a sample
count of $s \approx n^{1.5}$.

However, note that this is the worst case distortion over a subset $F$.
The upper bound that we use, Lemma~\ref{lem:IntersectionPairs}
sums over these bounds over all subsets, and over all edges
we still have
\[
\sum{e \in G} \levscore{e} = n - 1.
\]
Incorporating this allows us to show a tighter bound that
depends on $\epsilon^2$.

Similar to the proof of Lemma~\ref{lem:IntersectionPairs},
we can regroup the summation over all ${m \choose k}$ subsets
of $E(G)$, and bound the fraction of trees containing each
subset $F$ via $\sum_{T: F \subseteq T}\ww(T) \leq \totaltrees{G}\prod_{e \in F} \ttau_e$
via Lemma~\ref{lem:SubsetTree}.
\[
\sum_{\substack{T_1, T_2 \\ \abs{T_1 \cap T_2} = k}}
\ww\left( T_1 \right) \cdot \ww\left( T_2 \right)
\leq
\sum_{\substack{F \subseteq E \\ |F| = k}} \totaltrees{G}^2\prod_{e \in F} \levscore{e}^2
\]

The proof will heavily utilize the fact that
$\sum_{e \in E} \levscore{e} = n-1$.
We bound this in first two steps: first treat
it as a symmetric product over $\levscore{e}^2$,
and bound the total as a function of
\[
\sum_{e} \levscore{e}^2,
\]
then we bound this sum using the fact that
$\sum_{e} \levscore{e} = n - 1$.

The first step utilizes the concavity of the product function,
and bound the total by the sum:
\begin{lemma}
\label{lem:SymmetricProductUpper}
For any set of non-negative values $\xx_1 \ldots \xx_m$
with $\sum_{i} \xx_i \leq z$, we have
\[
\sum_{\substack{F \subseteq \left[1 \ldots m\right] \\ \abs{F} = k}}
\prod_{i \in F} \xx_i
\leq {m \choose k} \left( \frac{z}{m} \right)^{k}.
\]
\end{lemma}

\begin{proof}
We claim that this sum is maximized when $\xx_{i}
= \left(\frac{z}{m} \right)$ for all $e$.

Consider fixing all variables other than some
$\xx_i$ and $\xx_j$, which we assume to be
$\xx_1 \leq \xx_2$ without loss of generality as the 
function is symmetric on all variables:
\[
\sum_{\substack{F \subseteq \left[1 \ldots m\right] \\ \abs{F} = k}}
\prod_{i \in F} \xx_i
=
\xx_1\xx_2 \left(
\sum_{\substack{F \subseteq \left[3 \ldots m\right] \\ \abs{F} = k-2}}
\prod_{i \in F} \xx_i \right)
+ (\xx_1 + \xx_2)
\cdot \left( \sum_{\substack{F \subseteq \left[3 \ldots m\right] \\ \abs{F} = k-1}}
\prod_{i \in F} \xx_i \right)
+ \sum_{\substack{F \subseteq \left[3 \ldots m\right] \\ \abs{F} = k}}
\prod_{i \in F} \xx_i.
\]
	
Then if $\xx_1 < \xx_2$, locally changing their values to 
$\xx_1 + \epsilon$ and $\xx_2 - \epsilon$ keeps the second
term the same.
While the first term becomes
\[
\left( \xx_1 + \epsilon \right) \left( \xx_2 - \epsilon \right)
= \xx_1 \xx_2 + \epsilon \left( \xx_2 - \xx_1\right) - \epsilon^2,
\]
which is greater than $\xx_1\xx_2$ when $0 < \epsilon < ( \xx_2 - \xx_1)$.

This shows that the overall summation is maximized when
all $\xx_i$ are equal, aka
\[
\xx_i = \frac{z}{m},
\]	
which upon substitution gives the result.
\end{proof}


The second step is in fact the $k = 1$ case of
Lemma~\ref{lem:IntersectionPairs}.

\begin{lemma}
\label{lem:SumOfSquaresUpper}
For any set of values $\yy_e$ such that
\[
\sum_{e} \yy = n - 1,
\]
and
\[
\frac{\left(1 - \epsilon\right) n}{m}
\leq \yy_{e} \leq
\frac{\left(1 + \epsilon\right) n}{m},
\]
we have
\[
\sum_{e} \yy_{e}^2 \leq \frac{(1 + \epsilon^2)(n-1)^2}{m}.
\]
\end{lemma}

\begin{proof}
Note that for any $a \leq b$, and any $\epsilon$, we have
\[
\left( a - \epsilon \right)^2 + \left( b + \epsilon\right)^2
= a^2 + b^2 + 2 \epsilon^2 + 2 \epsilon \left( b - a \right),
\]
and this transformation must increase the sum for $\epsilon > 0$.
This means the sum is maximized when half of the leverage scores are $\frac{(1 -\epsilon)(n-1)}{m}$ and the other half are $\frac{(1 +\epsilon)(n-1)}{m}$. This then gives
	
\[
\sum_{e \in E} \yy_{e}^2
\leq \frac{m}{2}\left(\frac{(1 +\epsilon)(n-1)}{m}\right)^2 + \frac{m}{2}\left(\frac{(1 -\epsilon)(n-1)}{m}\right)^2 =
\frac{(1 + \epsilon^2)(n-1)^2}{m}.
\]
\end{proof}

\begin{proof}(of Lemma~\ref{lem:SecondMomentApprox})

We first derive an analog of Lemma~\ref{lem:IntersectionPairs}
for bounding the total weights of pairs of trees containing
subsets of size $k$, where we again start with the bounds 

\[\sum_{\substack{T_1, T_2 \\ \abs{T_1 \cap T_2} = k}}
\ww\left( T_1\right) \cdot \ww\left( T_2 \right)
\leq 
\sum_{\substack{F \subseteq E \\ \abs{F} = k}}
\sum_{\substack{T_1, T_2 \\ F \subseteq {T_1 \cap T_2}}}
\ww\left( T_1 \right) \cdot \ww\left( T_2 \right) = \sum_{\substack{F \subseteq E \\ \abs{F} = k}}
\left(\sum_{\substack{T: F \subseteq T}}
\ww\left( T \right) \right)^2 \] 
Applying Lemma~\ref{lem:SubsetTree} to the inner term of the summation then gives
\[\sum_{\substack{T_1, T_2 \\ \abs{T_1 \cap T_2} = k}}
\ww\left( T_1\right) \cdot \ww\left( T_2 \right)
\leq
\sum_{\substack{F \subseteq E \\ \abs{F} = k}} \totaltrees{G}^2 \cdot \prod_{e \in F} \levscore{e}^2
\]

The bounds on $\levscore{e}$ and $\sum_{e} \levscore{e} = n - 1$
gives, via Lemma~\ref{lem:SumOfSquaresUpper}
\[
\sum_{e} \levscore{e}^2 \leq
\frac{(1 + \epsilon^2)(n-1)^2}{m}.
\]
Substituting this into Lemma~\ref{lem:SymmetricProductUpper}
with $\xx_i = \levscore{e}^2$ then gives
\[
\sum_{\substack{F \subseteq E \\ |F| = k}}
\prod_{e \in F} \levscore{e}^2 
\leq {m \choose k} \left(\frac{(1 + \epsilon^2)n^2}{m^2}\right)^{k}
\leq \frac{m^{k}}{k!} \left(\frac{(1 + \epsilon^2)n^2}{m^2}\right)^{k}
= \frac{1}{k!} \left(\frac{(1 + \epsilon^2)n^2}{m}\right)^{k}.
\] which implies our analog of Lemma~\ref{lem:IntersectionPairs}

\[
\sum_{\substack{T_1, T_2 \\ \abs{T_1 \cap T_2} = k}}
\ww\left( T_1\right) \cdot \ww\left( T_2 \right)
\leq  \totaltrees{G}^2 \cdot
\frac{1}{k!} \left(\frac{(1 + \epsilon^2)n^2}{m}\right)^{k}.
\]

We can then duplicate the proof of Lemma~\ref{lem:SecondMoment}.
Similar to that proof, we can regroup the summation by
$k = \abs{T_1 \cap T_2}$ and invoking Lemma~\ref{lem:PairwiseProb}
to get:
\begin{equation*}
\expec{H}{\totaltrees{H}^2}
\leq \sum_{k = 0}^{n - 1}
\sum_{\substack{T_1, T_2 \\ \abs{T_1 \cap T_2} = k}}
\ww\left( T_1\right) \cdot \ww\left( T_2 \right) \cdot 
p^{2n-2} \exp\left(-\frac{2n^2}{\amount}\right)
\left( \frac{1}{p} \left(1 + \frac{2n}{\amount}\right) \right)^{k}.
\label{eqn:VarianceExpanded}
\end{equation*}
where $p = s / m$. When incorporated with our analog of Lemma~\ref{lem:IntersectionPairs} gives:
\begin{multline*}
\expec{H}{\totaltrees{H}^2}
\leq \sum_{k = 0}^{n - 1} 
p^{2n-2} \exp\left(-\frac{2n^2}{\amount}\right)
\left( \frac{1}{p} \left(1 + \frac{2n}{\amount}\right) \right)^{k}
\cdot \totaltrees{G}^2 \frac{1}{k!} \left(\frac{(1 + \epsilon^2)n^2}{m}\right)^{k}\\
= \totaltrees{G}^2 p^{2n-2} \cdot \exp\left(-\frac{2n^2}{\amount}\right)
\cdot \sum_{k = 0}^{n - 1}  \frac{1}{k!}
\cdot 
\left( \frac{(1 + \epsilon^2)n^2}{s} \left(1 + \frac{2n}{\amount}\right) \right)^{k}.
\end{multline*}
Substituting in the Taylor expansion of
$\sum_{k} \frac{z^k}{k!} 
\leq \exp(z)$ then leaves us with:
\[
\expec{H}{\totaltrees{H}^2}
\leq \totaltrees{G}^2 \cdot  p^{2n-2} \cdot
\exp\left(-\frac{n^2}{\amount}
+ \frac{n^2 \epsilon^2}{\amount}
+ O\left(\frac{n^3}{\amount^2}\right) \right)
\]
and finishes the proof. 

\end{proof}

\subsection{Generalization to Graphs with Arbitrary Leverage Score Distributions}
\label{subsec:GeneralLeverage}

The first condition of $m \geq \frac{\amount^2}{n}$ will be easily achieved
by splitting each edge a sufficient number of times,
which does not need to be done explicitly in the sparsification algorithm.
Furthermore, from the definition of statistical leverage score splitting an edge into $k$ copies
will give each copy a $k$th fraction of the edge's leverage score.
Careful splitting can then ensure the second condition, but will require $\epsilon$-approximate leverage score estimates on the edges.
The simple approach would compute this for all edges,
then split each edge according to this estimate and draw from the resulting edge set.
Instead, we only utilize this algorithm as a proof technique, and give a sampling scheme that's equivalent to this algorithm's limiting behavior as $m \rightarrow \infty$.
Pseudocode of this routine is in Algorithm~\ref{alg:IdealSparsify}.

\begin{algorithm}								
	\caption{$\textsc{IdealSparsify}(G, \ttautil, \amount)$:
		Sample $\amount$ (multi) edges of $G$ to produce $H$ such that
		$\totaltrees{G} \approx \totaltrees{H}$.}
		
	\label{alg:IdealSparsify}
	
	\SetAlgoVlined
	\SetKwProg{myproc}{Procedure}{}{}

	\KwIn{Graph $G$, approximate leverage scores $\ttautil$, sample count $\amount$}
	
	Initialize $H$ as the empty graph, $H \leftarrow \emptyset$;
	
	\For{ i = $1 \ldots s$}{
		Pick edge $e$ with probability proportional to $\ttautil_e$;
		
		Add $e$ to $H$ with new weight:
		\[
			\frac{\ww_e \left( n - 1 \right)}{\ttautil_e \amount}
				\exp\left( \frac{n^2}{2 \left(n - 1\right)s}\right).
		\]
	}	Output $H$
\end{algorithm}

Note that this sampling scheme is with replacement:
the probability of a `collision' as the number of copies tend to $\infty$
is sufficiently small that it can be covered by the proof as well.

The guarantee that we will show for Algorithm~\ref{alg:IdealSparsify} is:
\begin{lemma}
\label{lem:IdealSparsify}
For any graph $G$ and any set of approximate leverage scores $\ttautil$
such that
\[
\left( 1 - \epsilon \right) \ttau_e
\leq \ttautil_e
\leq \left( 1 + \epsilon \right) \ttau_e
\]
for all edges $e$.
The graph $H = \textsc{IdealSparsify}(G, \ttautil, \amount)$ satisfies:
\[
\left( 1 - O\left(\frac{n^3}{s^2} \right) \right)  \totaltrees{G}
\leq \expec{H}{\totaltrees{H}} \leq 
\totaltrees{G},
\]
and
\[
\frac{\expec{H}{\totaltrees{H}^2}}{\expec{H}{\totaltrees{H}}^2}
\leq \exp\left(
	O\left(\frac{\epsilon^2 n^2}{\amount} + \frac{n^3}{\amount^2} \right)
\right).
\]
\end{lemma}

Our proof strategy is simple: claim that this algorithm is statistically
close to simulating splitting each edge into a very large number of copies.
Note that these proofs are purely for showing the convergence of statistical
processes, so all that's needed is for the numbers that arise in
this proof (in particular, $m$) to be finite.

We first show that $G$ and $\ttautil$ can be perturbed to become
rational numbers.
\begin{lemma}
\label{lem:perturb}
For any graph $G$ and any set of $\ttautil$ such that
$(1 - \epsilon) \ttaubar^{(G)}_e \leq \ttautil_e \leq (1 + \epsilon) \ttaubar^{(G)}_e$
for all edges $e$ for some constant $\epsilon > 0$,
and any perturbation threshold $\delta$,
we can find graph $G'$ with all edge weights rationals, and
$\ttautil'$ with all entries rational numbers such that:
\begin{enumerate}
	\item $\totaltrees{G} \leq \totaltrees{G'} \leq (1 + \delta) \totaltrees{G}$, and
	\item $(1 - 2\epsilon) \ttaubar^{(G')}_e \leq \ttautil'_e
		\leq (1 + 2\epsilon) \ttaubar^{(G')}_e$ for all edges $e$.
\end{enumerate}
\end{lemma}

\begin{proof}
	This is a direct consequence of the rational numbers being everywhere dense,
	and that perturbing edge weights by a factor of $1 \pm \alpha$ perturbs
	leverage scores by a factor of up to $1 \pm O(\alpha)$, and
	total weights of trees by a factor of $(1 \pm \alpha)^{n - 1}$.
\end{proof}

Having all leverage scores as integers means that we can do an exact
splitting by setting $m$, the total
number of split edges, to a multiple of the common denominator
of all the $\ttautil'_e$ values times $n - 1$.
Specifically, an edge with approximate leverage score $\ttautil'_e$ becomes
\[
\ttautil'_e \cdot \frac{m}{n - 1}
\]
copies, each with weight
\[
\frac{\ww_e \left( n - 1 \right)}{\ttautil'_e m},
\]
and `true' leverage score
\[
\frac{\levscore{e} \left( n - 1 \right)}{\ttautil'_e m}.
\]
In particular, since
\[
\left( 1 - 2 \epsilon \right) 
\leq \frac{\levscore{e}}{\ttautil_e} \leq \left( 1 + 2 \epsilon \right),
\]
this splitted graph satisfies the condition of
Lemma~\ref{lem:SecondMomentApprox}.
This then enables us to obtain the guarantees of Lemma~\ref{lem:IdealSparsify}
by once again letting $m$ tend to $\infty$.

\begin{proof}(of Lemma~\ref{lem:IdealSparsify})
We first show that Algorithm~\ref{alg:IdealSparsify} works for the graph with
rational weights and approximate leverage scores as generated by
Lemma~\ref{lem:perturb}.

The condition established above means that we can apply
Lemma~\ref{lem:SecondMomentApprox} to the output of picking
$\amount$ random edges among these $m$ split copies.
This graph $H'$ satisfies
\[
\expec{H'}{\totaltrees{H'}} = \totaltrees{G'} \left( \frac{s}{m} \right)^{n - 1}
\exp\left( -\frac{n^2}{2s} - O\left(\frac{n^3}{s^2}\right) \right),
\]
and
\[
\frac{\expec{H'}{\totaltrees{H'}^2}}{\expec{H'}{\totaltrees{H'}}^2}
\leq \exp{\left(\frac{n^2\epsilon^2}{\amount} + O\left(\frac{n^3}{\amount^2} \right)\right)}.
\]
The ratio of the second moment is not affected by rescaling,
so the graph
\[
H'' \leftarrow \frac{m}{s} \exp\left( \frac{n^2}{2s \left( n - 1\right)} \right)
\]
meets the requirements on both the expectation and variances.
Furthermore, the rescaled weight of an single edge being picked is:
\[
\frac{\ww_e \left( n - 1 \right)}{\ttautil'_e m}
\cdot \frac{m}{s} \exp\left( \frac{n^2}{2s \left( n - 1\right)} \right)
= \frac{\ww_e \left( n - 1 \right)}{\ttautil'_e s}
\exp\left( \frac{n^2}{2s \left( n - 1\right)} \right),
\]
which is exactly what Algorithm~\ref{alg:IdealSparsify} assigns.

It remains to resolve the discrepancy between sampling with and without
replacement:
the probability of the same edge being picked twice in two different steps
is at most $1 / m$, so the total probability of a duplicate sample is
bounded by $s^2 / m$. We then give a finite bound on the size of $m$ for which this probability becomes negligible in our routine.
The rescaling factor of a single edge is (very crudely) bounded by 
\[
\frac{\left( n - 1 \right)}{\ttautil'_e s}
 \exp\left( \frac{n^2}{2s \left( n - 1\right)} \right)
\leq \exp\left(n^3 \right) \frac{1}{\min_e \ttautil'_e},
\]
which means that any of the $H''$ returned must satisfy
\[
\totaltrees{H''}
\leq \exp\left(n^4 \right) \left( \frac{1}{\min_e \ttautil'_e}\right)^{n}
\totaltrees{G'},
\]
which is finite.
As a result, as $m \rightarrow \infty$, the difference that this causes
to both the first and second moments become negligible.


The result for $H \leftarrow \textsc{IdealSparsify}(G, \ttautil, \amount)$
then follows from the infinitesimal perturbation made to $G$, as the rational
numbers are dense everywhere.
\end{proof}

\subsection{Incorporating Crude Edge Sampler Using Rejection Sampling}
\label{subsec:RejectionSample}

Under Lemma~\ref{lem:IdealSparsify} we assumed access to $\epsilon$-approximate leverage scores, which could be computed with $m$ calls to our assumed subroutine $\textsc{ApproxLeverage}_G$, where $m$ here is the number of edges of $G$. However, we roughly associate $\textsc{ApproxLeverage}_G$ with Lemma~\ref{lem:ERDS} that requires $\tilde{O}(\epsilon^{-2})$ time per call (and we deal with the w.h.p. aspect in the proof of Theorem~\ref{thm:SparsifyDeterminant}), and to achieve our desired sparsification of $O(n^{1.5})$ edges, we will need $\epsilon = n^{-1/4}$ for the necessary concentration bounds. Instead, we will show that we can use rejection sampling to take $\amount$
edges drawn from approximate leverage scores using a cruder
distribution $\pp_e$, which will only require application of $\textsc{ApproxLeverage}_G$ with error $\epsilon$ for an expected $O(\amount)$ number of edges.

Rejection sampling is a known technique that allows us to sample from some distribution
$f$ by instead sampling from a distribution $g$ that approximates $f$ and accept
the sample with a specific probability based on the probability of drawing that
sample from $f$ and $g$.

More specifically, suppose we are given two probability distributions
$f$ and $g$ over the same state space $X$, such that for all $x\in X$
we have $Cg(x) \geq f(x)$ for some constant $C$.
Then we can draw from $f$ by instead drawing $x \sim g$, and accepting the draw with probability $\frac{f(x)}{Cg(x)}$. 

This procedure only requires a lower bound on $g$ with respect to $f$,
but in order to accept a draw with constant probability,
there need to be weaker upper bound guarantees.
Our guarantees on $\levscoreapprox{e}$ will fulfill these requirements,
and the rejection sampling will accept a constant fraction of the draws.
By splitting into a sufficient number of edges,
we ensure that drawing the same multi-edge from any split edge will
occur with at most constant probability.

Specifically, each sample is drawn via. the following steps:
\begin{enumerate}
\item Draw a sample according the distribution $g$, $e$.
\item Evaluate the values of $f(e)$ and $g(e)$.
\item Keep the sample with probability $f(e) / g(e)$.
\end{enumerate} 
As the running time of $\textsc{ApproxLeverage}_G(e, \epsilon)$ will ultimately depend
on the value of $\epsilon$ apply this algorithmic framework, 
we also need to perform rejection sampling twice, once with constant
error, and once with leverage scores extracted from the true
approximate distribution.
Pseudocode of this routine is shown in Algorithm~\ref{alg:DetSparsify}.

\begin{algorithm}								
	\caption{$\textsc{DetSparsify}(G, \amount, \textsc{SampleEdge}_G()), \rho,
		\textsc{ApproxLeverage}_G(u, v, \epsilon))$:
		Sample $\amount$ (multi) edges of $G$ to produce $H$ such that
		$\totaltrees{G} \approx \totaltrees{H}$.}

	\label{alg:DetSparsify}
	
	\SetAlgoVlined
	\SetKwProg{myproc}{Procedure}{}{}
	
	\KwIn{Graph $G$.\\
	Sample count $\amount$, leverage score approximation error $0 < \epsilon < 1/2$,\\
	$\textsc{SampleEdge}_G()$ that samples an edge
	 $e$ from a probability distribution $\pp$
	 ($\sum_{e} \pp_e = 1$), and returning the corresponding value of $\pp_e$.\\
	 $\rho$ that bounds the under-sampling rate of $\textsc{SampleEdge}_G()$.\\
	 $\textsc{ApproxLeverage}_G(u, v, \epsilon)$ that returns the approximate
	 leverage score of an edge $u, v$ in $G$ to an error of $\epsilon$.
	}
	
	Initialize $H$ as the empty graph, $H \leftarrow \emptyset$;
	
	\While{$H$ has fewer than $\amount$ edges}{
		$e, \pp_e \leftarrow \textsc{SampleEdge}_G()$.
		
		Let $\pp_e' \leftarrow \frac{2}{n - 1} \textsc{ApproxLeverage}_G(u, v, 0.1)$
		
		Reject $e$ with probability $1 - \pp_e' / (4 \rho \cdot \pp_e)$.

		Let $\pp_e'' \leftarrow 
		\frac{1}{n - 1} \textsc{ApproxLeverage}_G(u, v, \epsilon)$

		Reject $e$ with probability $1 - \pp_e'' / \pp_e'$.
	
		Add $e$ to $H$ with new weight
		\[\frac{\ww_e}{\pp_e'' s}
			\exp\left( \frac{n^2}{2(n - 1)\amount}\right).
		\]
	}	Output $H$
\end{algorithm}

We first show that this routine will in fact sample edges according to $\epsilon$-approximate leverage scores, as was assumed in $\textsc{IdealSparsify}$

\begin{lemma}
\label{lem:RejectionProbGood}
The edges are being sampled with probability
proportional to $\ttautil^{(G, \epsilon)}$, the leverage score estimates
given by $\textsc{ApproxLeverage}_G(\cdot, \epsilon)$.
\end{lemma}

Note that this algorithm does not, at any time, have access to the full
distribution $\ttautil^{(G, \epsilon)}$.

\begin{proof} Our proof will assume the known guarantees of rejection sampling, which is to say that the following are true:
	
	\begin{enumerate}
		\item Given distributions $\pp$ and $\pp'$, sampling an edge $e$ from $\pp$ and accepting with probability $\pp_e'/(4\rho \cdot \pp_e)$ is equivalent to drawing an edge from $\pp'$ as long as $\pp_e'/(4\rho \cdot \pp_e) \in [0,1] $ for all $e$.
		
		\item Given distributions $\pp'$ and $\pp''$, sampling an edge $e$ from $\pp'$ and accepting with probability $\pp_e''/\pp_e'$ is equivalent to drawing an edge from $\pp''$ as long as $\pp_e''/\pp_e' \in [0,1]$ for all $e$.
	\end{enumerate}
	As a result, we only need to check that  $\pp_e' / (4 \rho \pp_e)$
	and $\pp_e'' / \pp_e'$ are at most $1$.
	
	The guarantees of $\textsc{SampleEdge}_G()$  gives
	\[
		\frac{\levscoreexact{e}}{n - 1} \leq \rho \pp_e.
	\]
	As $\pp_e'$ was generated with error $1.1$, we have
	\[
		\pp_e' \leq \frac{2.2 \levscoreexact{e} }{(n - 1)} 
		\leq 2.2 \rho \pp_e,
	\]
	so $\pp_e' / (4 \rho \pp_e) \leq 1$. To show $\pp_e'' / \pp_e'\leq 1 $, once again the guarantees of
	 $\textsc{SampleEdge}_G()$ gives:
	\[
		\pp_e'' \leq \left( 1 + \epsilon \right) \frac{\levscoreexact{e}}{n - 1}
		\leq 2 \cdot 0.9 \frac{\levscoreexact{e}}{n - 1}
		\leq \pp_e'.
	\]

\end{proof}

It remains to show that this rejection sampling process still makes
sufficiently progress, yet also does not call $\textsc{ApproxLeverage}_G(e, \epsilon)$
(the more accurate leverage score estimator) too many times.
\begin{lemma}
\label{lem:RejectionProgress}
	At each step, the probability of \textsc{DetSparsify}
	calling $\textsc{ApproxLeverage}_G(e, \epsilon)$ is at most
	$\frac{1}{\rho}$, while the probability of it adding an edge
	to $H$ is at least $\frac{1}{8 \rho}$.
\end{lemma}

\begin{proof}
	The proof utilizes the fact $\sum_{e} \levscoreexact{e} = n - 1$
	(Fact~\ref{fact:foster}) extensively.
	
	If the edge $e$ is picked,
	$\textsc{ApproxLeverage}_G(e, \epsilon)$ is called with probability
	\[
		\frac{\pp_e'}{4 \rho \cdot \pp_e}
		\leq \frac{2.2 \levscoreexact{e} }{4 \rho \cdot \pp_e \cdot (n - 1)}
	\]
	Summing over this over all edge $e$ by the probability of picking them
	gives:
	\[
		\sum_{e} \pp_e \frac{2.2 \levscoreexact{e} }{4 \rho \cdot \pp_e \cdot (n - 1)}
		= \frac{2.2 \sum_{e}  \levscoreexact{e} }{4 \rho \cdot (n - 1)}
		\leq \frac{1}{\rho}.
	\]
	
	On the other hand, 	the probability of picking edge $e$,
	and not rejecting it is:
	\[
	\pp_e \cdot \frac{\pp_e'}{4 \rho \cdot \pp_e} \cdot \frac{\pp_e''}{\pp_e'}
	= \frac{\ttautil^{(G, \epsilon)}}{4 \rho (n - 1)}.
	\] where this follows by cancellation and how we set $\pp_e''$ in our algorithm. Summing over all edges then gives the probability of not rejecting an edge to be 
	
	\[ \sum_{e} \frac{\ttautil^{(G, \epsilon)}}{4 \rho (n - 1)} \geq  \sum_{e} \frac{(1 - \epsilon)\levscore{e}}{4 \rho (n - 1)} =  \frac{(1 - \epsilon)\sum_{e}\levscore{e}}{4 \rho (n - 1)} \geq  \frac{1}{8\rho}\] 
\end{proof}

\begin{proof}(of Theorem~\ref{thm:SparsifyMain})
	Lemma~\ref{lem:RejectionProbGood} implies that edges are sampled in $\textsc{DetSparsify}$ with probability proportional to $\epsilon$-approximate leverage scores guaranteed by $\textsc{ApproxLeverage}_{G}(\cdot,\epsilon)$. Therefore, we can apply Lemma~\ref{lem:IdealSparsify} to achieve the desired expectation and concentration bounds. Finally, Lemma~\ref{lem:RejectionProgress} implies that we expect to sample at most $O(\amount \cdot \rho)$ edges, each of which require a call to $\textsc{SampleEdge}_{G}(e)$ and $\textsc{ApproxLeverage}_{G}$ with constant error. It additionally implies that we expect to make $O(\amount)$ calls to $\textsc{ApproxLeverage}_{G}$ with $\epsilon$ error.
	
\end{proof}

Directly invoking this theorem leads to the sparsification algorithm.
\begin{proof}(of Theorem~\ref{thm:SparsifyDeterminant})
Consider invoking Theorem~\ref{thm:SparsifyMain} with parameters
\begin{align*}
s & \leftarrow O\left(n^{1.5} \delta^{-2}\right),\\
\epsilon & \leftarrow n^{-1/4}.
\end{align*}
This gives:
\[
\frac{\epsilon^2 n^2}{s},  \frac{n^3}{s^2} \leq \delta,
\]
which then implies
\[
\left( 1 - O\left( \delta^2 \right)\right) \totaltrees{G}
\leq \expec{H}{\totaltrees{H}} \leq \left( 1 + O\left( \delta^2 \right)\right) \totaltrees{G},
\]
and
\[
\expec{H}{\totaltrees{H}^2} \leq \left( 1 + O\left( \delta^2 \right) \right) \expec{H}{\totaltrees{H}}^2.
\]
The second condition is equivalent to $\var{H}{\totaltrees{H}} \leq \delta^2 \expec{H}{\totaltrees{H}}$,
which by Chebyshev inequality gives that with constant probability we have
\[
\left( 1 - O\left( \delta \right)\right) \totaltrees{G}
\leq \totaltrees{H} \leq \left( 1 + O\left( \delta \right)\right) \totaltrees{G}.
\]
Combining this with the bounds on $\expec{H}{\totaltrees{H}}$,
and adjusting constants gives the overall bound.

Constructing the probability distribution $\pp$ for sampling edges only requires computing constant approximate leverage scores for all edges, and then sampling proportionally for each edge, giving a constant value for $\rho$. By Lemma~\ref{lem:ERDS}, this requires $\tilde{O}(m)$ time.
The running time then is dominated by the $O(s)$ calls made to the
effective resistance oracle with error $\epsilon = n^{-1/4}$.
Invoking Lemma~\ref{lem:ERDS} gives that this cost is bounded by
\[
O\left(n\epsilon^{-4} + s \epsilon^{-2} \right)
= O\left(n^2 \delta^{-2} \right).
\]

Furthermore, because Lemma~\ref{lem:ERDS} holds w.h.p. we can absorb the probability of failure into our constant probability bound
\end{proof}

Another immediate consequence of our sparsification routine in Theorem~\ref{thm:SparsifyMain}, along with bounds on total variation distance that we prove in Section~\ref{sec:TVBound}, is that we can give a faster spanning tree sampling algorithm for dense graphs by plugging the sparsified graph into previous algorithms for generating random spanning trees.

\begin{proof}(of Corollary~\ref{cor:AlgoOneShot})
	As in the proof of Theorem~\ref{thm:SparsifyDeterminant}, we invoke Theorem~\ref{thm:SparsifyMain} with parameters
	\begin{align*}
	s & \leftarrow O\left(n^{1.5} \delta^{-2}\right),\\
	\epsilon & \leftarrow n^{-1/4}.
	\end{align*}
	
	giving \[
	\frac{\expec{H}{\totaltrees{H}^2}}{\expec{H}{\totaltrees{H}}^2}
	\leq 1 + \delta^2.\]
	Applying Lemma~\ref{lem:VarianceTV}, which is proven in Section~\ref{subsec:EasierTVBound}, we then have that drawing a tree from $H$ according to the $\ww$-uniform distribution gives a total variation distance of $\delta$ from drawing a tree according to the $\ww$-uniform distribution of $G$. 
	The running time of drawing $H$ is dominated by the $O(s)$ calls made to the
	effective resistance oracle with error $\epsilon = n^{-1/4}$.
	Invoking Lemma~\ref{lem:ERDS} gives that this cost is bounded by
	\[
	O\left(n\epsilon^{-4} + s \epsilon^{-2} \right)
	= O\left(n^2 \delta^{-2} \right).
	\]
	
	Furthermore, because Lemma~\ref{lem:ERDS} holds w.h.p. we can absorb the probability of failure into our total variation distance bound (where we implicitly assume that $\delta$ is at most polynomially small).
	
	We then use the $\tilde{O}(m^{1/3}n^{5/3})$ time algorithm in ~\cite{DurfeeKPRS16} with $m = O(n^{1.5}\delta^{-2})$ to draw a tree from $H$. This then achieves our desired running time and total variation distance bound.
	
\end{proof}
	\section{Implicit Sparsification of the Schur Complement}
\label{sec:ImplicitSchur}

Note that the determinant sparsification routine in Theorem~\ref{thm:SparsifyMain}
only requires an oracle that samples edges by an approximate distribution to resistance,
as well as access to approximate leverage scores on the graph.
This suggests that a variety of naturally dense objects, such as random
walk matrices~\cite{ChengCLPT15,JindalKPS17:arxiv} and
Schur complements~\cite{KyngLPSS16,DurfeeKPRS16} can also be sparsified in ways that preserve
the determinant (of the minor with one vertex removed) or the spanning tree distributions.
The latter objects, Schur complements, have already been shown to lead to speedups
in random spanning tree generation algorithms recently~\cite{DurfeeKPRS16}.

Furthermore the fact that Schur complements preserve effective resistances
exactly (\ref{fact:SchurResistance}) means that we can directly invoke the effective
resistances data structure as constructed in Lemma~\ref{lem:ERDS}
to produce effective resistance estimates on any of its Schur complements.
As a result, the main focus of this section is an efficient way of producing samples from
a distribution that approximates drawing a multi-edge from the Schur complement
with probabilities proportional to its leverage score.
Here we follow the template introduced in~\cite{KyngLPSS16} of only eliminating
$(1 + \alpha)$-diagonally-dominant subsets of vertices, as it in turn allows the
use of walk sampling based implicit sparsification similar to those
in~\cite{ChengCLPT15,JindalKPS17:arxiv}.

$(1 + \alpha)$-diagonally-dominant subsets have been used in Schur complement
based linear system solvers to facilitate the convergence of iterative
methods in the $\LL_{[V_2, V_2]}$ block~\cite{KyngLPSS16}.
Formally, the condition that we require is:


\begin{definition}
	\label{def:AlphaDD}
	In a weighted graph $G = (V, E, \ww)$,
	a subset of vertices $V_2 \subseteq V$ is $(1 + \alpha)$-diagonally-dominant,
	or $(1 + \alpha)$-DD if for every $u \in V_2$ with weighted degree $\deg_u$ we have:
	\[
	\sum_{v \sim u, v \notin V_2} \ww_{uv}
	\geq \frac{1}{1 + \alpha} \deg_u
	= \frac{1}{1 + \alpha} \sum_{v \sim u} \ww_{uv}.
	\]
\end{definition}
It was shown in \cite{KyngLPSS16} that large sets
of such vertices
can be found by trimming a uniformly random sample.
\begin{lemma}
	\label{lem:FindingAlphaDDSubsets}
	(Lemma 3.5. of~\cite{KyngLPSS16} instantiated on graphs)\\
	There is a routine $\textsc{AlmostIndependent}(G, \alpha)$ that
	for a graph $G$ with $n$ vertices, and a parameter $\alpha \geq 0$,
	returns in $O(m)$ expected time a subset $V_2$ with  $|V_2| \geq n/(8(1 + \alpha))$
	such that $\LL_{G,[V_2, V_2]}$ is $(1+\alpha)$-DD.
\end{lemma}

Given such a subset $V_2$, we then proceed to sample edges in $\sc{G}{V_1}$
via the following simple random walk sampling algorithm:
\begin{enumerate}
	\item Pick a random edge in $G$.
	\item Extend both of its endpoints in random walks until they first
	reach somewhere in $V_1$.
\end{enumerate}
Incorporating this scheme into the determinant preserving
sparsification schemes then leads these guarantees:
\begin{theorem}
	\label{thm:SchurSparse}
	Conditioned on Lemma~\ref{lem:ERDS} holding, there is a procedure $\textsc{SchurSparse}$ that takes
	a graph $G$, and an $1.1$-DD subset of vertices $V_2$,
	returns a graph $H^{V_1}$ in $\tilde{O}(n^2 \delta^{-1})$ expected time
	such that the distribution over $H^{V_1}$ satisfies:
	\[
	\totaltrees{\sc{G}{V_1}} \exp\left(-\delta\right)
	\leq \expec{H^{V_1}}{\totaltrees{H^{V_1}}} \leq
	\totaltrees{\sc{G}{V_1}} \exp\left(\delta\right),
	\]
	and
	\[
	\frac{\expec{H^{V_1}}{\totaltrees{H^{V_1}}^2}}{\expec{H^{V_1}}{\totaltrees{H^{V_1}}}^2}
	\leq \exp\left( \delta \right).
	\]	
	Furthermore, the number of edges of $H^{V_1}$ can be set to
	anywhere between $O(n^{1.5}\delta^{-1})$ and $O(n^2\delta^{-1})$
	without affecting the final bound.
\end{theorem} 



We let this subset of vertices produced to be $V_2$,
and let its complement be $V_1$.
Our key idea is to view $\sc{G}{V_1}$ as a multi-graph where each multi-edge
corresponds to a walk in $G$ that starts and ends in $V_1$, but has all
intermediate vertices in $V_2$.
Specifically a length $k$ walk
\[
u_0, u_1, \ldots u_k,
\]
with $u_0, u_k \in V_1$ and $u_i \in V_2$ for all $0 < i < k$,
corresponds to a multi-edge between $u_0$ and $u_k$ in $\sc{G}{V_1}$
with weight given by
\begin{equation}\label{eq:schur-weights}
\ww^{\sc{G}{V_1}}_{u_0, u_1, \ldots u_k}
\defeq
\frac{\prod_{0 \leq i < k} \ww^{G}_{u_i u_{i + 1}}}{\prod_{0 < i < k} \deg^{G}_{u_i}}.
\end{equation}
We check formally that this multi-graph defined on $V_1$ is exactly the same
as $\sc{G}{V_1}$ via the Taylor expansion of $\LL_{[V_2, V_2]}^{-1}$ based Jacobi iteration.

\begin{lemma}
\label{lem:SchurWeights}
Given a graph $G$ and a partition of its vertices into $V_1$ and $V_2$,
the graph $G^{V_1}$ formed by all the multi-edges corresponding to walks
starting and ending at $V_1$, but stays entirely within $V_2$ with weights
given by Equation~\ref{eq:schur-weights} is exactly $\sc{G}{V_1}$.
\end{lemma}

\begin{proof}
Consider the Schur complement:
\[
\sc{G}{V_1} = \LL_{[V_1, V_1]} -  \LL_{[V_2, V_1]} \LL_{[V_2, V_2]}^{\dag} \LL_{[V_1, V_2]}.
\]
If there are no edges leaving $V_2$, then the result holds trivially.
Otherwise, $\LL_{[V_2, V_2]}$ is a strictly diagonally dominant matrix,
and is therefore full rank.
We can write it as
\[
\LL_{\left[V_2, V_2\right]} = \DD - \AA
\]
where $\DD$ is the diagonal of $\LL_{[V_2,V_2]}$
and $\AA$ is the negation of the off-diagonal entries,
and then expand $\LL_{[V_2, V_2]}^{-1}$ via the Jacobi series:
\begin{multline}
\label{eq:Jacobi}
\LL_{[V_2, V_2]}^{-1}
= \left( \DD - \AA \right)^{-1}
= \DD^{-1/2} \left( \II - \DD^{-1/2} \AA \DD^{-1/2} \right)^{-1} \DD^{-1/2} \\
= \DD^{-1/2} \left[ \sum_{k = 0}^{\infty} \left( \DD^{-1/2} \AA \DD^{-1/2}\right)^{k} \right] \DD^{-1/2}
= \sum_{k = 0}^{\infty} \left( \DD^{-1}  \AA \right)^{k} \DD^{-1}.
\end{multline}
Note that this series converges because
the strict diagonal dominance of $\LL_{[V_2, V_2]}$
implies $(AD^{-1})^k$ tends to zero as $k \to \infty$.
Substituting this in place of $\LL_{[V_2, V_2]}^{-1}$ gives:
\[
\sc{G}{V_1}
= \LL_{\left[V_1,V_1 \right]}
- \sum_{k = 0}^{\infty} \LL_{\left[V_1, V_2 \right]}
\left( \DD^{-1}  \AA \right)^{k} \DD^{-1}
\LL_{\left[V_2, V_1 \right]}.
\]
As all the off-diagonal entries in $\LL$ are non-positive,
we can replace $\LL_{[V_1, V_2]}$ with $-\LL_{[V_1, V_2]}$ to make
all the terms in the trailing summation positive.
As these are the only ways to form new off-diagonal entries,
the identity based on matrix multiplication of
\[
\left[ \left( -\LL_{\left[V_1, V_2 \right]} \right)
\left( \DD^{-1}  \AA \right)^{k} \DD^{-1}
\left( -\LL_{\left[V_2, V_1 \right]} \right) \right]_{u_0, u_k}
= \sum_{u_1 \ldots u_{k - 1}}
\frac{\prod_{0 \leq i < k} \ww^{G}_{u_i u_{i + 1}}}{\prod_{0 < i < k} \deg^{G}_{u_i}}
\]
gives the required identity.
\end{proof}

This characterization of $\sc{G}{V_1}$, coupled with the $(1 + \alpha)$-diagonal-dominance
of $V_2$, allows us to sample the multi-edges in $\sc{G}{V_1}$ in the same way as the (short)
random walk sparsification algorithms from~\cite{ChengCLPT15,JindalKPS17:arxiv}.

\begin{algorithm}								
\caption{$\textsc{SampleEdgeSchur}(G = (V, E, \ww), V_1)$:
\label{algo:SampleEdgeSchur}
samples an edge from $\sc{G}{V_1}$}
\SetAlgoVlined
\SetKwProg{myproc}{Procedure}{}{}

\KwIn{Graph ${G}$, vertices $V_1$ to complement onto,
	and (implicit) access to a $2$-approximation of the leverage
	scores of $G$, $\ttautil^{G}$.}
\KwOut{A multi-edge $e$ in $\sc{G}{V_1}$ corresponding to a walk $u_0, u_1, \ldots u_k$,
	and the probability of it being picked in this distribution $\pp_{u_0, u_1, \ldots u_k}$}

Sample an edge $e$ from $G$ randomly with probability drawn
from $\ttautil^{G}_e$\;
Perform two independent random walks from the endpoints of $e$
until they both reach some vertex in $V_1$, let the walk be $u_0 \ldots u_k$\;
Output edge $u_0 u_k$ (corresponding to the path $u_0, u_1, \ldots u_k$) with
\begin{align*}
\ww_{u_0 \ldots u_k} & \leftarrow \frac{\prod_{0 \leq i < k} \ww^{G}_{u_i u_{i + 1}}}{\prod_{0 < i < k} \deg^{G}_{u_i}},
\qquad \text{(same as Equation~\ref{eq:schur-weights})}\\
\pp_{u_0 \ldots u_k} & \leftarrow \frac{1}{\sum_{e'} \ttautil^{G}_{e'}}
\sum_{0 \leq i < k} \ttautil^{G}_{u_i u_{i + 1}}
\cdot \left( \prod_{0 \leq j < i}
\frac{\ww^{G}_{u_{j} u_{j + 1}}}{\deg_{u_{j + 1}}}
\cdot \prod_{i + 1 \leq j < k} \frac{\ww^{G}_{u_{j} u_{j + 1}}}{\deg_{u_{j}}} \right).
\end{align*}
\end{algorithm}

\begin{lemma}\label{lem:implicit-sampling}
Given any graph $G = (V, E, \ww)$, an $(1 + \alpha)$-DD subset $V_2$,
and access to $2$-approximations of statistical leverage scores on $G$, $\ttautil^{G}$,
$\textsc{SampleEdgeSchur}$ returns edges in $G$ according to the distribution
$\pp_e$ in $O(\alpha)$ expected time per sample.
Furthermore, the distribution that it samples edges in $\sc{G}{V_1}$ from,
$\pp$, satisfies
\[
O\left(1 \right) \cdot \pp_{u_0, \ldots u_k}
 \geq \frac{\ttaubar^{\sc{G}{V_1}}_{u_0, \ldots u_k}}{n - 1}.
\]
for every edge in $\sc{G}{V_1}$ corresponding to the walk $u_0, \ldots u_k$.
\end{lemma}

The guarantees of this procedure are analogous to the random walk sampling
sparsification scheme from~\cite{ChengCLPT15,JindalKPS17:arxiv},
with the main difference being the terminating condition for the
walks leads to the removal of an overhead related to the number
of steps in the walk.
The modification of the initial step to picking the initial edge
from $G$ by resistance is necessary to get $\rho$ to a constant,
as the about $n^{1.5}$ samples limits the amount of overhead
that we can have per sample.

\begin{proof}
We first verify that $\pp$ is indeed a probability on
the multi-edges of $\sc{G}{V_1}$, partitioned by the
walks that they correspond to in $G$, or formally
\[
\sum_{\substack{u_0, u_1, \ldots u_k:\\
	u_0, u_k \in V_1,\\
	u_i \in V_2 ~\forall 1\leq i<k}}
\pp_{u_0, u_1 \ldots u_k} = 1.
\]
To obtain this equality, note that for any
random walk starting at vertex $i$, the total
probabilities of walks starting at $i$ and ending in $V_1$ is upper bounded by $1$.
Algebraically this becomes:
\[
\sum_{u_1, u_2, \ldots u_k} \prod_{0 \leq i < k} \frac{\ww_{u_i u_{i + 1}}}{\deg_{u_{i}}} = 1,
\]
so applying this to both terms of each edge $e$ gives
that the total probability mass over any starting edge
is $\frac{\ttautil^{G}_e}{\sum_{e'} \ttautil^{G}_{e'}}$,
and in turn the total.

For the running time, since $V_2$ is $(1 + \alpha)$-almost independent,
each step of the walk takes expected time $O(\alpha)$.
Also, the value of $\pp_{u_0, u_1, \ldots u_k}$ can be computed
in $O(k)$ time by computing prefix/suffix products of the transition
probabilities along the path (instead of evaluating each summand in $O(k)$
time for a total of $O(k^2)$).

Finally, we need to bound the approximation of $\pp$ compared
to the true leverage scores $\ttaubar$.
As $\ttautil^{G}_e$ is a $2$-approximation of the true leverage
scores, $\sum_{e} \ttautil^{G}_e$ is within a constant factor of $n$.
So it suffices to show
\[
O\left(1\right) \cdot
\sum_{0 \leq i < k} \ttautil^{G}_{u_i u_{i + 1}}
\left( \prod_{0 \leq j < i}
\frac{\ww^{G}_{u_{j} u_{j + 1}}}{\deg_{u_{j + 1}}}
\cdot \prod_{i + 1 \leq j < k} \frac{\ww^{G}_{u_{j} u_{j + 1}}}{\deg_{u_{j}}} \right)
\geq \er^{\sc{G}{V_1}}\left(u_0, u_k\right) \cdot \ww_{u_0, u_1, \ldots u_k}.
\]
Here we invoke the equivalence of effective resistances in $G$
and $\sc{G}{V_1}$ given by Fact~\ref{fact:SchurResistance}
in the reverse direction.
Then by Rayleigh's monotonicity principle, we have
\[
\er^{\sc{G}{V_1}} \left(u_0, u_k\right)
= \er^{G}\left(u_0, u_k\right)
\leq \sum_{0 \leq i < k}
\frac{2 \ttautil^{G}_{u_i u_{i + 1}}}{\ww_{u_i u_{i + 1}}},
\]
which when substituted into the expression for $\ww_{u_0, u_1, \ldots u_k}$
from Equation~\ref{eq:schur-weights} gives
\[
\left( \sum_{0 \leq i < k}
\frac{2 \ttautil^{G}_{u_i u_{i + 1}}}{\ww_{u_i u_{i + 1}}} \right)
\ww_{u_0, u_1, \ldots u_k}
= 
\sum_{0 \leq i < k} 2 \ttautil^{G}_{u_i u_{i + 1}}
\left( \prod_{0 \leq j < i}
\frac{\ww^{G}_{u_{j} u_{j + 1}}}{\deg_{u_{j + 1}}}
\cdot \prod_{i + 1 \leq j < k} \frac{\ww^{G}_{u_{j} u_{j + 1}}}{\deg_{u_{j}}} \right).
\]

\end{proof}


This sampling procedure can be immediately combined with
Theorem~\ref{thm:SparsifyMain} to give algorithms for generating
approximate Schur complements.
Pseudocode of this routine is in Algorithm~\ref{algo:SchurSparse}.

\begin{algorithm}	
\label{algo:SchurSparse}							
	\caption{$\textsc{SchurSparse}(G, V_1, \delta)$}
	\label{alg:schurSparse}
	\SetAlgoVlined
	\SetKwProg{myproc}{Procedure}{}{}

	\KwIn{Graph ${G}$, $1.1$-DD subset of vertices $V_2$ and error parameter $\delta$}
	\KwOut{Sparse Schur complement of $\sc{G}{V_1}$}
	Set $\epsilon \leftarrow 0.1$\;
	Set $\amount \leftarrow n^{2} \delta^{-1}$\;
	Build leverage score data structure on $G$ with errors $0.1$ (via Lemma~\ref{lem:ERDS})\;
	Let $H^{V_1} \leftarrow \textsc{DetSparsify}(\sc{G}{V_1}, s,
		\textsc{SampleEdgeSchur}(G,V_1),
		\textsc{LeverageApprox}_G, \epsilon)$\;
	Output $H^{V_1}$\;
	
\end{algorithm}

\begin{proof}(Of Theorem~\ref{thm:SchurSparse})
Note that the choices of $\epsilon$ and $s$ must ensure that
\begin{align*}
\frac{n^2\epsilon^2}{\amount} & = \delta\\
\frac{n^3}{\amount^2} & \leq \delta
\end{align*}

This is then equivalent to $\amount \geq n^{1.5}\delta^{-1}$ and $\frac{\amount}{\epsilon^2} = n^2\delta^{-1}$. This further implies that $\epsilon \geq n^{1/4}$.
Our $\epsilon$ and $\amount$ in $\textsc{SchurSparse}$ meet these conditions (and the ones specifically chosen in the algorithm will also be necessary for one of our applications).
The guarantees then follow from putting the quality of the sampler
from Lemma~\ref{lem:implicit-sampling} into the requirements of
the determinant preserving sampling procedure from Theorem~\ref{thm:SparsifyMain}. Additionally, Lemma~\ref{lem:implicit-sampling} only requires access to $2$-approximate leverage scores, which can be computed by Lemma~\ref{lem:ERDS} in $\tilde{O}(m)$ time. Furthermore, Lemma~\ref{lem:implicit-sampling} gives that our $\rho$ value is constant, and our assumption in Theorem~\ref{thm:SchurSparse} that we are given an $1.1$-DD subset $V_2$ implies that our expected $O(\amount \cdot \rho)$ calls to $\textsc{SampleEdgeSchur}$ will require $O(1)$ time.
The only other overheads are the computation and invocations 
of the various copies of approximate resistance data structures.
Since $m \leq n^2$ and $\epsilon \geq n^{1/4}$, Lemma~\ref{lem:ERDS}
gives that this cost is bounded by
$\tilde{O}(m + n^2 + \frac{\amount}{\epsilon^2}) = \tilde{O}(n^2\delta^{-1})$.
\end{proof}

	\section{Approximate Determinant of SDDM Matrices}
\label{sec:determinant_algo}


In this section, we provide an algorithm for computing an approximate determinant of SDDM matrices, which are minors of graph Laplacians
formed by removing one row/column.



Theorem~\ref{thm:SparsifyDeterminant} allows us to sparsify a dense graph while still approximately preserving the determinant of the graph minor. If there were some existing algorithm for computing the determinant that had good dependence on sparsity, we could achieve an improved runtime for determinant computation by simply invoking such an algorithm on a minor of the sparsified graph.\footnote{To get with high probability one could use standard boosting tricks involving taking the median of several estimates of the determinant obtained in this fashion.} Unfortunately, current determinant computation algorithms (that achieve high-accuracy) are only dependent on $n$, so simply reducing the edge count does not directly improve the runtime for determinant computation. Instead the algorithm we give will utilize Fact~\ref{fact:detMinor} 
\[
\detp{\left(\LL\right)}
= \det{\left(\LL_{[V_2, V_2]}\right)}
\cdot \detp{\left(\sc{\LL}{V_1}\right)}.
\] (where we recall that $\detp$ is the determinant of the matrix minor) to recursively split the matrix. Specifically, we partition the vertex set based upon the routine $\textsc{AlmostIndependent}$ from Lemma~\ref{lem:FindingAlphaDDSubsets}, then compute Schur complements according to $\textsc{SchurSparse}$ in Theorem~\ref{thm:SchurSparse}. Our algorithm will take as input a Laplacian matrix. However, this recursion naturally produces two matrices, the second of which is a Laplacian and the first of which is a submatrix of a Laplacian. Therefore, we need to convert
$\LL_{[V_2,V_2]}$ into a Laplacian. We do this by adding one vertex with
appropriate edge weights such that each row and column sums to $0$.
Pseudocode of this routine is in Algorithm~\ref{alg:addRowColumn},
and we call it with the parameters
$\LL^{V_2} \leftarrow \textsc{AddRowColumn}(\LL_{[V_2,V_2]})$.

\begin{algorithm}								
	\caption{$\textsc{AddRowColumn}(\MM):$ complete $\MM$
	into a graph Laplacian by adding one more row/column}
	\label{alg:addRowColumn}
	\SetAlgoVlined
	\SetKwProg{myproc}{Procedure}{}{}

	\KwIn{SDDM Matrix $\MM$}
	\KwOut{Laplacian matrix $\LL$ with one extra row / column than $\MM$}
	Let $n$ be the dimension of $\MM$\;
	\For{$i = 1$ to $n$}{
		Sum non-zero entries of row $i$, call $\ss_i$\;
		Set $\LL(n + 1,i), \LL(i, n + 1) \leftarrow -\ss_i$\;
	}
	Let $\LL(n + 1, n + 1) \leftarrow \sum_{i = 1}^n \ss_i$\;
	Output $\LL$\;
	
\end{algorithm}

The procedure $\textsc{AddRowColumn}$ outputs a Laplacian $\LL^{V_2}$ such that $\LL_{[V_2,V_2]}$ can be obtained if one removes this added row/column. This immediately gives $\detp(\LL^{V_2}) = \det(\LL_{[V_2,V_2]})$ by definition, and we can now give our determinant computation algorithm of the minor of a graph Laplacian.

\begin{algorithm}								
	\caption{$\textsc{DetApprox}(\LL, \delta, \overline{n}):$
			Compute $\detp(\LL)$ with error parameter $\delta$}
	\label{alg:detApprox}
	\SetAlgoVlined
	\SetKwProg{myproc}{Procedure}{}{}

	\KwIn{Laplacian matrix $\LL$, top level error threshold $\delta$,
			and top level graph size $\overline{n}$}
	\KwOut{Approximate $\detp(\LL)$}

	\uIf{this is the top-level invocation of this function in the recursion tree}{
	$\delta' \gets \Theta(\delta^2 / \log^3 n)$
	}
	\Else{
	$\delta' \gets \delta$
	}
	
	\If{$\LL$ is $2 \times 2$}{
	\Return{the weight on the (unique) edge in the graph}
	}
	
	$V_2 \gets \textsc{AlmostIndependent}(\LL, \frac{1}{10})$
	\hfill \{Via Lemma~\ref{lem:FindingAlphaDDSubsets}\}\\
	$V_1 \gets V \setminus V_2$ \;
	$\LL^{V_1} \gets \textsc{SchurSparse}(\LL,V_1, \delta')$;
	\hfill \{$\abs{V_1} / \overline{n}$ is the value of $\beta$ in Lemma~\ref{lem:detErrorBoundEachLevel}.\}\\
	$\LL^{V_2}\gets \textsc{AddRowColumn}(\LL_{[V_2,V_2]})$\;
	Output $\textsc{DetApprox}(\LL^{V_1}, \delta' \abs{V_1} / \overline{n}, \overline{n})
		 \cdot\textsc{DetApprox}(\LL^{V_2}, \delta' \abs{V_2} / \overline{n}, \overline{n}) $\;
	
\end{algorithm}

Our analysis of this recursive routine consists of bounding the
distortions incurred at each level of the recursion tree.
This in turn uses the fact that the number of vertices across all calls within a level and the total ``amount'' of $\delta$ across all calls within a level both remain unchanged from one level to the next.
This can be summarized by the following Lemma which bounds the error accumulated within one level of recursion in our algorithm.

\begin{lemma}\label{lem:detErrorBoundEachLevel}
	Suppose we are given some small $\delta \geq 0$ and non-negative $\beta_1,...,\beta_k$ such that $\sum_{i=1}^k \beta_i = O(1)$,  along with Laplacian matrices $\LL(1),
	\ldots, \LL(k)$ and each having a corresponding vertex partition $V_1(i), V_2(i)$, where
\begin{align*}
		\LL(i)
		& =
		\left[
		\begin{array}{cc}
		\LL\left(i\right)_{\left[V_1(i),V_1(i)\right]} &
		\LL\left(i\right)_{\left[V_1(i),V_2(i)\right]} \\
		\LL\left(i\right)_{\left[V_2(i),V_1(i)\right]} &
		\LL\left(i\right)_{\left[V_2(i),V_2(i)\right]}
		\end{array}
		\right].
\end{align*}
Let $\LL^{V_1(i)}$ denote the result of running
$\textsc{SchurSparse}$ to remove the $V_2(i)$ block in each of these matrices:\footnote{This Lemma only applies when the matrices are fixed with respect to the randomness used in the invocations of \textsc{SchurSparse} mentioned in the Lemma. In other words, it only applies when the result of running \textsc{SchurSparse} on each of these $\LL(i)$ matrices is independent of the result of running it on the other matrices. This is why the Lemma only immediately bounds error within a level of the recursion---where this independence holds---rather than for the entire algorithm.}
\[
\LL^{V_1\left(i\right)}
\defeq
\textsc{SchurSparse}\left(\LL(i),V_1(i),\beta_i\delta\right).
\]
Then conditioning upon a with high probability event\footnote{namely, the event that all the leverage score estimation calls to Lemma~\ref{lem:ERDS} from $\textsc{SchurSparse}$ succeed} in each of these
calls to $\textsc{SchurSparse}$, for any $p$ we have
with probability at least $1 - p$:
\[
\prod_{i=1}^k \detp\left(\LL\left(i\right)\right)
= \left(1 \pm O\left(\sqrt{\delta / p}\right) \right)
\prod_{i=1}^k \det\left(\LL_{\left[V_2(i),V_2(i)\right]}(i)\right)
\cdot \detp\left(\LL^{V_1\left(i\right)}\right).
\]
\end{lemma}

Here the $\beta_i$ corresponds to the $\abs{V_1} / n$ and $\abs{V_2} / n$ values
that $\delta$ is multiplied against in each call parameter
to \textsc{SchurSparse}.
An example of the main steps in this determinant
approximation algorithm, as well as the graphs corresponding
to applying Lemma~\ref{lem:detErrorBoundEachLevel} to one of
the layers is in Figure~\ref{fig:RecursiveDet}.

\begin{figure}[ht]

\begin{center}

\begin{tikzpicture}
		
\draw (0,4) ellipse (8 and 0.5);
\node at (0, 4) {G on $\overline{n}$ vertices with Laplacian $\LL$};
	
\draw (-3,1) ellipse (5 and 0.5);
\node at (-3, 1) {$\LL(1)$};
\draw[vecArrow] (-3, 3.5) to (-3, 1.5);
\node at (-3, 2.5) [fill = white]{$\textsc{SchurSparse}
(\LL^{G}, V_1, \beta \delta)$};

\draw (5,1) ellipse (3 and 0.5);
\node at (5, 1) {$\LL(2)$};
\draw[vecArrow] (5, 3.5) to (5, 1.5);
\node at (5, 2.5) [fill = white]{$\textsc{AddRowColumn}(\LL^{G}_{[V_2, V_2]})$};

\draw[vecNarrowArrow] (-6, 0.5) to (-6, -1.5);
\draw (-6, -2) ellipse (2 and 0.5);
\node at (-6, -2) {$\LL^{V_1(1)}$, $\abs{V_1(1)} = \beta_1 n$};
\node at (-5, -0.5) [fill = white]{$\textsc{SchurSparse}
	(\LL(1), V_1(1), \beta_1 \delta)$};

\draw[vecNarrowArrow] (-1.5, 0.5) to (-1.5, -1.5);
\draw (-1, -2) ellipse (3 and 0.5);
\node at (-1, -2) {$\LL^{V_2(1)}$};

\draw[vecNarrowArrow] (4, 0.5) to (4, -1.5);
\draw (4, -2) ellipse (2 and 0.5);
\node at (4, -2) {$\LL^{V_1(2)}$, $\abs{V_1(2)} = \beta_2 n$};
\node at (4, -0.5) [fill = white]{$\textsc{SchurSparse}
	(\LL(2), V_1(2), \beta_2 \delta)$};

\draw[vecNarrowArrow] (7, 0.5) to (7, -1.5);
\draw (7, -2) ellipse (1 and 0.5);
\node at (7, -2) {$\LL^{V_2(2)}$};
		
\end{tikzpicture}
\end{center}
	
\caption{Two layers of the call Structure of
		the determinant approximation algorithm \textsc{DetApprox}
		 (algorithm~\ref{alg:detApprox}),
	with the transition from the first to the second
	layer labeled as in Lemma~\ref{lem:detErrorBoundEachLevel}.}
	\label{fig:RecursiveDet}
\end{figure}
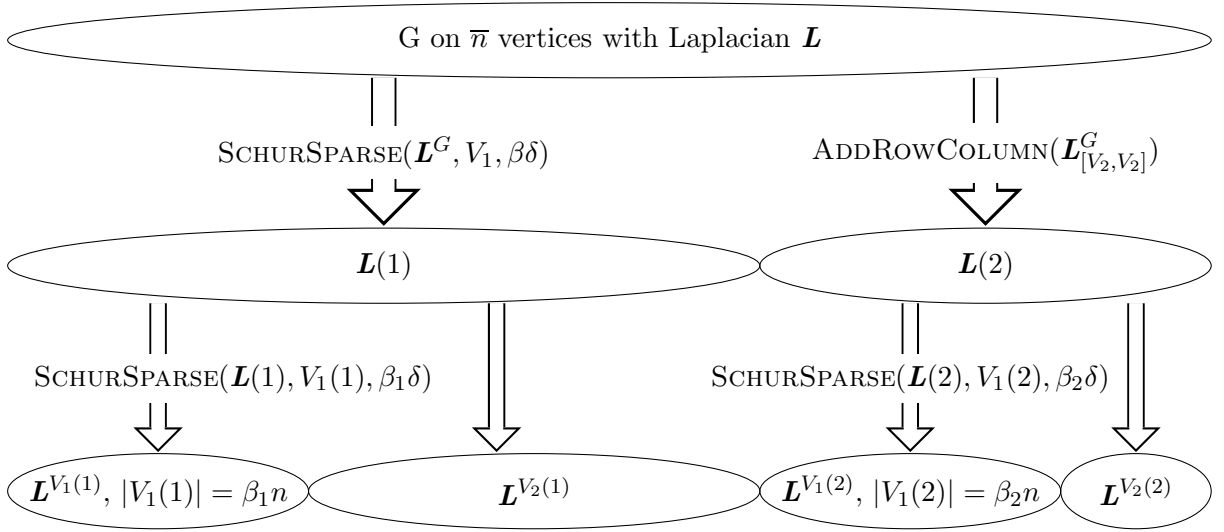

Applying Lemma~\ref{lem:detErrorBoundEachLevel} to all the layers
of the recursion tree gives the overall guarantees.

\begin{proof}[Proof of Theorem~\ref{thm:detApproxIntro}]
~~
	
\textbf{Running Time:}
Let the number of vertices and edges in the current graph corresponding
to $\LL$ be $n$ and $m$ respectively.
Calling $\textsc{AlmostIndependent}$ takes expected time ${O}(m)$
and guarantees
\[
\frac{n}{16} \leq \abs{V_2} \leq \frac{n}{8},
\]
which means the total recursion terminates in $O(\log{n})$ steps.

For the running time,
note that as there are at most $O(n)$ recursive calls,
the total number of vertices per level of the recursion is $O(n)$.
The running time on each level are also dominated by the calls
to $\textsc{SchurSparse}$, which comes out to
\[
\tilde{O}\left( \abs{V_1\left(i\right)}^2 \frac{n}{\delta' \abs{V_1\left(i\right)}} \right)
= \tilde{O}\left( \abs{V_1\left( i\right)} n \delta^{-2} \right),
\]
and once again sums to $\tilde{O}(n^2 \delta^{-2})$.
We note that this running time can also be obtained from more standard
analyses of recursive algorithms,
specifically applying guess-and-check to a running
time recurrence of the form of:
\[
T\left(n, \delta\right)
= T\left(\theta n, \theta \delta\right)
+ T\left(\left(1 - \theta \right)n + 1, \left( 1 - \theta \right) \delta\right)
+ \tilde{O}(n^2\delta^{-1}).
\]
	
	
	\textbf{Correctness}. As shown in the running time analysis, our recursion tree has depth at most $O(\log{n})$, and there are at most $O(n)$ total vertices at any given level. We associate each level of the recursion in our algorithm with the list of matrices which are given as input to the calls making up that level of recursion. For any level in our recursion, consider the product of $\detp$ applied to each of these matrices. We refer to this quantity for level $j$ as $q_j$. Notice that $q_0$ is the determinant we wish to compute and $q_{\text{\# levels}-1}$ is what our algorithm actually outputs. As such, it suffices to prove that for any $j$, $q_j = (1 \pm \frac{\delta}{\text{\# levels}})q_{j-1}$ with probability of failure at most $\frac{1}{10 \cdot \text{\# levels}}$. However, by the fact that we set $\delta' = \Theta(\delta^2 / \log^3 n)$ in the top level of recursion with sufficiently small constants, this immediately follows from Lemma~\ref{lem:detErrorBoundEachLevel}.
	
A minor technical issue is that Lemma~\ref{lem:detErrorBoundEachLevel} only gives guarantees conditioned on a WHP event. However, we only need to invoke this Lemma a logarithmic number of times, so we can absorb this polynomially small failure probability into the our total failure probability without issue.
	
Standard boosting techniques---such as running $O(\log{n})$ independent
instances and taking the medians of the estimates---
give our desired with high probability statement.
\end{proof}

It remains to bound the variances per level of the recursion.

\begin{proof}(Of Lemma~\ref{lem:detErrorBoundEachLevel})
As a result of Fact~\ref{fact:detMinor}	
\[
\prod_{i=1}^k \detp\left(\LL(i)\right)\
= \prod_{i=1}^k \det\left(\LL\left(i\right)_
	{\left[V_2\left(i\right),V_2\left(i\right)\right]}\right)
\detp\left(\sc{\LL\left(i\right)}{V_1\left(i\right)}\right).
\]
Consequently, it suffices to show that
with probability at least $1 - p$
\[
\prod_{i=1}^k
\detp\left(\sc{\LL\left(i\right)}{V_1\left(i\right)}\right)
= \left(1 \pm O\left( \sqrt{\delta / p}\right)\right)
\prod_{i=1}^k \detp \left(\LL^{V_1\left(i\right)} \right).
\] 
Recall that $\LL^{V_1(i)}$ denotes the random variable that is
the approximate Schur complement generated through
the call to $\textsc{SchurSparse}(\LL(i), V_1(i), \beta_i \delta)$.

Using the fact that our calls to $\textsc{SchurSparse}$ are independent along with the assumption of $\sum_{i=1}^k \beta_i = O(1)$, we can apply the guarantees of Theorem~\ref{thm:SchurSparse} to obtain
\begin{multline*}
\expec{\LL^{V_1\left(1\right)} \ldots \LL^{V_1\left(k\right)}}{
\prod_{i=1}^k \detp\left(\LL^{V_1\left(1\right)}\right)}
= \prod_{i=1}^{k} \expec{\LL^{V_1\left(i\right)}}
{\detp\LL^{V_1\left(i\right)}}
\\ = \left(1 \pm O\left({\delta}\right)\right)
\prod_{i=1}^k \detp\left(\sc{\LL(i)}{V_1(i)}\right),
\end{multline*}
and 
\begin{multline*}
\frac{
\expec{\LL^{V_1\left(1\right)} \ldots \LL^{V_1\left(k\right)}}
{\prod_{i=1}^k \detp\left(\LL^{V_1\left(i\right)}\right)^2}}
{\expec{\LL^{V_1\left(1\right)} \ldots \LL^{V_1\left(k\right)} }
{\prod_{i=1}^k \detp\left(\LL^{V_1\left(i\right)}\right)^2}
}
= \prod_{i=1}^k
\frac{
\expec{\LL^{V_1\left(i\right)}}{ \detp\left(\LL^{V_1\left(i\right)}\right)^2}}
{\expec{\LL^{V_1\left(i\right)}}{ \detp\left(\LL^{V_1\left(i\right)}\right)}^2
}
\\
\leq \prod_{i=1}^k \exp{\left(O\left(\beta_i \delta\right)\right)}
\leq \exp{\left(O\left(\delta\right)\right)}.
\end{multline*}
	
By assumption $\delta$ is small,
so we can approximate $\exp{\left(O(\delta)\right)}$
with $1 + O(\delta)$, which with bound above gives 
\[
\var{\LL^{V_1\left(1\right)} \ldots \LL^{V_1\left(k\right)}}
{\prod_{i=1}^k \detp\left(\LL^{V_1\left(i\right)}\right)}
\leq O\left(\delta\right)
 \expec{\LL^{V_1\left(1\right)} \ldots \LL^{V_1\left(k\right)}}{\prod_{i=1}^k \detp\left( \LL^{V_1\left(i\right)} \right)}^2,
\]
Then applying the approximation on $\expec{}{\prod_{i=1}^k \detp\left(\textsc{SchurSparse}(\LL(i),V_1(i),\beta_i\delta)\right)}$ gives 
\[
\var{\LL^{V_1\left(1\right)} \ldots \LL^{V_1\left(k\right)}}{
\prod_{i = 1}^{k} \detp\left(\LL^{V_1\left(i\right)}\right)}
\leq O\left(\delta\right) 
	\left(\prod_{i=1}^k \detp\left(\sc{\LL(i)}{V_1(i)}\right) \right)^2.
\]
At which point we can apply
Chebyshev's inequality to obtain our desired result.
	
\end{proof}


\section{Random Spanning Tree Sampling}
\label{sec:spanning_tree}

In this section we will give an algorithm for generating a random spanning tree from a weighted graph, that uses $\textsc{SchurSparse}$ as a subroutine, and ultimately prove Theorem~\ref{thm:spanningTreeAlgoIntro}.


In order to do so, we will first give an $O(n^{\omega})$ time recursive algorithm using Schur complement that exactly generates a random tree from the $\ww$-uniform distribution. The given algorithm is inspired by the one introduced in \cite{ColbournDN89}, and its variants utilized in  \cite{ColbournMN96,HarveyX16,DurfeeKPRS16}. However, we will (out of necessity for our further extensions) reduce the number of branches in the recursion to two, by giving an efficient algorithmic implementation of a bijective mapping between spanning
trees in $G$ and spanning trees in $\sc{G}{V_2}$ when $V_1$,
the set of vertices removed, is an independent set.
We note that this also yields an alternative algorithm for
generating random spanning trees from the $\ww$-uniform
distribution in $O(n^{\omega})$ time.


The runtime of this recursion will then be
achieved similar to our determinant algorithm.
We reduce $\delta$ proportional to the decrease in the number
of vertices for every successive recursive call in exactly
the same was as the determinant approximation algorithm from
Section~\ref{sec:determinant_algo}.
As has been previously stated and which is proven in Section~\ref{subsec:EasierTVBound}, drawing a random spanning
tree from a graph after running our sparsification routine
 which takes $\widetilde{O}(n^2\delta^{-1})$,
will have total variation distance $\sqrt{\delta}$
from the $\ww$-uniform distribution.

Similar to our analysis of the determinant algorithm,
we cannot directly apply this bound to each tree because
the lower levels of the recursion will contribute far
too much error when $\delta$ is not decreasing at a
proportional rate to the total variation distance.
Thus we will again need to give better bounds on the
variance across each level, allowing stronger bounds on the contribution to total variation distance of the entire level.

This accounting for total variance is more difficult here
due to the stronger dependence between the recursive calls.
Specifically, the input to the graph on $V_2$ depends on
the set of edges chosen in the first recursive call on
$V_1$, specifically $\sc{G}{V_1}$, or a sparsified version of it.

Accounting for this dependency will require proving additional
concentration bounds shown in Section~\ref{sec:cond_conc},
which we specifically achieve by sampling $\amount = O(n^2\delta^{-1})$
edges in each call to $\textsc{SchurSparse}$.
While this might seem contradictory to the notion of ``sampling",
we instead consider this to be sampling from the graph in which
all the edges generated from the Schur complement are kept separate
and could be far more than $n^2$ edges.

\subsection{Exact $O(n^{\omega})$ Time Recursive Algorithm}
\label{subsec:ExactSpanningAlgo}

We start by showing an algorithm that samples trees from
the exact $\ww$-uniform distribution via the computation
of Schur complements.
Its pseudocode is in Algorithm~\ref{alg:exactTree},
and it forms the basis of our approximate algorithm:
the faster routine in Section~\ref{subsec:RandSpanningTreeAlgo}
is essentially the same as inserting sparsification steps
between recursive calls.

\begin{algorithm}		
	\caption{$\textsc{ExactTree}(G):$ Take a graph and output a tree randomly from the $\ww$-uniform distribution}
	\label{alg:exactTree}
	\SetAlgoVlined
	\SetKwProg{myproc}{Procedure}{}{}

	\KwIn{Graph $G$}
	\KwOut{A tree randomly generated from the $\ww$-uniform distribution of $G$}
	If there is only one edge $e$ in $G$, return $G$ \;
	Partition $V$ evenly into $V_1$ and $V_2$\;
	$T_1 = \textsc{ExactTree}(\schur(G,V_1))$\;
	\For{each $e \in T_1$}{
		with probability $\frac{\ww_e(G)}{\ww_e(\schur(G,V_1))}$,
		$G \gets G/e$, $T \gets T \cup e$
	\label{ln:ExactTreeSplit}\;}
	Delete the remaining edges, i.e., $G \gets G \setminus E(V_1)$\;
	$T_2 = \textsc{ExactTree}(\schur(G,V_2))$\;
	$T \gets T \cup \textsc{ProlongateTree}(G, V_1 \sqcup V_2, T_2)$\;	
	Output $T$\;
\end{algorithm}

The procedure $\textsc{ProlongateTree}$ is invoked when
$V_1 = V \setminus V_2$ maps a tree $T_2$ from the
Schur complement $\sc{G}{V_2}$ to a tree back in $G$.
It crucially uses the property that $V_1$ is an independent
set, and its pseudocode is given in Algorithm~\ref{alg:ProlongateTree}.

\begin{algorithm}				
\caption{$\textsc{ProlongateTree}(G, V_1 \sqcup V_2, T_2)$:
prolongating a tree on $\sc{G}{V_2}$ to a tree on $G$.
}
\label{alg:ProlongateTree}

\SetAlgoVlined
\SetKwProg{myproc}{Procedure}{}{}
	
\KwIn{A graph $G$, a splitting of vertices $V_1 \sqcup V_2$
such that $V_1$ is an independent set, tree $T_2$ of $\sc{G}{V_2}$.}
\KwOut{A tree in $G$}

$T \gets \emptyset$\;

\For{each $e = xy \in T_2$}{
	Create distribution $\lambda_e$,
		set $\lambda_e(\emptyset) = \ww_e(G)$\;
	\For{each $v \in V_1$ such that $(v, x), (v, y) \in E(G)$}{
		Set $\lambda_e(v) = \ww_{(v, x)}(G) \ww_{(v, y)}(G) \deg_v(G)^{-1}$\;
	}
	Randomly assign $f(e)$ to $\{\emptyset \cup V_1\}$ with probability
		proportional to $\lambda$\;
}

\For{each $v \in V_1$}{
	\For{ each $e = (x,y) \in T_2$ such that $(v,x),(v,y) \in E(G)$}
			{\If{$f(e) \neq v$}{Contract $x$ and $y$ \;
			}
		}
		\For{each contracted vertex $X$ in the neighborhood of $v$}{
		Connect $X$ to $v$ with edge $(v,u) \in G$ with probability
		proportional to $w_G((v,u))$ \label{step:ProlongateTree:SampleStar}\;
		$T \gets T \cup (v,u)$\;
	}
}
	
Output $T$\;
	
\end{algorithm}

\begin{lemma}\label{lem:ExactTree}
The procedure $\textsc{ExactTree}(G)$ will generate a random tree
of $G$ from the $\ww$-uniform distribution in $O(n^\omega)$ time.
\end{lemma}

The algorithm we give is similar to the divide and
conquer approaches of~\cite{ColbournDN89,ColbournMN96,HarveyX16,DurfeeKPRS16}. The two main facts used by these approaches can be summarized as follows:
\begin{enumerate}
\item Schur complements preserves the leverage score of original edges,
and
\item The operation of taking Schur complements,
and the operation of deleting or contracting an edge are associative.
\end{enumerate}

We too will make use of these two facts.
But unlike all previous approaches, at every stage
we need to recurse on only two sub-problems.
All previous approaches have a branching factor of at least four. 

We can do this by exploiting the structure of the Schur complement when one eliminates an independent set of vertices.
We formalize this in Lemma~\ref{lem:ProlongateTree}.

Before we can prove the lemma, we need to state an important property of
Schur complements that follows from Fact~\ref{fact:schurPreservesProb}.
Recall the notation from Section~\ref{sec:background} that for a weighted
graph $G = (V, E, \ww)$, $\Pr_{T}^{G}(\cdot)$ denotes the probability of
$\cdot$ over trees $T$ picked from the $\ww$-uniform distribution on
spanning trees of $G$.

\begin{lemma}
\label{lem:SamplingSubset}
Let $G$ be a graph with a partition of vertices
$V = V_1 \sqcup V_2$.
Then for any set of edges $F$ contained in $G[V_1]$,
the induced subgraph on $V_1$, we have:
\[
\Pr_T^{G}\left( T \cap E\left(G\left[V_1\right]\right) = F \right)
= \Pr_T^{\schur{(G,V_1)}}
	\left(T \cap E\left(G\left[V_1\right]\right) = F \right),
\]
where the edges in $\sc{G}{V_1}$ are treated as the sum
of $G[V_1]$ and $G_{sc}[V_1]$, the new edges added to the
Schur complement.
\end{lemma}

\begin{proof}
If $F$ contains a cycle, then $\Pr_T^{G}\left( T \cap E\left(G\left[V_1\right]\right) = F \right) =0= \Pr_T^{\schur{(G,V_1)}}(T \cap E\left(G\left[V_1\right]\right)=F)$. Therefore, we will assume $F$ does not contain any cycle, and we will prove  by induction on the size of $F$. If $|F| > |V_1|-1,$ then $F$ will have to contain a cycle. When  $|F| = |V_1|-1,$ then $F$ will have to be the edge set of a tree in $\schur{(G,V_1)}.$  Then by Fact~\ref{fact:schurPreservesProb}, the corollary holds.
Now suppose that the corollary holds for all $F$ with $|F| = |V_1| - 1 -k$.
Now consider some $F$ with $|F| = |V_1| - 1 - (k+1)$.
We know
\[
\Pr^{G}_{T}\left(F \subseteq T\right)
= \Pr^{G}_{T}\left(F = \left(T \cap E\left(G\left[V_1\right]\right)\right)\right)
+ \sum_{F' \supset F} \Pr^{G}_{T}
	\left(F' = \left(T \cap E\left(G\left[V_1\right]\right)\right)\right).
\]	
Since $\abs{F'} > \abs{F}$, by assumption 
\[
\sum_{F' \supset F} \Pr_T^{G}
	\left(F' = \left(T \cap E\left(G\left[V_1\right]\right)\right)\right)
= \sum_{F' \supset F} \Pr_{T}^{\schur\left(G,V_1\right)}
	\left(F' = \left(T \cap E\left(G\left[V_1\right]\right)\right)\right),
\]
then by Fact~\ref{fact:schurPreservesProb} we have
$\Pr_{T}^{G}(F \subseteq T) = \Pr_{T}^{\schur(G,V_1)}(F \subseteq T)$,
which implies
 \[
 \Pr_T^{G}\left(F  = \left(T \cap E\left(G\left[V_1\right]\right)\right)\right)
 = \Pr_{T}^{\schur{\left(G,V_1\right)}}
	 \left(F= \left(T \cap E\left(G\left[V_1\right]\right)\right)\right).
 \]
	
\end{proof}

The tracking of edges from various layers of the Schur complement
leads to another layer of overhead in recursive algorithms.
They can be circumvented by merging the edges, generating a
random spanning tree, and the `unsplit' the edge by
random spanning.
The following is a direct consequence of the definition of $\ww(T)$:
\begin{lemma}
\label{lem:UnSplit}
Let $\widehat{G}$ be a multi-graph, and $G$ be the simple
graph formed by summing the weights of overlapping edges.
Then the procedure of:
\begin{enumerate}
\item Sampling a random spanning tree from $G$, $T$.
\item For each edge $e \in T$, assign it to an original edge
from $\widehat{G}$, $\widehat{e}$ with probability
\[
\frac{\ww_{\widehat{e}}\left(\widehat{G} \right)}{\ww_{e}\left( G \right)}.
\]
\end{enumerate}
Produces a $\ww$-uniform spanning tree from $\widehat{G}$,
the original multi-graph.
\end{lemma}


This then leads to the following proto-algorithm:
\begin{enumerate}
\item Partition the vertices (roughly evenly) into
\[
V = V_1 \sqcup V_2.
\]
\item Generate a $\ww$-uniform tree of $\sc{G}{V_1}$,
and create $F_1 = T \cap E(G[V_1])$ by re-sampling
edges in $G[V_1]$ using Lemma~\ref{lem:UnSplit}.
By Lemma~\ref{lem:SamplingSubset}, this subset is precisely
the intersection of a random spanning tree with $G[V_1]$.
\item This means we have `decided' on all edges in $G[V_1]$.
So we can proceed by contracting all the edges of $F_1$,
and deleting all the edges corresponding to $E(G[V_1]) \backslash F$.
Let the resulting graph be $G'$ and let $V_1'$ be the remaining vertices
in $V_1$ after this contraction.
\item Observe that $V_1'$ is an independent set, and its
complement is $V_2$.
We can use another recursive call to generate a
$\ww$-uniform tree in $\schur(G',V_2)$.
Then we utilize the fact that $V_1'$ is an independent set to lift
this to a tree in $G'$ efficiently via Lemma~\ref{lem:ProlongateTree}.
\end{enumerate}

Our key idea for reducing the number of recursive calls of the algorithm,
that when $V_1$ (from the partition of vertices $V = V_1 \sqcup V_2$)
is an independent set, we can directly lift a tree from $\schur(G,V_2)$
to a tree in $G$.
This will require viewing $G_{\schur}[V_2]$ as a sum of cliques,
one per vertex of $V_1$, plus the original edges in $G[V_2]$.

\begin{fact}\label{fact:schurCompleteNeighborGraph}
Given a graph $G$ and a vertex $v$, the graph
$\schur(G,V\setminus v)$ is the induced graph $G[V \setminus \{v\}]$
plus a weighted complete graph $K(v)$ on the neighbors of $v$.
This graph $K(v)$ is formed by adding one edge $xy$ for every
pair of $x$ and $y$ incident to $v$ with weight
\[
\frac{\ww_{(v,x)}\ww_{(v,y)}}{deg_v},
\]
where $\deg_v \defeq \sum_{x} \ww_{(v, x)}$
is the weighted degree of $v$ in $G$.	
\end{fact}

\begin{lemma}\label{lem:ProlongateTree}
Let $G$ be a graph on $n$ vertices and $V_1$ an independent set.
If $T$ is drawn from the $\ww$-uniform distribution of $\schur(G,V_2)$,
then in $O(n^2)$ time $\textsc{ProlongateTree}(G, V_1 \sqcup V_2, T_2)$ returns a
tree from the $\ww$-uniform distribution of $G$.	
\end{lemma}

\begin{proof}
The running time of $\textsc{ProlongateTree}$ is $O(n^2)$ as $T_2$ has $\leq n-1$ edges and  $|V_1|\leq n.$ 	
	
Now we will show the correctness. Let $V_1= \{v_1,...,v_k\}$.
We will represent $\schur(G,V_2)$ as a multi-graph arising by Schur
complementing out the vertices in $V_1$ one by one and keeping the new
edges created in the process separate from each other as a multi-graph.
We represent this multi-graph as
\[
\schur(G,V_2)
=
G\left[V_2\right] + K\left(v_1\right) + ... + K\left(v_k\right),
\]
where $G[V_2]$ is the induced subgraph on $V_2$ and $K(v_i)$
is the weighted complete graph on the neighbors of $v_i$.
Then
\begin{itemize}
\item By the unsplitting procedure from Lemma~\ref{lem:UnSplit},
the function $f$ maps $T_2$ to a tree in the multi-graph
$G[V_2] + K(v_1) + ... + K(v_k)$, and
\item the rest of the sampling steps maps this tree to one in $G$.
\end{itemize}

We will now prove correctness by induction on the size of
the independent set $V_1$.
The case of $|V_1| = 0$ follows from $\sc{G}{V_2} = G$.
If $|V_1| = 1$, i.e, $V_1 =\{ v \}$ for some vertex $v$,
then $\schur(G, V_2)$ is $G[V_2] + K(v)$.
Given a tree $T_2$ of $\schur(G, V_2)$, the creation of $f$
will first map $T_2$ to a tree in the multigraph
$G[V_2] + K(v)$ by randomly deciding for each edge $e \in T$
to be in $G(V_1)$ or $K(v)$ depending on it's weight.
If we let $T'(V_2) = T' \cap G[V_2]$, then by Lemma~\ref{lem:SamplingSubset},
\[
\Pr_T^{G} \left(T \cap E({G\left[V_2\right]}) = T'\left(V_2\right)\right)
= \Pr_T^{G\left[V_2\right] + K\left(v\right)}
\left(T \cap E({G\left[V_2\right]}) = T'\left(V_2\right)\right).
\]
Therefore, we can contract all the edges of $T'(V_2) \cap G\left[V_2\right]$
and delete all other edges of $G\left[V_2\right]$.
This results in a multi-graph star with $v$ at the center.
Now, $\textsc{ProlongateTree}$ does the following
to decide on the remaining edges.
For every multi-edge of the star graph obtained by contracting
or deleting edges in $G[V_2]$, we choose exactly one edge,
randomly according to its weight.
This process generates a random tree of multi-graph star.

Now we assume that the lemma holds for all $V_1'$
with $\abs{V_1'} < k$. 
Let $V_1 = \{v_1,...,v_k\}$.
The key thing to note is that when $V_1$ is an independent set,
we can write
\[
\sc{G}{V_2}
= G\left[V_2\right] + K\left(v_1\right) + \ldots + K\left(v_k\right),
\]
and
\[
\sc{G}{V_2 \cup v_k}
= G\left[V_2 \cup v_k\right] + K\left(v_1\right)
+ \ldots + K\left(v_{k-1}\right).
\]
Therefore, by the same reasoning as above,
we can take a random tree $T'$ of the multi-graph
$G[V_2] + K(v_1) + ... + K(v_k)$ and map it to a tree on
$G[V_2 \cup v_k] + K(v_1) + ... + K(v_{k-1})
= \sc{G}{V_2 \cup v_k}$
by our procedure $\textsc{ProlongateTree}$.
We then apply our inductive hypothesis on the set $V_1 \setminus \{v_k\}$
to map $\schur{(G,V_2 \cup v_k)}$ to a tree of $G$
by $\textsc{ProlongateTree}$, which implies the lemma.

\end{proof}

We also remark that the running time of \textsc{ProlongateTree}
can be reduced to $O(m \log{n})$ using dynamic trees,
which can be abstracted as a data structure supporting
operations on rooted forests~\cite{SleatorT85,AlstrupHLT05}.
We omit the details here as this does not bottleneck the running time.

With this procedure fixed, we can now show the overall guarantees
of the exact algorithm.
\begin{proof} {of Lemma~\ref{lem:ExactTree}}
Correctness follows immediately from
Lemmas~\ref{lem:SamplingSubset}~and~\ref{lem:ProlongateTree}.
The running time of $\textsc{ProlongateTree}$ is $O(n^2)$ and
contracting or deleting all edges contained in $G[V_1]$ takes $O(m)$ time.
Note that in this new contracted graph, the vertex set containing $V_1$
is an independent set.
Furthermore, computing the Schur complement takes $O(n^\omega)$ time,
giving the running time recurrence
\[
T\left( n \right)
= 2 T\left( n / 2 \right) + O\left(n^\omega\right)
= O\left(n^\omega \right).
\]
\end{proof}


\subsection{Fast Random Spanning Tree Sampling using Determinant Sparsification of Schur complement}
\label{subsec:RandSpanningTreeAlgo}

Next, we note that the most expensive operation from the exact
sampling algorithm from Section~\ref{subsec:ExactSpanningAlgo}
was the Schur complement procedure.
Accordingly, we will substitute in our sparse Schur complement
procedure to speed up the running time.

However, this will add some complication in applying
Line~\ref{ln:ExactTreeSplit} of \textsc{ExactTree}.
To address this, we need the observation that the
$\textsc{SchurSparse}$ procedure can be extended to
distinguish edges from the original graph, and the
Schur complement in the multi-graph that it produces.
\begin{lemma}
\label{lem:ApproxSchurSplit}
The procedure $\textsc{SchurSparse}(G, V_1, \delta)$ given in
Algorithm~\ref{algo:SchurSparse} can be modified to record whether an edge in
its output, $H^{V_1}$ is a rescaled copy of an edge from the original
induced subgraph on $V_1$, $G[V_1]$, or one of the new edges generated
from the Schur complement, $G_{SC}(V_1)$.
\end{lemma}

\begin{proof}
The edges for $H^{V_1}$ are generated by the
random walks via $\textsc{SampleEdgeSchur}(G, V_1)$, whose
pseudocode is given in Algorithm~\ref{algo:SampleEdgeSchur}.
Each of these produces a walk between two vertices in $V_1$,
and such a walk belongs to $G[V_1]$ if it is length $1$, 
and $G_{SC}(V_1)$ otherwise.
\end{proof}

We can now give our algorithm for generating random
spanning trees and prove the guarantees that lead to the main
result from Theorem~\ref{thm:spanningTreeAlgoIntro}.

\begin{algorithm}								
	\caption{$\textsc{ApproxTree}(G,\delta, \overline{n})
		$ Take a graph and output a tree randomly from a distribution $\delta$-close to the $\ww$-uniform distribution}
	\label{alg:approxTree}
	\SetAlgoVlined
	\SetKwProg{myproc}{Procedure}{}{}

	\KwIn{Graph $G$, error parameter $\delta$, and initial number of vertices $\overline{n}$}
	\KwOut{A tree randomly generated from a distribution $\delta$-close to the $\ww$-uniform distribution of $G$}
	$V_2 \gets \textsc{AlmostIndependent}(G, \frac{1}{10})$;
	\hfill \{Via Lemma~\ref{lem:FindingAlphaDDSubsets}\}\\
	$H_1 \gets \textsc{SchurSparse}(G ,V_1, \delta \cdot |V_1| / \overline{n}))$,
		while tracking whether the edge is from $G[V_1]$ via
		the modifications from Lemma~\ref{lem:ApproxSchurSplit} \label{ln:ApproxTreeG1}\;
	$T_1 = \textsc{ApproxTree}(H_1, \delta, \overline{n})$\;
	$G' \gets G$~\label{ln:GPrimeStart}\;
	\For{each $e \in T_1$}{
		\If{\label{ln:ApproxTreeSplit}$\textsc{RAND}[0, 1]
	\leq \ww_e^{ori}(G_1)  / \ww_{e}(G_1)$}
	{ 
\nonl
\hfill \{$\ww_e^{ori}(G_1)$ is  calculated using the weights tracked from Line~\ref{ln:ApproxTreeG1})\;
		$G' \gets G'/ \{e\}$~\label{ln:GPrimeContract}\;
		$T \gets T \cup \{e\}$\;
	}
	}
	Delete all edges between (remaining) vertices in $V_1$ in $G'$,
		$G' \gets G' \setminus E(G'[V_1])$~\label{ln:GPrimeEnd}\;
	$H_2 \gets \textsc{SchurSparse}(G',V_2, \delta \cdot |V_2| / \overline{n})$  \label{ln:ApproxTreeG2}\; 
	$T_2 = \textsc{ApproxTree}(H_2, \delta, n)$\;
	$T \gets T \cup \textsc{ProlongateTree}(G, V_1 \sqcup V_2, T_2)$
	\label{ln:MapT2Back}\;	
	Output $T$\;
	
\end{algorithm}

Note that the splitting on Line~\ref{ln:GPrimeContract} is
mapping $T_1$ first back to a tree on a the sparsified multi-graph
of $\sc{G}{V_1}$: where the rescaled edges that originated from
$G[V_1]$ are tracked separately from the edges that arise from
new edges involving random walks that go through vertices in $V_2$.

The desired runtime will follow equivalently to the analysis of the determinant algorithm in Section~\ref{sec:determinant_algo} as we are decreasing $\delta$ proportionally to the number of vertices. It remains to bound the distortion to the spanning tree distribution
caused by the calls to \textsc{SchurSparse}.

Bounds on this distortion will not follow equivalently to that of the determinant algorithm, which also substitutes \textsc{SchurSparse} for exact Schur complements, due to the dependencies in our recursive structure. In particular, while the calls to \textsc{SchurSparse} are independent, the graphs that they are called upon depend on the randomness in Line~\ref{ln:ApproxTreeSplit} and \textsc{ProlongateTree}, which more specifically, are simply the resulting edge contractions/deletions in previously visited vertex partitions within the recursion. Each subgraph \textsc{SchurSparse} is called upon is additionally dependent on the vertex partitioning from \textsc{AlmostIndependent}.

The key idea to our proof will then be a layer-by-layer analysis of distortion incurred by $\textsc{SchurSparse}$ at each layer to the probability of sampling a \textit{fixed} tree. By considering an alternate procedure where we consider exactly sampling a random spanning tree after some layer, along with the fact that our consideration is restricted to a \textit{fixed} tree, this will allow us to separate the randomness incurred by calls to $\textsc{SchurSparse}$ from the other sources of randomness mentioned above. Accordingly, we will provide the following definition.

\begin{definition}
	\label{lem:Truncated}
	For any $L \geq 0$, the \textbf{level-$L$ truncated algorithm}
	is the algorithm given by modifying
	$\textsc{ApproxTree}(G, \delta, \overline{n})$
	so that all computations of sparsified Schur complements
	are replaced by exact calls to Schur complements
	(aka. $\sc{G}{V_1}$ or $\sc{G'}{V_2})$) after level $l$.
	
	The tree distribution $\mathcal{T}^{(L)}$ is defined
	as the output of the level-$L$ truncated algorithm.
\end{definition}

Note that in particular, $\mathcal{T}^{(0)}$ is the tree distribution
produced by $\textsc{ExactTree(G)}$, or the $\ww$-uniform distribution;
while $\mathcal{T}^{(O(\log{n}))}$ is the distribution outputted
by $\textsc{ApproxTree}(G, \delta)$.

The primary motivation of this definition is that we can separate the randomness between $\mathcal{T}^{(l)}$ and $\mathcal{T}^{(l+1)}$ by only the calls to $\textsc{SchurSparse}$ at level $l+1$, which will ultimately give the following lemma that we prove at the end of this section 

\begin{lemma}
	\label{lem:DistPerLayer}
	For an invocation of $\textsc{ApproxTree}$ on a graph $G$
	with variance bound $\delta$, for any layer $L > 0$, we have
	\[
	d_{TV}\left( \mathcal{T}^{\left( L - 1 \right)},
	\mathcal{T}^{\left( L \right)} \right) \leq O(\sqrt{\delta}).
	\]
\end{lemma}

To begin, we consider the differences between
$\mathcal{T}^{(0)}$ and $\mathcal{T}^{(1)}$ and the probability
of sampling a fixed tree $\hatT$ on a recursive call on $G$.
The most crucial observation is that the two recursive calls
to $\textsc{ApproxTree}(G_1, \delta, \overline{n})$
and $\textsc{ApproxTree}(G_2, \delta, \overline{n})$
can be viewed as independent:
\begin{claim}
\label{ln:Independent}
For a call to $\textsc{ApproxTree}(G, \delta, \overline{n})$
(Algorithm~\ref{alg:approxTree})
to return $\hatT$, there is only one possible choice of
$G'$ as generated via Lines~\ref{ln:GPrimeStart} to~\ref{ln:GPrimeEnd}. 
\end{claim}

\begin{proof}
Note that the edges removed from Line~\ref{ln:GPrimeContract} are
precisely the edges in $T$ with both endpoints contained in $V_1$,
$E(T[V_1])$.
For a fixed $\hatT$, this set is unique, so $G'$ is unique as well.
\end{proof}

This allows us to analyze a truncated algorithm by splitting the
probabilities into those that occur at level $l$ or above.
Specifically, at the first level, this can be viewed as
pairs of graphs $\sc{G}{V_1}$ and $\sc{G}{V_2}$ along with
the `intended' trees from them: 
\begin{definition}
\label{def:PAboveOneLevel}
We define the level-one probabilities of returning a pair
of trees $T_1$ and $T_2$ that belong a pair of graphs
$G_1$, $G_2$,
\[
p^{\left(\leq 1\right)} \left(\left(G, G_1, G_2\right),
	\left(T_1, T_2 \right), \hatT \right).
\]
as the product of:
\begin{enumerate}
\item The probability (from running \textsc{AlmostIndependent}) that $G$ is partitioned into $V_1 \sqcup V_2$ so that
$\sc{G}{V_1} = G_1$ and $\sc{G'}{V_2} = G_2$, where $G'$ is $G$ with the edges $T \cap G[V_1]$ contracted and all other edges in $G[V_1]$ are deleted. 
\item The probability that $T_1$ is
mapped to $\hatT[V_1]$ in Line~\ref{ln:ApproxTreeSplit}.
\item The probability that $T_2$ is mapped to
$\hatT / \hatT[V_1]$ by the call to \textsc{ProlongateTree}
on Line~\ref{ln:MapT2Back}. 
\end{enumerate}
\end{definition}

This definition then allows us to formalize the splitting
of probabilities above and below level $1$. 
More importantly, we note that if we instead call \textsc{SchurSparse} to generate $G_1$ and $G_2$, this will not affect the level-one probability because (1) both the calls to \textsc{AlmostIndependent} and \textsc{ProlongateTree} do not depend on $G_1$ and $G_2$, and (2) we can consider $T_1$ to be drawn from the multi-graph of $G_1$ where we track which edges are from the original graph and which were generated by the Schur complement.

Consequently, the only difference between the distributions $\mathcal{T}^{(0)}$ and  $\mathcal{T}^{(1)}$ will be the distortion of drawing $T_1$ and $T_2$ from $G_1$ and $G_2$ vs the sparsified version of $G_1$ and $G_2$.
This handling of sparsifiers of the Schur complements is further
simplified with by the following observation:
\begin{claim}\label{claim:equivToIdealSparsify}
The output of $\textsc{SchurSparse}(G, V', \delta)$ is
identical to the output of
\[
\textsc{IdealSparsify}\left(\sc{G}{V'}, \ttautil, n^2 \delta^{-1}\right),
\]
for some set of $1.1$-approximate statistical leverage
scores of $\sc{G}{V'}$, $\ttautil$.
\end{claim}

This can be seen by revisiting the Schur complement sparsification
and rejection sampling algorithms from Section~\ref{sec:ImplicitSchur}
and~\ref{subsec:RejectionSample} which show that this statement
also extends to the approximate Schur complements produced
on lines~\ref{ln:ApproxTreeG1} and~\ref{ln:ApproxTreeG2}
of Algorithm~\ref{alg:approxTree}.

This means we can let $\mathcal{H}_1$ and $\mathcal{H}_2$
denote the distribution produced by $\textsc{IdealSparsify}$
on $G_1$ and $G_2$ respectively.

\begin{lemma}
\label{lem:Level0Probability}
There exists a collection of graphs and tree pairs
$(\vec{\mathcal{G}}, \vec{\mathcal{T}} )^{\leq 1}$
such that for any tree $\hatT$, with the probabilities
given above in Definition~\ref{def:PAboveOneLevel} we have:
\[
\Pr^{\mathcal{T}^{\left(0\right)}}\left( \hatT \right)
= \sum_{\left( \left(G, G_1, G_2 \right),
	\left( T_1, T_2 \right) \right)
	\in \left( \mathcal{G}, \mathcal{T}\right)^{\left(\leq 1\right)}}
p^{\left(\leq 1\right)} \left( \left(G, G_1, G_2 \right),
	\left(T_1, T_2\right), \hatT\right)
\cdot \Pr^{G_1}\left(T_1\right)
\cdot \Pr^{G_2}\left(T_2\right).
\]
and
\begin{multline*}
\Pr^{\mathcal{T}^{\left(1\right)}}\left( \hatT \right)
= \sum_{\left( \left(G, G_1, G_2 \right),
	\left( T_1, T_2 \right) \right)
	\in \left( \mathcal{G}, \mathcal{T}\right)^{\left(\leq 1\right)}}
p^{\left(\leq 1\right)} \left(\left(G, G_1, G_2 \right),
	\left(T_1, T_2\right), \hatT\right)\\
\cdot \expec{H_1 \in \mathcal{H}_1}{\Pr^{H_1}\left(T_1\right)}
\cdot \expec{H_2 \in \mathcal{H}_2}{\Pr^{G_2}\left(T_2\right)}.
\end{multline*}
\end{lemma}

We can then in turn extend this via induction to multiple levels.
It is important to note that in comparing the distributions
$\mathcal{T}^{(L - 1)}$ and $\mathcal{T}^{(L)}$ for $L \geq 1$ both
will make calls to \textsc{IdealSparsify} through level $L$.
We will then need to additionally consider the possible graphs generated
by sparsification through level $L$,
then restrict to the corresponding exact graphs at level $L + 1$.

\begin{definition}
\label{def:LevelL}
We will use $\vec{\mathcal{G}}^{(\leq L)}, \vec{\mathcal{T}}^{(L)}$
to denote a sequence of graphs on levels up to $L - 1$, plus the peripheral
exact Schur complements on level $L$, along with the spanning trees generated
on these peripheral graphs.

As these graphs and trees can exist on different vertex sets, we will use
$(\vec{\mathcal{G}}, \vec{\mathcal{T}})^{(\leq L)}$ to denote the set of graph/tree
pairs that are on the same set of vertices.
For a sequence of graphs $\vec{G}^{\leq L}$ and a
sequence of trees on their peripherals, $\vec{T}^{L}$, we will use
\[
p^{\left(\leq L\right)} \left(\vec{G}^{\left( \leq L \right)},
	\vec{T}^{\left( L \right)}, \hatT \right)
\]
to denote the product of the probabilities of the level-by-level
vertex split and resulting trees mapping back correctly as defined
in Definition~\ref{def:PAboveOneLevel},
times the probabilities that the subsequent graphs are generated
as sparsifiers of the ones above

Furthermore, we will use $\vec{G}^{(L)}$ to denote just the peripheral
graphs, and $\vec{\mathcal{H}}(\vec{G}^{(L)})$ to denote the
product distribution over sparsifiers of these graphs, and
$\vec{H}^{(L)}$ to denote one particular sequence of such sparsifiers
on this level.
We can also define the probabilities of trees being picked
in a vector-wise sense:
\[
\Pr^{\vec{G}^{\left( L \right)}}\left(\vec{T}^{\left( L \right)}\right)
\defeq \prod_{j} \Pr^{\vec{G}^{\left( L \right)}_j}
	\left(\vec{T}^{\left( L \right)}_j\right),
\qquad 
\Pr^{\vec{H}^{\left(L \right)}}\left(\vec{T}^{\left( L \right)}\right)
\defeq \prod_{j} \Pr^{\vec{H}^{\left( L \right)}_j}
	\left(\vec{T}^{\left( L \right)}_j\right).
\]
\end{definition}

Applying Lemma~\ref{lem:Level0Probability} inductively then allows
us to extend this to multiple levels.

\begin{corollary}
\label{cor:LevelLProbability}
There exists a collection of graphs and tree pairs
$(\vec{\mathcal{G}}, \vec{\mathcal{T}} )^{(\leq L)}$
such that for any tree $\hatT$ we have:
\[
\Pr^{\mathcal{T}^{\left(L - 1\right)}}\left( \hatT \right)
= \sum_{\left( \vec{G}^{\left(\leq L \right)}, \vec{T}^{\left( L \right)} \right)
	\in \left( \mathcal{G}, \mathcal{T}\right)^{\left(\leq L \right)}}
p^{\left(\leq L \right)} \left(
\vec{G}^{\left(\leq L \right)}, \vec{T}^{\left( L \right)}, \hatT\right)
\cdot \Pr^{\vec{G}^{\left( L \right)}}\left(\vec{T}^{\left( L \right)}\right),
\]
and
\[
\Pr^{\mathcal{T}^{\left( L \right)}}\left( \hatT \right)
= \sum_{\left( \vec{G}^{\left(\leq L \right)}, \vec{T}^{\left( L \right)} \right)
	\in \left( \mathcal{G}, \mathcal{T}\right)^{\left(\leq L \right)}}
p^{\left(\leq L \right)} \left(
\vec{G}^{\left(\leq L \right)}, \vec{T}^{\left( L \right)}, \hatT\right)
\cdot
\expec{\vec{H}^{\left( L \right)} \sim \vec{\mathcal{H}}\left(G^{\left( L \right)}\right)}
{\Pr^{\vec{H}^{\left( L \right)}}\left(\vec{T}^{\left( L \right)}\right)}.
\]
\end{corollary}

This reduces our necessary proof of bounding the total variation
distance between ${\mathcal{T}^{(L - 1)}}$ and ${\mathcal{T}^{(L)}}$ to
examining the difference between
\[
\Pr^{\vec{G}^{\left( L \right)}}\left(\vec{T}^{\left( L \right)}\right)
\qquad 
\text{ and}
\qquad 
\expec{\vec{H}^{\left( L \right)} \sim \vec{\mathcal{H}}\left(G^{\left( L \right)}\right)}
{\Pr^{\vec{H}^{\left( L \right)}}\left(\vec{T}^{\left( L \right)}\right)}.
\]

Recalling the definition of $\Pr^{\vec{H}^{(L)}}(\vec{T}^{(L)})$:
we have that the inverse of each probability in the expectation is
\[
\Pr^{\vec{H}_j^{\left( L \right)}}
\left(\vec{T}_j^{\left( L \right)}\right)^{-1}
= \frac{\totaltrees{\vec{H}_j^{\left( L \right)}}}
	{\ww^{\vec{H}_j^{\left( L \right)}}
		\left(\vec{T}_j^{\left( L \right)}\right)},
\]
and we have concentration bounds for the total trees in $\vec{H}_j^{(L)}$.
However, it is critical to note that this probability is $0$
(and cannot be inverted) when $\vec{T}_j^{(L)}$ is not contained
in $\vec{H}_j^{(L)}$ for some $j$.

This necessitates extending our concentration bounds to random graphs where
we condition upon a certain tree remaining in the graph.
This will be done in the following Lemma, proven in Section~\ref{sec:cond_conc},
and we recall that we set $\amount$ such that $\delta = O(\frac{n^2}{\amount})$
in \textsc{SchurSparse}.

\begin{lemma}\label{lem:InvConc}
	Let $G$ be a graph on $n$ vertices and $m$ edges,
	$\ttautil$ be an $1.1$-approximate estimates of leverage scores,
	$\amount$ be a sample count such that $\amount \geq 4n^2$
	and $m \geq \frac{\amount^2}{n}$.
	Let $\mathcal{H}$ denote the distribution over the outputs of
	$\textsc{IdealSparsify}(G, \ttautil, s)$, and for a
	any fixed spanning $\hatT$, let $\mathcal{H}|_{T}$ denote the
	distribution formed by conditioning on the graph containing $\hatT$.
	Then we have:
	\[
	\prob{H \sim \mathcal{H}}{\hatT \subseteq H}^{-1} \cdot 
	\expec{H|_{\hatT} \sim \mathcal{H}_{|_{\hatT}}}{
		\Pr^{H|_{\hatT}}\left(\hatT \right)^{-1}}
	= \left(1 \pm O\left(\frac{n^2}{\amount}\right)\right)
	\Pr^{G}\left(\hatT \right)^{-1},
	\]
	and
	\[
	\prob{H \sim \mathcal{H}}{\hatT \subseteq H}^{-2} \cdot 
	\var{H|_{\hatT} \sim \mathcal{H}_{|_{\hatT}}}{
		\Pr^{H|_{\hatT}}\left(\hatT \right)^{-1}}
	\leq O\left(\frac{n^2}{\amount}\right)
	\Pr^{G}\left(\hatT \right)^{-2}.
	\] 
\end{lemma}

Due to the independence of each call to \textsc{IdealSparsify},
we can apply these concentration bounds across the product
\[
\Pr^{\vec{H}^{\left( L \right)}}
	\left(\vec{T}^{\left( L \right)}\right)
=
\prod_{j} \Pr^{\vec{H}^{\left( L \right)}_j}
	\left(\vec{T}^{\left( L \right)}_j\right)
\]
and use the fact that $\delta$ decreases proportionally to vertex size in our algorithm:

\begin{corollary}
\label{cor:InvConcVec}
For any sequence of peripheral graphs $\vec{G}^{(l)}$,
with associated sparsifier distribution $\mathcal{H}^{S}$,
and any sequence of trees $\vec{T}^{(L)}$ as defined in Definition~\ref{def:LevelL}
such that $\Pr^{\vec{G}^{(L)}}(\vec{T}^{(L)}) > 0$, we have
\begin{multline*}
\prob{\vec{H}^{\left( L \right)} \sim \vec{\mathcal{H}}\left(G^{\left( L \right)}\right)}
{\Pr^{\vec{H}^{\left( L \right)}}\left(\vec{T}^{\left( L \right)} \right) > 0}^{-1}
\cdot 
\expec{\vec{H}^{\left( L \right)} \sim \vec{\mathcal{H}}\left(G^{\left( L \right)}\right)
	\left| \Pr^{\vec{H}^{\left( L \right)}} \left(\vec{T}^{\left( L\right)} \right) > 0\right.}
{\Pr^{\vec{H}^{\left( L \right)}}\left(\vec{T}^{\left( L \right)} \right)^{-1}}\\
= \left(1 \pm \delta \right)
\Pr^{\vec{G}^{\left( L \right)}}\left(\vec{T}^{\left( L \right)} \right)^{-1},
\end{multline*}
and
\begin{multline*}
\prob{\vec{H}^{\left( L \right)} \sim \vec{\mathcal{H}}\left(G^{\left( L \right)}\right)}
{\Pr^{\vec{H}^{\left( L \right)}}\left(\vec{T}^{\left( L \right)} \right) > 0}^{-2}
\cdot 
\expec{\vec{H}^{\left( L \right)} \sim \vec{\mathcal{H}}\left(G^{\left( L \right)}\right)
	\left| \Pr^{\vec{H}^{\left( L \right)}}\left(\vec{T}^{\left( L \right)} \right) > 0\right.}
{\Pr^{\vec{H}^{\left( L \right)}}\left(\vec{T}^{\left( L \right)} \right)^{-2}}\\
\leq \left(1 + \delta \right)
\Pr^{\vec{G}^{\left( L \right)}}\left(\vec{T}^{\left( L \right)} \right)^{-2}.
\end{multline*}
\end{corollary}

\begin{proof}
The independence of the calls to \textsc{IdealSparsify},
and the definition of
\[
	\Pr^{\vec{G}^{ \left( L \right) }}\left(\vec{T}^{ \left( L \right) }\right)
	\defeq \prod_{j} \Pr^{\vec{G}^{ \left( L \right) }_j}
		\left(\vec{T}^{ \left( L \right) }_j\right),
	\qquad 
	\Pr^{\vec{H}^{ \left( L \right) }}\left(\vec{T}^{ \left( L \right) }\right)
	\defeq \prod_{j} \Pr^{\vec{H}^{ \left( L \right) }_j}
		\left(\vec{T}^{ \left( L \right) }_j\right).
\]
Applying Lemma~\ref{lem:InvConc} to each call of \textsc{IdealSparsify},
where $\amount$ was set such that $\delta / \overline{n} = \frac{n^2}{\amount}$ gives
gives that the total error bounded by
\[
\exp\left(\sum_{j}
\frac{\abs{V\left(G^{\left( L \right)} \right)}}{\overline{n}}\right),
\]
and the bound then follows form the total size of each level of the recursion
being $O(\bar{n})$.
\end{proof}

It then remains to use concentration bounds on the inverse of the desired probability to bound the total variation distance, which can be done by the following lemma which can be viewed as an extension of Lemma~\ref{lem:VarianceTV},
and is also proven in Section~\ref{sec:TVBound}.
\begin{lemma}\label{lem:invVarTVBound}
Let $\mathcal{U}$ be a distribution over a universe of elements, $u$,
each associated with random variable $P_u$ such that
\[
\expec{u \sim \mathcal{U}}{\expec{}{P_u}} = 1,
\]
and for each $P_u$ we have
\begin{enumerate}
\item $P_u \geq 0$, and
\item $
\prob{}{P_u > 0 }^{-1} \cdot
\expec{p \sim P_u \left| p > 0 \right.}{p^{-1}}
= 1 \pm \delta,
$
and
\item $
\prob{}{P_u > 0 }^{-2}
\expec{p \sim P_{u} \left| p > 0 \right.}{p^{-2}}
\leq 1 + \delta,
$
\end{enumerate}
then
\[
\expec{u \sim \mathcal{U}}{\abs{1 - \expec{}{P_u}}}
	\leq O\left(\sqrt{\delta}\right).
\]
\end{lemma}

To utilize this lemma, we observe that the values
\[
p^{\left(\leq L\right)}
\left(\vec{G}^{\left( \leq L \right)},
	\vec{T}^{\left( L \right)}, \hatT\right)
\cdot\Pr^{\vec{G}^{\left( L \right)}}
	\left(\vec{T}^{\left( L \right) }\right)
\]
forms a probability distribution over tuples
$\vec{G}^{(\leq L)}, \vec{T}^{(L)}, \hatT$,
while the distribution $\mathcal{H}(\vec{G}^{(L)})$, once
rescaled, can play the role of $P_u$.
Decoupling the total variation distance per tree
into the corresponding terms on pairs of
$\vec{G}^{(\leq L)}, \vec{T}^{(L)}$ then allows us to
bound the overall total variation distance between
$\mathcal{T}^{(L - 1)}$ and $\mathcal{T}^{(L)}$.

\begin{proof}[Proof of Lemma~\ref{lem:DistPerLayer}]
By the definition of total variation distance
\[
d_{TV}\left( \mathcal{T}^{\left(L - 1\right)},
\mathcal{T}^{\left(L\right)} \right)
= \sum_{\hatT} \abs{\Pr^{\mathcal{T}^{\left(L - 1\right)}}\left( \hatT \right)
	 - \Pr^{\mathcal{T}^{\left( L \right)}}\left( \hatT \right)}.
\]
	
By Corollary~\ref{cor:LevelLProbability} and triangle inequality
we can then upper bound this probability by
\begin{multline*}
d_{TV}\left( \mathcal{T}^{\left(L - 1\right)}, \mathcal{T}^{\left(L\right)} \right)
\leq \sum_{\hatT} \sum_{\left( \vec{G}^{\left( \leq L \right)},
		\vec{T}^{\left(L \right)}\right)
			\in \left( \mathcal{G}, \mathcal{T}\right)^{\left( \leq L \right)}}
p^{\left(\leq L\right)} \left(\vec{G}^{\left( \leq L \right)},
	\vec{T}^{\left( L \right)}, \hatT\right)\\ 
\cdot
\abs{\Pr^{\vec{G}^{\left( L \right)}}\left(\vec{T}^{\left( L \right) }\right)
	- \expec{\vec{H}^{\left( L \right)} \sim \vec{\mathcal{H}}\left(G^{\left( L \right)}\right)}
	{\Pr^{\vec{H}^{\left(L\right)}}\left(\vec{T}^{\left(L\right)}\right)}}.
\end{multline*}

The scalar $p^{(\leq L)} (\vec{G}^{( \leq L )},
\vec{T}^{( L )}, \hatT)$
is crucially the same for each, and the inner term in the summation is equivalent to
\[
\abs{p^{\left(\leq L\right)} \left(\vec{G}^{\left( \leq L \right)},
	\vec{T}^{\left( L \right)}, \hatT\right)\cdot\Pr^{\vec{G}^{\left( L \right)}}\left(\vec{T}^{\left( L \right) }\right)
	- p^{\left(\leq L\right)} \left(\vec{G}^{\left( \leq L \right)},
	\vec{T}^{\left( L \right)}, \hatT\right)\cdot\expec{\vec{H}^{\left( L \right)} \sim \vec{\mathcal{H}}\left(G^{\left( L \right)}\right)}
	{\Pr^{\vec{H}^{\left(L\right)}}\left(\vec{T}^{\left(L\right)}\right)}}
\]

Our goal is to use Lemma~\ref{lem:invVarTVBound} where $\mathcal{U}$
here is the distribution over tuples $(\vec{G}^{(L)}, \vec{T}^{(L)},
\hatT)$ with density equaling:
\[ p^{\left(\leq L\right)} \left(\vec{G}^{\left( \leq L \right)},
\vec{T}^{\left( L \right)}, \hatT\right)
\cdot\Pr^{\vec{G}^{\left( L \right)}}
	\left(\vec{T}^{\left( L \right) }\right),
\]
and $P_{u}$ is the distribution over the corresponding value of
$\mathcal{H}(\vec{G}^{(L)})$, with the same density, and
values equaling to:
\[
\Pr^{\vec{G}^{( L )}} \left(\vec{T}^{( L ) }\right)^{-1}
\Pr^{\vec{H}^{( L )}} \left(\vec{T}^{( L ) }\right).
\]
Note that the fact that each $\vec{T}^{L}$ maps back to some tree
$\hat{T}$ imply that $\mathcal{U}$ is a distribution, as well as
$\expec{u \sim \mathcal{U}}{\expec{}{P_u}} = 1$.
A rescaled version of Corollary~\ref{cor:InvConcVec} then gives
the required conditions for Lemma~\ref{lem:invVarTVBound},
which in turn gives the overall bound.

\end{proof}




\begin{proof}[Proof of Theorem~\ref{thm:spanningTreeAlgoIntro}]
The running time follows the same way as the analysis of the determinant
estimation algorithm in the Proof of Theorem~\ref{thm:detApproxIntro}
at the end of Section~\ref{sec:determinant_algo}.
	
For correctness, the total variation distance bound is
implied by appropriately setting $\delta$, and then invoking
the per-layer bound from Lemma~\ref{lem:DistPerLayer}.
Note that factors of $\log{n}$ are absorbed by the $\widetilde{O}$
notation.
	
Finally, note that for simplicity our analysis of total variation
distance does not account for the failure probability of Lemma~\ref{lem:ERDS}.
To account for these, we can simply use the fact that only
$O(n\log{n})$ calls to $\textsc{SchurSparse}$ are made.
Hence, the probability of any call failing is polynomially small,
which can be absorbed into the total variation distance.
\end{proof}

	\section{Conditional Concentration Bounds}
\label{sec:cond_conc}

In this section, we extend our concentration bounds to conditioning on
a certain tree being in the sampled graph,
specifically with the goal of proving Lemma~\ref{lem:InvConc}.
By edge splitting arguments similar to those
in Section~\ref{subsec:GeneralLeverage}, it suffices to analyze
the case where all edges have about the same leverage score.

\begin{lemma}\label{lem:InvConcUniform}
Let $G$ be a graph on $n$ vertices and $m$ edges such that
all edges have statistical leverage scores
$\ttaubar_e \leq \frac{2n}{m}$,
and $\amount$ be a sample count such that $\amount \geq 4n^2$
and $m \geq \frac{\amount^2}{n}$.
Let $H$ be a subgraph containing $\amount$ edges picked at
random without replacement, and let $\mathcal{H}$ denote
this distribution over subgraphs on $\amount$ edges.
Furthermore for any fixed spanning tree, $\hatT$,
let $\mathcal{H}|_{T}$ denote the distribution induced
by those in $\mathcal{H}$ that contain $\hatT$, and use
$H|_{\hatT}$ to denote such a graph, then
\[
\prob{H \sim \mathcal{H}}{\hatT \subseteq H}^{-1}
\cdot \expec{H|_{\hatT} \sim \mathcal{H}_{|_{\hatT}}}{
	\Pr^{H|_{\hatT}}\left(\hatT \right)^{-1}}
= \left(1 \pm O\left(\frac{n^2}{\amount}\right)\right)
\Pr^{G}\left(\hatT \right)^{-1},
\]
and
\[
\prob{H \sim \mathcal{H}}{\hatT \subseteq H}^{-2} \cdot 
\var{H|_{\hatT} \sim \mathcal{H}_{|_{\hatT}}}{
	\Pr^{H|_{\hatT}}\left(\hatT \right)^{-1}}
\leq O\left(\frac{n^2}{\amount}\right)
\Pr^{G}\left(\hatT \right)^{-2}.
\] 
\end{lemma}

Note that the `uniform leverage score' requirement here is not
as strict as the analysis from Lemma~\ref{lem:SecondMomentApprox}.
This is because we're eventually aiming for a bound of
$\amount \approx n^2$ samples.
This also means that constant factor leverage score approximations
suffices for this routine.

The starting point of this proof is the observation that because
we're doing uniform sampling, the only term in
\[
\Pr^{H|_{\hatT}}\left(\hatT\right)
 = \frac{\ww^{H|_{\hatT}}\left( \hatT \right) }{\totaltrees{H|_{\hatT}}}
 = \frac{\ww^{G}\left( \hatT \right) }{\totaltrees{H|_{\hatT}}}
\]
that is dependent on $H|_{\hatT}$ is $\totaltrees{H|_{\hatT}}$.
The proof will then follow by showing concentration of this variable
which will be done similarly to the concentration of $\totaltrees{H}$
that was done in Section~\ref{sec:overview} and~\ref{sec:sparsification}.

The primary difficulty of extending the proof will come from the fact that
trees will have different probabilities of being in the sampled graph depending on how many edges they share with $\hatT$.
Much of this will be dealt with by the assumption that $\amount \geq 4n^2$,
which makes the exponential terms in the probabilities associated with a
tree being in a sampled graph negligible.
Additionally, this assumption implies that for any fixed tree $\hatT$
the expected number of edges it shares with a random tree is close to $0$.
As a result, trees that intersect with $\hatT$ will have negligible contributions, and our analysis can follow similarly to that in Section~\ref{sec:overview} and~\ref{sec:sparsification}.

We further note that due to the larger sample count of $\amount \geq 4n^2$,
the concentration bounds in this section will also hold,
and would in fact be slightly simpler to prove,
if the edges were sampled independently with probability $\amount / m$.
We keep our assumption of sampling $\amount$ edges globally without
replacement though in order to avoid changing our algorithm,
and the analysis will not require much additional work. 

The section will be organized as follows: In Section~\ref{subsec:conditionalExpect} we give upper and lower bounds on the expectation of $\totaltrees{H|_{\hatT}}$. In Section~\ref{subsec:conditionalVariance} we give an upper bound on the variance of $\totaltrees{H|_{\hatT}}$. In Section~\ref{subsec:invConcentrationProof} we combine the bounds from the previous two sections to prove Lemma~\ref{lem:InvConcUniform}.

\subsection{Upper and Lower Bounds on Conditional Expectation}\label{subsec:conditionalExpect}

In order to prove upper and lower bounds on $\expec{H|_{\hatT}}{\totaltrees{H|_{\hatT}}}$,
we will first give several helpful definitions,
corollaries, and lemmas to assist in the proof.
Our examination of $\expec{H|_{\hatT}}{\totaltrees{H|_{\hatT}}}$
will require approximations of $\prob{H|_{\hatT}}{T \subseteq H|_{\hatT}}$,
and, as we are now fixing $n-1$ edges and drawing $\amount - n + 1$ edges
from the remaining $m - n + 1$ edges, each edge will now have probability $\frac{\amount - n + 1}{m - n + 1}$ of being in the sampled graph.
We will denote this probability with
\[
\widehat{p} \defeq \frac{\amount - n + 1}{m - n + 1}.
\]
It will often be easier to exchange $\widehat{p}$ for
\[
p \defeq \frac{\amount}{m},
\]
the probability of a single edge being picked without
the conditioning on $\hatT$.
The errors of doing so is governed by:
\begin{equation}
\label{eq:ProbEstimate}
\left(1 - \frac{n}{\amount}\right) p
= \frac{\amount - n}{m}
\leq \frac{\amount - n +1}{m - n + 1} = \hat{p}
\leq \frac{s}{m} = p.
\end{equation}
We remark that these errors turn out to be acceptable even when
$\widehat{p}$ is raised to the $O(n)$ power.

Furthermore, our assumption of $\amount \geq 4n^2$ implies that
we expect a randomly chosen tree not to intersect with $\hatT$.
This will often implicitly show up in the form of the geometric
series below, for which a bound is immediately implied by our assumption.
\begin{lemma}\label{lem:geoSeries}
If $\amount \geq 4n^2$, then
\[
\sum_{k=1}^{\infty} \left(\frac{2n^2}{\amount}\right)^k
= O\left(\frac{n^2}{\amount}\right).
\]
\end{lemma}

The change in our sampling procedure will alter the formulation of $\prob{H|_{\hatT}}{T \subseteq H|_{\hatT}}$, so we first want to write  $\expec{H|_{\hatT}}{\totaltrees{H|_{\hatT}}}$ in terms of values that we are familiar with while only losing small errors. Additionally, many of the exponential terms in the previous analysis will immediately be absorbed into approximation error by our assumption that $\amount \geq 4n^2$.

\begin{lemma}\label{lem:reduceExpect} 	Let $G$ be a graph on $n$ vertices and $m$ edges and $\amount$ a value such that $m \geq \frac{\amount^2}{n}$,
Fix some tree $\hatT \in G$.
For a random subset of $\amount \geq 4n^2$ edges containing
$\hatT$, $H|_{\hatT} \supseteq \hatT$, we have
\[
\expec{H|_{\hatT}}{\totaltrees{H|_{\hatT}}}
= \left(1 - O\left(\frac{n^2}{\amount}\right)\right)
\sum_{k=0}^{n-1} {p}^{n-1-k}
\sum_{\substack{T :\, \abs{T \cap \widehat{T}} = k}} \ww(T),
\]	
where $p = s / m$ is the probability of each edge
being picked in the sample.
\end{lemma}

\begin{proof} Given that all edges of $\hatT$ are in $H|_{\widehat{T}}$,
the remaining $\amount - n + 1$ edges are chosen uniformly from all
$m - n + 1$ edges not in $\hatT$.
Accordingly, for any tree $T \in G$, the probability $\prob{H|_{\widehat{T}}}{T\subseteq H|_{\widehat{T}}}$
is obtained by dividing the number of subsets of $\amount - n + 1$
edges that contain all edges in $T \setminus \hatT$,
against the number of subsets of $\amount - n + 1$ edges from $m - n + 1$:
\[
\prob{H|_{\widehat{T}}}{T\subseteq H|_{\widehat{T}}}
= {{m - n + 1 - \abs{T \setminus \widehat{T}}}
\choose {\amount - n + 1 -  \abs{T \setminus \widehat{T}}}}
/ {m -n + 1\choose \amount - n + 1}
= \frac{ \left(\amount - n + 1\right)_{\abs{T \setminus \widehat{T}}} }
{\left(m - n + 1\right)_{ \abs{T \setminus \widehat{T}}} }.
\]	

Following the proof Lemma~\ref{lem:subsetSampled}, this reduces to
\[
	\prob{H|_{\widehat{T}}}{T\subseteq H|_{\widehat{T}}}
	=  \widehat{p}^{\abs{T \setminus \widehat{T}}}\exp\left(-\frac{\abs{T \setminus \widehat{T}}^2}{2\amount} - O\left(\frac{n^3}{\amount^2}\right) \right),
\]
which we can further reduce using the assumption of
$\amount \geq 4n^2$ to:
\[
\prob{H|_{\widehat{T}}}{T\subseteq H|_{\widehat{T}}}
=  \left(1 - O\left(\frac{n^2}{\amount}\right)\right)
	\widehat{p}^{\abs{T \setminus \widehat{T}}},
\]
and in turn obtain via linearity of expectation:
\[
\expec{H|_{\hatT}}{\totaltrees{H|_{\hatT}}} = \left(1 - O\left(\frac{n^2}{\amount}\right)\right)\sum_{T} \ww(T) \widehat{p}^{\left|T\setminus\widehat{T}\right|}.
\]
	
We then subdivide the summation based on the amount of edges
in the intersection of $T$ and $\hatT$ and move our $\widehat{p}$
term inside the summation
\[
\expec{H|_{\hatT}}{\totaltrees{H|_{\hatT}}}
= \left(1 - O\left(\frac{n^2}{\amount}\right)\right) \sum_{k=0}^{n-1} \widehat{p}^{n-1-k} \sum_{\substack{T :\, T \cap \widehat{T} = k}} \ww(T).
\]	
Finally, we can use Equation~\ref{eq:ProbEstimate} to
replace $\widehat{p}$ by $p$ because
\[
1 \geq \left( 1 - \frac{n}{\amount} \right)^{n}
\geq  \left( 1 - \frac{2 n^2}{\amount} \right)
\]
where $n^2 \amount < 0.1$.

	
\end{proof}

We will also require a strong lower bound of the expectation.
The following lemma shows that most of the trees do not intersect
with $\hatT$.
Restricting our consideration to such trees will be much
easier to work in obtaining the lower bound on
$\expec{H|_{\hatT}}{\totaltrees{H|_{\hatT}}}$.

\begin{lemma}\label{lem:noIntersection}
Let $G$ be a graph on $n$ vertices and $m$ edges such that
$m \geq 4n^2$ and all edges have statistical leverage scores
$\leq \frac{2n}{m}$.
For any tree $\hatT \in G$. 
\[
\sum_{\substack{T : \, \abs{T \cap \widehat{T}} = 0}} \ww(T)
\geq \left(1 - O\left(\frac{n^2}{\amount}\right)\right)\totaltrees{G}.
\] 
\end{lemma}

\begin{proof}
By definition, we can classify the trees by their intersection
with $\hatT$:
\[
\totaltrees{G}
=
\sum_{k=0}^{n-1}
\sum_{\substack{T : \, \abs{T \cap \widehat{T}} = k}} \ww(T).
\]
Consider each inner summation and further separating into each
possible forest of $\hatT$ with $k$ edges gives:
\[
\sum_{\substack{T :\, \abs{T \cap \widehat{T}} = k}} \ww(T)
= \sum_{\substack{F \subseteq \hatT \\ \abs{F} = k}}
\sum_{\substack{T \\ F = {T \cap \widehat{T}}}} \ww(T)
\leq  \sum_{\substack{F \subseteq \hatT \\ \abs{F} = k}}
\sum_{\substack{T : \,F \subseteq T}} \ww(T).
\]	
		
Invoking Lemma~\ref{lem:SubsetTree} on the inner summation
and the fact that there are ${n - 1 \choose k}$ forests of
$\hatT$ with $k$ edges, gives an upper bound of
\[
\sum_{\substack{T :\, \abs{T \cap \widehat{T}} = k}} \ww(T)
\leq {n - 1 \choose k}\totaltrees{G}\left(\frac{2n}{m}\right)^k
\leq \totaltrees{G}\left(\frac{2n^2}{m}\right)^k.
\]

We will utilize this upper bound for all $k > 0$ and achieve
a lower bound from rearranging our initial summation
\[
\sum_{\substack{T : \, \abs{T \cap \widehat{T}} = 0}} \ww(T)
= \totaltrees{G} - \sum_{k=1}^{n-1}
\sum_{\substack{T : \, \abs{T \cap \widehat{T}} = k}} \ww(T)
\geq \totaltrees{G}\left(1 - \sum_{k=1}^{n-1}
\left(\frac{2n^2}{m}\right)^k \right).
\]
Applying the assumption of $m \geq 4n^2$ and
Lemma~\ref{lem:geoSeries} gives our desired result.	
\end{proof}

With the necessary tools in place, we will now give upper and
lower bounds on the expectation in terms of $\totaltrees{G}p^{n-1}$,
which we note is also a close approximation of $\expec{H}{\totaltrees{H}}$
by our assumption that $\amount \geq 4n^2$.

\begin{lemma}\label{lem:condExpBound}
Let $G$ be a graph on $n$ vertices and $m$ edges such
that all edges have statistical leverage scores $\leq \frac{2n}{m}$,
and let $\amount$ be such that $m \geq \frac{\amount^2}{n}$.
Fix some tree $\hatT \in G$. For a random subset of
$\amount \geq 4n^2$ edges that contain $\hatT$,
$H|_{\hatT} \subseteq \hatT$ we have: 
\[
\expec{H|_{\hatT}}{\totaltrees{H|_{\hatT}}}
= \left(1 \pm O\left(\frac{n^2}{\amount}\right)\right) \totaltrees{G}p^{n-1}.
\]	
\end{lemma}

\begin{proof} We will first prove the upper bound. From Lemma~\ref{lem:reduceExpect} we have
\[
\expec{H|_{\hatT}}{\totaltrees{H|_{\hatT}}}
\leq  \sum_{k=0}^{n-1} {p}^{n-1-k}
\sum_{\substack{T :\, \abs{T \cap \widehat{T}} = k}} \ww\left(T\right),
\]
while a proof similar to Lemma~\ref{lem:noIntersection} gives
\[
\sum_{\substack{T :\, \abs{T \cap \widehat{T}} = k}} \ww(T)
\leq \totaltrees{G}\left(\frac{2n^2}{m}\right)^k.
\]
	
Moving $p^{n-1}$ outside the summation and substituting $\frac{\amount}{m}$
for $p$ gives
\[
\expec{H|_{\hatT}}{\totaltrees{H|_{\hatT}}}
\leq  \totaltrees{G} p^{n-1}\sum_{k=0}^{n-1}
\left(\frac{2n^2}{\amount}\right)^k,
\]
and applying Corollary~\ref{lem:geoSeries} to upper bound the
summation gives
\[
\expec{H|_{\hatT}}{\totaltrees{H|_{\hatT}}}
\leq \left(1 + O\left(\frac{n^2}{\amount}\right)\right)
\totaltrees{G}p^{n-1}.
\]
	
For the lower bound, we again first using Lemma~\ref{lem:reduceExpect}
and then restrict to trees that do not intersect $\hatT$ using
Lemma~\ref{lem:noIntersection}.
Formally we have:
\begin{multline*}
\expec{H|_{\hatT}}{\totaltrees{H|_{\hatT}}}
= \left(1 - O\left(\frac{n^2}{\amount}\right)\right)
\sum_{k=0}^{n-1} {p}^{n-1-k}
\sum_{\substack{T :\, \abs{T \cap \widehat{T}} = k}} \ww\left(T\right) \\
\geq \left(1 - O\left(\frac{n^2}{\amount}\right)\right) {p}^{n-1} \sum_{\substack{T :\, \abs{T \cap \widehat{T}} = 0}} \ww\left(T\right)
\geq \left(1 - O\left(\frac{n^2}{\amount}\right)\right)p^{n-1}\totaltrees{G}.
\end{multline*}
\end{proof}

\subsection{Upper Bound on Conditional Variance}\label{subsec:conditionalVariance}

The bound on variance is by
upper bounding $\expec{H|_{\hatT}}{\totaltrees{H|_{\hatT}}^2}$
in a way similar to Lemma~\ref{lem:SecondMoment}.
Once again, the assumption of $\amount > 4n^2$ means the situation is
simpler because the exponential term is negligible.

As with the proof of Lemma~\ref{lem:SecondMoment}, we will often separate summations of pairs of trees based upon the number of edges in their
intersection, then frequently invoke Lemma~\ref{lem:SubsetTree}.
However there will be more moving pieces in each summation due to
intersections with $\hatT$,
so Lemma~\ref{lem:disjointTrees} proven later in this section,
which is analogous to Lemma~\ref{lem:IntersectionPairs},
will be much more involved. 

\begin{lemma}\label{lem:condVarBound}
Let $G$ be a graph on $n$ vertices and $m$ edges such that all edges
have statistical leverage scores $\leq \frac{2n}{m}$, and $\amount$
a sample count such that $m \geq \frac{\amount^2}{n}$.
For some tree $\hatT \in G$, let $H|_{\hatT}$  denote a random subset
of $\amount$ edges such that $\hatT \subseteq H|_{\hatT}$, then:
\[
\frac{\expec{H|_{\hatT}}{\totaltrees{H|_{\hatT}}^2}}
{\expec{H|_{\hatT}}{\totaltrees{H|_{\hatT}}}^2}
\leq \left(1 + O\left(\frac{n^2}{\amount}\right)\right).
\]	
\end{lemma}

\begin{proof}
By analogous reasoning to the proof in Lemma~\ref{lem:reduceExpect},
for any pair of trees $T_1,T_2 \in G$ we have
\[
\prob{H|_{\widehat{T}}}{T_1,T_2\subseteq H|_{\widehat{T}}}
= {{m - n + 1 - \abs{(T_1 \cup T_2) \setminus \widehat{T}}}
\choose {\amount - n + 1 -  \abs{(T_1 \cup T_2) \setminus \widehat{T}}}}
/ {m -n + 1\choose \amount - n + 1}
= \frac{ \left(\amount - n + 1\right)_{\abs{(T_1 \cup T_2)
	\setminus \widehat{T}}} }
{\left(m - n + 1\right)_{ \abs{(T_1 \cup T_2)
	\setminus \widehat{T}}} }.
\]

As a consequence of Equation~\ref{eq:ProbEstimate}, specifically
the bound $\frac{s - k}{m - k} \leq \frac{s}{m}$ when $k \geq 0$,
we can obtain the upper bound
\[
\prob{H|_{\widehat{T}}}{T_1,T_2\subseteq H|_{\widehat{T}}}
\leq p^{\abs{\left(T_1 \cup T_2\right) \setminus \hatT}},
\]
and in turn summing over all pairs of trees:
\[
\expec{H|_{\hatT}}{\totaltrees{H|_{\hatT}}^2}
\leq \sum_{T_1,T_2 }
\ww\left(T_1\right)\ww\left(T_2\right)
{p}^{\abs{\left(T_1 \cup T_2\right)\setminus\widehat{T}}}.
\]

We note that
$|(T_1 \cup T_2)\setminus\widehat{T}|
= |T_1 \setminus\widehat{T}| + |T_2\setminus\widehat{T}|
 - |(T_1 \cap T_2)\setminus\widehat{T}|$.
Furthermore,
$|T_1 \setminus\widehat{T}| = n- 1 - |T_1 \cap \hatT$,
so we separate the summation as per usual by each possible size of
$|T_1 \cap \hatT|$ and $|T_2 \cap \hatT|$,
and bring the terms outside of the summation that only depend on these values.	
\[
\expec{H|_{\hatT}}{\totaltrees{H|_{\hatT}}^2}
\leq {p}^{2n-2}\sum_{k_1,k_2} {p}^{-k_1-k_2}
\sum_{\substack{T_1,T_2 \\
	\abs{T_1 \cap \widehat{T}} = k_1 \\
	\abs{T_2 \cap \widehat{T}} = k_2}}
\ww\left(T_1\right)\ww\left(T_2\right)
{p}^{-  \left(T_1 \cap T_2\right)\setminus\widehat{T}}.
\]	
In order to deal with the inner most summation we will need
to again separate based on the size of $|(T_1 \cup T_2) \setminus \hatT|$,
and we further note that $|(T_1 \cap T_2) \setminus \hatT|
= |(T_1 \setminus \hatT) \cap (T_2 \setminus \hatT)|$:  
\[
\expec{H|_{\hatT}}{\totaltrees{H|_{\hatT}}^2}
\leq {p}^{2n-2}
\sum_{k_1,k_2} {p}^{-k_1-k_2}
\sum_{k=0}^{n-1} p^{-k}
\sum_{\substack{T_1, T_2 \\
	\abs{T_1 \cap \widehat{T}} = k_1 \\
	\abs{T_2 \cap \widehat{T}} = k_2 \\
	\abs{(T_1 \setminus \widehat{T}) \cap (T_2 \setminus \widehat{T})} = k}}
\ww\left(T_1\right)\ww\left(T_2\right).
\]
The last term is bounded in Lemma \ref{lem:disjointTrees},
which is stated and proven immediately after this.
Incorporating the resulting bound, and grouping the terms
by the summations over $k_1$, $k_2$, and $k$ respectively gives:	\begin{multline*}
\expec{H|_{\hatT}}{\totaltrees{H|_{\hatT}}^2}
\leq {p}^{2n-2}
\sum_{k_1,k_2} {p}^{-k_1-k_2}
\sum_{k=0}^{n-1} p^{-k}
{m \choose k}{n \choose k_1}{n \choose k_2}
\left(\frac{2n}{m}\right)^{2k + k_1 + k_2} \totaltrees{G}^2 \\
= \totaltrees{G}^2{p}^{2n-2}
\left(\sum_{k_1 = 0}^{n-1} {p}^{-k_1} {n \choose k_1}
\left(\frac{2n}{m}\right)^{k_1} \right)
\left(\sum_{k_2 = 0}^{n-1} {p}^{-k_2} {n \choose k_2}
\left(\frac{2n}{m}\right)^{k_2}\right)
\left( \sum_{k=0}^{n-1} p^{-k} {m \choose k}
	\left(\frac{2n}{m}\right)^{2k}\right) .
\end{multline*}
We then plug in $\frac{\amount}{m}$ for $p$ in each summation
and use the very crude upper bound ${a \choose b} \leq a^b$:	
\[
\expec{H|_{\hatT}}{\totaltrees{H|_{\hatT}}^2}
\leq \totaltrees{G}^2{p}^{2n-2}\left(\sum_{k_1 = 0}^{n-1} \left(\frac{2n^2}{\amount}\right)^{k_1} \right)\left(\sum_{k_2 = 0}^{n-1} \left(\frac{2n^2}{\amount}\right)^{k_2}\right)\left( \sum_{k=0}^{n-1} \left(\frac{2n^2}{\amount}\right)^{k}\right).
\]	
Lemma~\ref{lem:geoSeries} then upper bounds each summation
by $1 + O(n^2 / s)$, giving
\[
\expec{H|_{\hatT}}{\totaltrees{H|_{\hatT}}^2}\leq \left(1 + O\left(\frac{n^2}{\amount}\right)\right)\totaltrees{G}^2{p}^{2n-2}.
\]	
\end{proof}

It remains to prove the following bound on the number of of pairs
of trees with a certain intersection size with $\hatT$, and each other.
The following Lemma is a generalization to Lemma~\ref{lem:IntersectionPairs},
and is proven analogously using the negative correlation of edges
in spanning trees from Fact~\ref{fact:negCorrelation}
and Lemma~\ref{lem:SubsetTree}.

\begin{lemma}\label{lem:disjointTrees}	
Let $G$ be  graph with $m$ edges and $n$ vertices such that every edges
has leverage score $ \leq \frac{2n}{m}$.
For any tree $\hatT \in G$ and any integers $k,k_1,k_2 \in [0,n-1]$,
\[
\sum_{\substack{T_1, T_2\\
	\abs{T_1 \cap \widehat{T}} = k_1 \\
	\abs{T_2 \cap \widehat{T}} = k_2 \\
	\abs{(T_1 \setminus \widehat{T}) \cap (T_2 \setminus \widehat{T})} = k}}
\ww\left(T_1\right)\ww\left(T_2\right)
\leq {m \choose k}{n \choose k_1}{n \choose k_2}
\left(\frac{2n}{m}\right)^{2k + k_1 + k_2} \totaltrees{G}^2.
\]	
\end{lemma}

\begin{proof}
We will first separate the summation over all possible forests
$F$ of size $k$ that could be the intersection of
$T_1 \setminus \widehat{T}$ and $T_2 \setminus \widehat{T}$:
\[
\sum_{\substack{T_1, T_2 \\
	\abs{T_1 \cap \widehat{T}} = k_1 \\
	\abs{T_2 \cap \widehat{T}} = k_2 \\
	\abs{(T_1 \setminus \widehat{T}) \cap (T_2 \setminus \widehat{T})} = k}}
\ww\left(T_1\right)\ww\left(T_2\right)
= \sum_{\substack{F \subseteq E \\
	\abs{F} = k}}\sum_{\substack{T_1,T_2 \\
	\abs{T_1 \cap \widehat{T}} = k_1 \\
	\abs{T_2 \cap \widehat{T}} = k_2 \\
	F =(T_1 \setminus \widehat{T}) \cap (T_2 \setminus \widehat{T}) }}
\ww\left(T_1\right)\ww\left(T_2\right).
\]
		
We first consider the inner summation,
and will relax the requirement to only needing
\[
F \subseteq (T_1 \setminus \widehat{T}) \cap (T_2 \setminus \widehat{T}),
\]
which we note is equivalent to
$F \subseteq (T_1 \setminus \widehat{T})$
and $F \subseteq (T_2 \setminus \widehat{T})$.
This then allows us to separate the summation again
for a particular $F$ into terms involving just $T_1$ and $T_2$:
\[
\sum_{\substack{T_1,T_2 \\
	\abs{T_1 \cap \widehat{T}} = k_1 \\
	\abs{T_2 \cap \widehat{T}} = k_2 \\
	F =\left(T_1 \setminus \widehat{T}\right)
		\cap \left(T_2 \setminus \widehat{T}\right) }}
\ww\left(T_1\right)\ww\left(T_2\right)   
\leq
\left(
\sum_{\substack{T_1:\,\abs{T_1 \cap \widehat{T}} = k_1 \\ 
	F \subseteq \left(T_1 \setminus \widehat{T}\right)  }}
\ww\left(T_1\right)
\right)
\left(\sum_{\substack{ T_2:\, \abs{T_2 \cap \widehat{T}}= k_2 \\
	F \subseteq \left(T_2 \setminus \widehat{T}\right) }}
\ww\left(T_2\right)
\right).
\]	
	
We further examine the first term in the product,
and the second will follow equivalently.
Once again, we will split the summation by all possible
forests $\hatF$ of $\hatT$ with size $k_1$ that
$T_1 \setminus \widehat{T}$ could intersect in,
and further relax to them only having to contain $\hatF$. 
\[
\sum_{\substack{T_1 : \,\abs{T_1 \cap \widehat{T}} = {k_1} \\
	F \subseteq \left(T_1 \setminus \widehat{T}\right)}}
\ww\left(T_1\right)
\leq
\sum_{\substack{\widehat{F} \subseteq \widehat{T} \\
	\abs{\widehat{F}} = k_1}} \sum_{\substack{T_1 \\
	\widehat{F} \subseteq (T_1 \cap \widehat{T}) \\
	F \subseteq (T_1 \setminus \widehat{T})}}
\ww\left(T_1\right).
\]

Since $T_1 \cap \widehat{T}$ and $T_1 \setminus \widehat{T}$
are disjoint, we can restrict to $\widehat{F}$ that are 
disjoint from $F$, as well as relaxing to requiring
$(\widehat{F} \cup F) \subseteq T_1$
(instead of $\widehat{F} \subseteq (T_1 \cap \hatT)$
and $F \subseteq (T_1 \setminus \hatT)$):
\[
\sum_{\substack{T_1 : \,\abs{T_1 \cap \widehat{T}} = {k_1} \\
	F \subseteq \left(T_1 \setminus \widehat{T}\right)}}
\ww\left(T_1\right)
\leq
\sum_{\substack{\widehat{F} \subseteq \widehat{T} \\
	\abs{\widehat{F}} = {k_1} \\
	\left(\widehat{F} \cap F\right) = \emptyset}}
\sum_{\left(\widehat{F} \cup F\right) \subseteq T}
\ww\left(T\right).
\]
The assumption of $\widehat{F}$ and $F$ being disjoint means
their union must have exactly $k + k_1$ edges.
We can then apply Lemma~\ref{lem:SubsetTree} to the inner summation
and use the fact that there are at most ${ n - 1 \choose k_1}$
sets $\hatF$ to achieve the upper bound 
\[
\sum_{\substack{T_1 : \,\abs{T_1 \cap \widehat{T}} = {k_1} \\
	F \subseteq \left(T_1 \setminus \widehat{T}\right)}}
\ww\left(T_1\right)
\leq
{n \choose {k_1}} \left( \frac{2n}{m}\right)^{k + {k_1}} \totaltrees{G}.
\]
Similarly, we can also obtain
\[
\sum_{\substack{T_2 : \,\abs{T_2 \cap \widehat{T}} = {k_2} \\
	F \subseteq \left(T_2 \setminus \widehat{T}\right)}}
\ww\left(T_2\right)
\leq {n \choose {k_2}} \left( \frac{2n}{m}\right)^{k + {k_2}}
\totaltrees{G},
\]
which, along with the fact that there are ${m \choose k}$
edge sets $F$ of size $k$, gives our desired bound.	
\end{proof}

\subsection{Concentration of Inverse Probabilities}\label{subsec:invConcentrationProof}

We now complete a proof of Lemma~\ref{lem:InvConcUniform}
using the concentration results on the number of trees in a sampled graph,
conditioned upon a certain tree being contained in the graph.

\begin{proof}[Proof of Lemma~\ref{lem:InvConcUniform}] 
The definition of
\[
\Pr^{H|_{\hatT}}\left(\hatT \right)^{-1} = \frac{\totaltrees{H|_{\hatT}}}{\ww(\hatT)}
\]
and Lemma~\ref{lem:subsetSampled} give
\[
\prob{H \sim \mathcal{H}}{\hatT \subseteq H}^{-1}
\cdot
\expec{H|_{\hatT} \sim \mathcal{H}_{|_{\hatT}}}{
	\Pr^{H|_{\hatT}}\left(\hatT \right)^{-1}}
= \left(\frac{1}{p}\right)^{n-1} 
\exp\left(\frac{n^2}{2\amount} + O\left(\frac{n^3}{\amount^2}\right) \right)
\frac{\expec{H|_{\hatT} \sim \mathcal{H}_{|_{\hatT}}}
	{\totaltrees{H|_{\hatT}}}}{\ww\left(\hatT\right)}.
\]
	
Our condition of $\amount \geq 4n^2$ allows us to bound
the term $\exp(n^2 / (2\amount) + O(n^3 / \amount^2))$
by $(1 + O(n^2 / \amount))$,
and incorporating our approximation of
${\expec{H|_{\hatT} \sim \mathcal{H}_{|_{\hatT}}}{
\totaltrees{H|_{\hatT}}}}$ from Lemma~\ref{lem:condExpBound} gives
\[
\prob{H \sim \mathcal{H}}{\hatT \subseteq H}^{-1}
\cdot
\expec{H|_{\hatT} \sim \mathcal{H}_{|_{\hatT}}}{
\Pr^{H|_{\hatT}}\left(\hatT \right)^{-1}}
= \left(1 \pm O\left(\frac{n^2}{\amount}\right)\right)
\cdot \frac{\totaltrees{G}}{\ww\left(\hatT\right)},
\]
and the definition of $\Pr^{G}\left(\hatT \right)^{-1}$ implies
the bounds on expectation.

For the variance bound, we use the identity
\[ \var{H|_{\hatT} \sim \mathcal{H}_{|_{\hatT}}}{
	\Pr^{H|_{\hatT}}\left(\hatT \right)^{-1}} = \expec{H|_{\hatT} \sim \mathcal{H}_{|_{\hatT}}}{
	\Pr^{H|_{\hatT}}\left(\hatT \right)^{-2}} - \expec{H|_{\hatT} \sim \mathcal{H}_{|_{\hatT}}}{
	\Pr^{H|_{\hatT}}\left(\hatT \right)^{-1}}^2,
\] 
which by the definition
\[
\Pr^{H|_{\hatT}}\left(\hatT \right)^{-1} = \frac{\totaltrees{H|_{\hatT}}}{\ww(\hatT)}
\]
reduces to
\[
\var{H|_{\hatT} \sim \mathcal{H}_{|_{\hatT}}}{
	\Pr^{H|_{\hatT}}\left(\hatT \right)^{-1}}
= \frac{
\expec{H|_{\hatT} \sim \mathcal{H}_{|_{\hatT}}}{\totaltrees{H|_{\hatT}}^2}
- \expec{H|_{\hatT} \sim \mathcal{H}_{|_{\hatT}}} {\totaltrees{H|_{\hatT}}}^2}
{\ww\left(\hatT\right)^2}
\leq O\left(\frac{n^2}{\amount}\right)
\cdot \frac{\totaltrees{G}^2p^{2n-2}}{\ww(\hatT)^2},
\] 
where the last inequality is from incorporating
Lemmas~\ref{lem:condExpBound}~and~\ref{lem:condVarBound}.
Applying Lemma~\ref{lem:subsetSampled}, and once again using 
the condition of $\amount \geq 4n^2$ to bound
\[
\exp\left(\frac{n^2}{2\amount} + O(\frac{n^3}{\amount^2} \right)
\leq \left(1 + O\left(\frac{n^2}{\amount}\right) \right)
\leq O\left(1\right)
\]
gives:
\[
\prob{H \sim \mathcal{H}}{\hatT \subseteq H}^{-2} \cdot 
\var{H|_{\hatT} \sim \mathcal{H}_{|_{\hatT}}}{
	\Pr^{H|_{\hatT}}\left(\hatT \right)^{-1}}
\leq
O\left(\frac{n^2}{\amount}\right) \cdot
\frac{\totaltrees{G}^2}{\ww(\hatT)^2},
\]
and the variance bound follows from
the definition of $\Pr^{G}\left(\hatT \right)^{-1}$.	
\end{proof}

	\section{Bounding Total Variation Distance}
\label{sec:TVBound}
In this section we will first bound the total variation distance between drawing a tree from the $\ww$-uniform distribution of $G$, and uniformly sampling $\amount$ edges, $H$, from $G$, then drawing a tree from the $\ww$-uniform distribution of $H$.
The first bound will only be based on a concentration for the number of trees
in $H$, and will give the $\tilde{O}(n^{13/6})$ time algorithm for sampling
spanning trees from Corollary~\ref{cor:AlgoOneShot}.

Next we will give a more general bound on the total variation distance
between two distributions based on concentration of inverse probabilities.
The resulting Lemma~\ref{lem:invVarTVBound} is used for proving the bound
on total variation distance in the recursive algorithm given in Section~\ref{sec:spanning_tree}.
However, as this bound requires a higher sample count of about $n^2$,
the direct derivation of TV distances from concentration bounds is still
necessary for uses of the $\tilde{O}(n^{1.5})$ edge sparsifier in
Corollary~\ref{cor:AlgoOneShot}.

\subsection{Simple Total Variation Distance Bound from Concentration Bounds}
\label{subsec:EasierTVBound}

We give here a proof of total variation distance being bounded based on the concentration
of spanning trees in the sampled graph.

\begin{proof}(of Lemma~\ref{lem:VarianceTV})
Substituting the definition of $p$ and $\tilde{p}$ into the
definition of total variation distance gives:
\[
d_{TV}\left(p,\tilde{p}\right) 
= \sum_{\hatT}
\abs{ \Pr^{G} \left( \hatT \right)
- \expec{H \sim \mathcal{H}}{\Pr^{H}\left( \hatT \right)}}.
\]
Substituting in the conditions of:
\begin{align*}
\Pr^{H}\left( \hatT \right)
& = \frac{\ww^{H}\left( T \right)}{\totaltrees{H}}, \qquad \text{(by definition of $\Pr^{H}(\hatT)$)}\\
\ww^{H}\left( \hatT \right)
& = \ww^{G} \left( \hatT \right) \cdot
\prob{H' \sim \mathcal{H}}{\hatT \subseteq H'}^{-1}
\cdot \frac{\expec{H' \sim \mathcal{H}}{\totaltrees{H'}}}{\totaltrees{G}},
	\qquad \text{(by given condition)}
\end{align*}
Using the fact that
\[
\expec{H \sim \mathcal{H}}{
	\one \left( \hatT \subseteq H \right)}
= \prob{H' \sim \mathcal{H}}{\hatT \subseteq H'},
\]
we can distribute the first term into:
\[
d_{TV}\left(p,\tilde{p}\right) 
= \sum_{\hat{T}}
\abs{ \expec{H \sim \mathcal{H}}{
	\one \left( \hatT \subseteq H \right) \cdot \prob{H' \sim \mathcal{H}}{\hatT \subseteq H'}^{-1}
	\cdot \Pr^{G} \left(\hatT \right)
	-	\Pr^{H}\left( \hatT \right)}},
\]
which by the condition on $\ww^{H}(\hatT)$ simplifies to:
\[
d_{TV}\left(p,\tilde{p}\right) 
= \sum_{\hat{T}}
\abs{ \expec{H \sim \mathcal{H}}{
	\one \left( \hatT \subseteq H \right) \cdot \frac{\ww^{H} \left( \hatT \right)}
	{\expec{H' \sim \mathcal{H}}{\totaltrees{H'}}}
-	\Pr^{H}\left( \hatT \right)}}.
\]
As $\one ( \hatT \subseteq H ) = 1$
iff $\Pr^{H}(\hatT) > 0$, this further simplifies into
\[
d_{TV}\left(p,\tilde{p}\right) 
= \sum_{\hat{T}}
\prob{H' \sim \mathcal{H'}}{\hatT \subseteq H'}
\abs{\expec{H \sim \mathcal{H}|_{T}}{
	\frac{\ww^{H} \left( \hatT \right)}{\expec{H' \sim \mathcal{H}}{\totaltrees{H'}}}
	-	\Pr^{H}\left( \hatT \right)}},
\]
which by triangle inequality gives:
\[
d_{TV}\left(p,\tilde{p}\right) 
= \sum_{\hat{T}}
\prob{H' \sim \mathcal{H'}}{\hatT \subseteq H'} \cdot
\expec{H \sim \mathcal{H}|_{T}}{
\abs{\frac{\ww^{H} \left( \hatT \right)}{\expec{H' \sim \mathcal{H}}{\totaltrees{H'}}}
		-	\Pr^{H}\left( \hatT \right)}},
\]
at which point we can rearrange the summation to obtain:
\[
d_{TV}\left(p,\tilde{p}\right) 
\leq \expec{H}{\sum_{\hatT \subseteq H}
\abs{\Pr^{H}\left(\hatT \right) 
- \frac{\ww^{H}\left(\hatT \right) }{ \expec{H'}{\totaltrees{H'}}
}}}
= \expec{H}{\sum_{\hatT \subseteq H} \ww^{H}\left( \hatT \right) \cdot 
\abs{\frac{1}{\totaltrees{H}} - \frac{1}{ \expec{H'}{\totaltrees{H'}}}}}.
\]
which by definition of $\totaltrees{H}$ simplifies to:
\[
d_{TV}\left(p,\tilde{p}\right) 
\leq \expec{H}{\abs{1 - \frac{\totaltrees{H} }{ \expec{H'}{\totaltrees{H'} }}}}.
\]
By the Cauchy-Schwarz inequality, which for distributions
can be instantiated as $\expec{X}{f(X)} \leq \sqrt{\expec{X}{f(X)^2}}$
for any random variable $X$ and function $f(X)$, we then get:
\[
d_{TV}\left(p,\tilde{p}\right) 
\leq \sqrt{\expec{H}{\left(1 - \frac{\totaltrees{H} }{ \expec{H'}{\totaltrees{H'} }}\right)^2}}
= \sqrt{\expec{H}{\left(\frac{\totaltrees{H}}{\expec{H'}{\totaltrees{H'} }}\right)^2} - 1}
= \sqrt{\delta}.
\]
\end{proof}

\subsection{Total Variation Distance Bound from Inverse Probability Concentration }\label{subsec:TVinverse}

We give here our proof of Lemma~\ref{lem:invVarTVBound}, that is a more general
bound on total variation distance based upon concentration results of the
inverse probabilities.

\begin{lemma}
\label{lem:chebyshevInv}
	Let $X$ be a random variable such that $X > 0$ over its entire support,
	and given some $\delta \geq 0$, such that $\expec{}{X} = (1 \pm \delta) \mu$ and $\var{}{X} \leq \delta \mu^2$, then \[\prob{}{|X^{-1} - \mu^{-1}| > 4k\sqrt{\delta}\mu^{-1}} \leq \frac{1}{k^2}\] if $1 
	< k < \delta^{-1/2}/4$
\end{lemma}

\begin{proof}
Chebyshev's inequality gives
\[
\prob{}{|X - (1 \pm \delta)\mu| > k \sqrt{\delta}\mu}\leq \frac{1}{k^2}.
\]
	
Furthermore, if we assume $X$ such that
\[
\abs{X - (1 \pm \delta)\mu} \leq k \sqrt{\delta}\mu
\] 
which reduces to
\[
\left(1 - 2k\sqrt{\delta}\right)\mu
\leq X
\leq \left(1 + 2k\sqrt{\delta}\right)\mu.
\]

Inverting and reversing the inequalities gives
\[
\frac{\mu^{-1}}{1 + 2k\sqrt{\delta}} \leq X^{-1} \leq \frac{\mu^{-1}}{1 - 2k\sqrt{\delta}}.
\]

Using the fact that
$\frac{1}{1 + \epsilon} = 1 - \frac{\epsilon}{1 + \epsilon} \leq 1 - \epsilon$
for $\epsilon > 0$, and
$\frac{1}{1 + \epsilon} = 1 + \frac{\epsilon}{1 - \epsilon} \leq 1 + 2\epsilon$
for $\epsilon \leq 1/2$, we can then conclude,
\[
\left(1 - 4k\sqrt{\delta}\right) {\mu^{-1}}\leq X^{-1}
\leq \left(1 + 4k\sqrt{\delta}\right) {\mu^{-1}},
\]
which implies
\[
\prob{}{\abs{X^{-1} - \mu^{-1}} > 4k\sqrt{\delta}\mu^{-1}}
\leq
\prob{}{\abs{X - \left(1 \pm \delta\right)\mu} > k \sqrt{\delta}\mu}
\]
and proves the lemma. 
\end{proof}

This bound does not allow us to bound $\expec{X}{|X - \mu}$ because
when $X$ close to $0$, the value of $X^{-1}$ can be arbitrarily
large, while this bound only bounds the probability of such events
by $O(\delta^{-1})$.
We handle this by treating the case of $X$ small separately,
and account for the total probability of such cases via
summations over $\mathcal{I}$ and $\xhat$.
First we show that once these distributions are truncated
to avoid the small $X$ case, its variance is bounded.

\begin{lemma}
	\label{lem:DistrInv}
	Let $Y$ be a random variable such that for parameters $\delta, \mu_{Y} > 0$ we have
	$0 < Y \leq 2 \mu_{Y}$ over its entire support,
	and that $\expec{}{Y^{-1}} = (1 \pm \delta) \mu_{Y}^{-1}$,
	$\var{}{Y^{-1}} \leq \delta \mu_{Y}^{-2}$, then
	\[
		\expec{}{\abs{Y - \mu_{Y}}} \leq O\left(\sqrt{\delta} \right) \mu_{Y}.
	\]
\end{lemma}

\begin{proof}
Since $\abs{Y - \mu_{Y}} \leq \mu_{Y}$, we can decompose this expected value
into buckets of $2$ via:
\[
\expec{}{\abs{Y - \mu_{Y}}}
\leq
\sum_{i = 0}^{\log \left(\delta^{-1/2} / 4\right)}
\prob{Y}{\abs{Y - \mu_{Y}} \geq 2^{i} \sqrt{\delta} \mu_{Y}}
	\cdot \left(2^{i} \sqrt{\delta} \mu_{Y} \right),
\]
where the last term is from the guarantee of $Y \leq 1$.
Lemma~\ref{lem:chebyshevInv} gives that each of the intermediate
probability terms is bounded by $O(2^{-2i})$, while the last one
is bounded by $\frac{1}{\delta}$, so this gives a total of
\[
\expec{}{\abs{Y - \mu}}
\leq\sum_{i = 0}^{\log \left(\delta^{-1/2}\right)}
\left(2^{i} \sqrt{\delta} \mu_{Y} \right) O\left(2^{-2i} \right)
\leq \sqrt{\delta} \mu_{Y}
\]
\end{proof}
\newcommand{\bad}{\mathit{BAD}}
We can now complete the proof via an argument similar to the proof
of Lemma~\ref{lem:VarianceTV} in Section~\ref{subsec:EasierTVBound}.
The only additional step is the definition of $\bad_u$, which
represents the portion of the random variable $P_u$ with high
deviation.


\begin{proof}[Proof of Lemma~\ref{lem:invVarTVBound}]
	For each $u$, we define a scaling factor corresponding
	to the probability that $P_u$ is non-zero:
	\[
	p_{u+} \defeq \prob{p \sim P_{u}}{p > 0 }.
	\]
	By triangle inequality, we have for each $P_{u}$
	\[
	\abs{1 - \expec{}{P_u}}
	\leq p_{u+} \cdot
	\expec{p \sim P_{u} \left| p > 0 \right.}
	{\abs{p_{u+}^{-1} - p}}.
	\]
	
	We will handle the case where $p$ is close and
	far from $p_{u+}^{-1}$ separately.
	This requires defining the portion of $P_u$
	with non-zero values, but large variance as
	\[
	\bad_{u} \defeq
	\left\{
	p \in \supp\left( P_u\right):
	\abs{p_{u+}^{-1} - p} > \frac{1}{2} p_{u+}^{-1} \right\}.
	\]
	Lemma~\ref{lem:chebyshevInv} gives that for each $u$,
	\[
	\prob{p \sim P_{u} \left| p > 0 \right.} {p \in \bad_u}
	\leq O\left(\sqrt{\delta}\right),
	\]
	which with the outer distribution
	and factoring the value of $p_{u+}^{-1}$ gives
	gives:
	\begin{align}
	\expec{u \sim \mathcal{U}}{
		p_{u+} \cdot
		\expec{p \sim P_{u} \left| p > 0 \right.}{
			{\one\left(p \in \bad_u\right) \cdot p_{u+}^{-1}}}}
	& \leq O\left(\sqrt{\delta}\right),
	\label{eq:pBad}\\
	\expec{u \sim \mathcal{U}}{
		p_{u+} \cdot 
		\expec{p \sim P_{u} \left| p > 0 \right.}
		{\one\left(p \notin \bad_u\right) \cdot p_{u+}^{-1} }}
	& \geq 1 - O\left(\sqrt{\delta}\right).
	\label{eq:pNotBad}
	\end{align}
	
	We then define the `fixed' distributions $\widetilde{P}_u$
	with the same distribution over $p$ as $P_u$, but
	whose values are set to $p_{u+}^{-1}$ whenever $p \in BAD_u$.
	Lemma~\ref{lem:DistrInv} then gives:
	\[
	\expec{p \sim \widetilde{P}_{u} \left| p > 0 \right.}
	{\abs{p_{u+}^{-1} - p}}
	\leq O\left(\sqrt{\delta} p_{u+}^{-1}\right),
	\]
	or taken over the support of $\mathcal{U}$, and written with
	indicator variables:
	\[
	\expec{u \sim \mathcal{U}}{
		p_{u+} \cdot 
		\expec{p \sim P_{u} \left| p > 0 \right.}
		{\one\left(p \notin \bad_u\right) \cdot \abs{p_{u+}^{-1} - p}}}
	\leq O\left(\sqrt{\delta}\right).
	\]
	Combining this with the lower bound on the mass of $p_{u+}^{-1}$
	on the complements of the bad sets from Equation~\ref{eq:pNotBad}
	via the triangle inequality $p \geq \pp_{u+}^{-1}
		-  |\pp_{u+}^{-1} - p|$ gives:
	\[
	\expec{u \sim \mathcal{U}}{
		p_{u+} \cdot 
		\expec{p \sim P_{u} \left| p > 0 \right.}
		{\one\left(p \notin \bad_u\right) \cdot p}}
	\geq 1 - O\left(\sqrt{\delta}\right),
	\]
	or upon taking complement again:
	\[
	\expec{u \sim \mathcal{U}}{
		p_{u+} \cdot 
		\expec{p \sim P_{u} \left| p > 0 \right.}
		{\one\left(p \in \bad_{u}\right) \cdot p}}
	\leq O\left(\sqrt{\delta}\right),
	\]
	which together with Equation~\ref{eq:pBad}
	and the non-negativity of $p_{u+}^{-1}$ and $p$ gives
	\[
	\expec{u \sim \mathcal{U}}{
		p_{u+} \cdot 
		\expec{p \sim P_{u} \left| p > 0 \right.}
		{\one\left(p \in \bad_{u}\right) \cdot \abs{p_{u+}^{-1} - p}}}
	\leq O\left(\sqrt{\delta}\right).
	\]
	Combining these two summations, and invoking the triangle
	inequality at the start then gives the bound.
\end{proof}

	\bibliographystyle{alpha}
	\bibliography{ref1}

\begin{appendix}

\section{Deferred Proofs}
\label{sec:deferred}
We now provide detailed proofs of the combinatorial
facts about random subsets of edges that are discussed
briefly in Section~\ref{sec:overview}.



\begin{proof}
	(of Lemma~\ref{lem:subsetSampled})
	
	This probability is obtained by dividing the number
	of subsets of $\amount$ edges that contain the $n - 1$ edges in $T$,
	against the number of subsets of $\amount$ edges from $m$,
	which using ${a \choose b} = \frac{(a)_{b}}{(b)_{b}}$,
	gives:
	\begin{align}
	{{m - n + 1} \choose {\amount - n + 1}} / {m \choose \amount}
	= \frac{\left(m - n + 1\right)_{\amount - n + 1} \left(\amount\right)_{\amount} }
	{\left(m\right)_{\amount} \left(\amount - n + 1\right)_{\amount - n + 1}},
	\end{align}
	and the two terms can be simplified by the rule
	$(a)_{b} / (a - k)_{b - k} = (a)_k$.
	
	Furthermore,
	
	\[(a)_b = a^b\left(1 - \frac{1}{a}\right) \cdots \left(1 - \frac{b-1}{a}\right) = a^b \exp{\left(\sum_{i=1}^{b-1} \ln{\left(1 - \frac{i}{a}\right)}\right)}\]
	
	We then use the Taylor expansion of $\ln(1-x) = -\sum_{i = 1}^{\infty} \frac{x^i}{i}$ to obtain
	
	\[ = a^b \exp{\left(-\frac{\sum_{i=1}^{b-1} i}{a} -\frac{\sum_{i=1}^{b-1} i^2}{2a^2} - \frac{\sum_{i=1}^{b-1} i^3}{3a^3} - .... \right)} = a^b \exp{\left(-\frac{b^2}{2a} - O\left(\frac{b^3}{a^2}\right)\right)} \]
	
	Substituting into $\frac{(\amount)_{n-1}}{(m)_{n-1}}$ gives
	
	\[ p^{n-1}\exp\left(-\frac{n^2}{2\amount} +\frac{n^2}{2m} - O\left(\frac{n^3}{\amount^2}\right) + O\left(\frac{n^3}{m^2}\right)\right) = p^{n-1}\exp\left(-\frac{n^2}{2\amount} - O\left(\frac{n^3}{\amount^2}\right) \right)\]
	where $\frac{n^2}{2m}$ is absorbed by $O\left(\frac{n^3}{\amount^2}\right)$ because $m \geq \frac{\amount^2}{n}$ was assumed.
	
\end{proof}

\begin{proof}(Of Lemma~\ref{lem:PairwiseProb})
	
As before, we have
\[
\prob{H}{T_1,T_2\in H}
=  p^{\abs{T_1 \cup T_2}}\exp\left(-\frac{\abs{T_1 \cup T_2}^2}{2\amount} - O\left(\frac{n^3}{\amount^2}\right) \right)
\]
Invoking the identity:
\[
\abs{T_1 \cup T_2}  = 2n - 2 - \abs{T_1 \cap T_2}
\]
gives
\[
\prob{H}{T_1,T_2\in H}
=p^{2n-2}  p^{-k}
\exp\left(-\frac{(2n-2 - k)^2}{2\amount}
- O\left(\frac{n^3}{\amount^2}\right) \right).
\]
Using the algebraic identity
\[
(2n - 2 - k)^2 \geq 4n^2 + 4nk
\]
and dropping the trailing (negative) lower order term gives:
\[
\prob{H}{T_1,T_2\in H}
\leq p^{2n-2} \cdot
p^{-k}\exp\left(-\frac{4n^2}{2\amount}
	+ \frac{4n k}{2\amount}\right),
\]
upon which we can pull out the $\frac{4n^2}{2s}$ term in the exponential to
get a term that only depends $k$.
Grouping the $p^{-k}$ term together with the
$\exp(\frac{2n}{s})^k$ term, and using the fact
that $\exp(t) \leq 1 + 2t$ when $t \leq 0.1$ then gives the result.
\end{proof}

\begin{proof}(of Lemma~\ref{lem:IntersectionPairs})
We first separate the summation in terms of all possible forests $F$ of size $k$ that any pair of trees could intersect in

\[
\sum_{\substack{T_1, T_2 \\ \abs{T_1 \cap T_2} = k}}
\ww\left( T_1 \right) \cdot \ww\left( T_2 \right)
=
\sum_{\substack{F \subseteq E \\ |F| = k}}\sum_{\substack{T_1, T_2 \\ F = T_1 \cap T_2}}
\ww\left( T_1 \right) \cdot \ww\left( T_2 \right)
\]

We then consider the inner summation, the number of pairs of trees $T_1,T_2$ with $T_1 \cap T_2 = F$ for some particular set $F$ of size $k$. This is upper bounded by the square of the number of trees containing $F$:

\[\sum_{\substack{T_1, T_2 \\ F = T_1 \cap T_2}}
\ww\left( T_1 \right) \cdot \ww\left( T_2 \right)
\leq
\sum_{\substack{T_1, T_2 \\ F \subseteq T_1 \cap T_2}}
\ww\left( T_1 \right) \cdot \ww\left( T_2 \right)
=
\left( \sum_{T: F \subseteq T} \ww\left( T \right) \right)^2
\] 

This allow us to directly incorporate the bounds from Lemma~\ref{lem:SubsetTree},
and in turn the assumption of
$\ttau_{e} \leq \frac{n}{m}$ to obtain the bound:
\[
\sum_{\substack{T_1, T_2 \\ F = T_1 \cap T_2}}
\ww\left( T_1 \right) \cdot \ww\left( T_2 \right)
\leq \left( \totaltrees{G} \left( \frac{n}{m} \right)^{k} \right)^2.
\]

Furthermore, the number of possible subsets of $F$
is bounded by ${m \choose k}$, which can be bounded
even more crudely by $\frac{m^k}{k!}$.
Incorporating this then gives:
\[
\sum_{\substack{T_1, T_2 \\ \abs{T_1 \cap T_2} = k}}
\ww\left( T_1 \right) \cdot \ww\left( T_2 \right)
\leq \frac{m^k}{k!} \cdot \left( \totaltrees{G} \left( \frac{n}{m} \right)^{k} \right)^2
= \totaltrees{G}^2 \cdot \frac{1}{k!} \left( \frac{n^2}{m} \right)^{k}.
\]
\end{proof}
\end{appendix}

\end{document}